\documentclass[11pt]{article}
\usepackage[margin=1in]{geometry}
\usepackage{amsmath,amssymb}
\usepackage{mathrsfs}
\usepackage{float}
\usepackage{placeins}
\usepackage{hyperref}
\usepackage{graphicx}
\usepackage{bm}
\usepackage{bbm}
\usepackage{color}
\usepackage{enumitem}
\usepackage{array}
\usepackage{empheq}
\usepackage{amsthm}
\usepackage{tikz} \usetikzlibrary{ positioning, calc, arrows.meta }

\newcommand{\doubleangle}[1]{\langle\kern-0.2em\langle #1 \rangle\kern-0.2em\rangle}
\newcommand{\pv}{\mbox{;}\,}

\newcommand{\I}{\text{i}}
\newtheorem{theorem}{Theorem}
\newtheorem{proposition}{Proposition}
\newtheorem{corollary}{Corollary}
\newtheorem{remark}{Remark}
\usepackage[normalem]{ulem}
\usepackage{authblk}  

\begin{document}

\title{Boolean Cumulants and Exact Reduced Descriptions of Renewal-Driven Systems}

\author[1]{Marco Bianucci}
\author[2]{Riccardo Mannella}

\affil[1]{CNR--ISMAR, Lerici (SP), Italy}
\affil[2]{Dipartimento di Fisica, Universit\`a di Pisa, 56100 Pisa, Italy}

\date{\today}
\maketitle
\tableofcontents
\newpage
\begin{abstract}
Reduced descriptions of unresolved fluctuations are commonly based on Gaussian
processes, although many realistic forcings have finite correlation times and a
renewal structure. For memoryless step (Kubo--Anderson) noise, we show that the
reduced dynamics admit two exact and complementary descriptions: a kernel
representation governed by the Boolean cumulants of the jump distribution, and a
frozen-noise representation adapted to stationary probability densities.
For exponentially distributed waiting times, the totally time-ordered
$G$-cumulants coincide with the Boolean cumulants of the jump law, and the memory
kernel is resummed exactly as the Boolean generator $\eta$ evaluated on a
resolvent operator. Second-order closure is exact if and only if the jump law is
symmetric Bernoulli; otherwise, the leading closure error is controlled by
$(b_4/b_2)\lambda^2$. The Boolean hierarchy, however, acts on the kernel, not on
the stationary measure. Stationary densities are approximated by replacing the
jump law with its $N$-point Gauss quadrature: each surrogate preserves all
multi-time correlations up to order $2N-1$ for any waiting-time law, its kernel
is a Pad\'e resummation of the Boolean one, and it is itself an exactly solvable
renewal problem. For the linear system with linear multiplicative interaction
(LIMI/CAM), the frozen-noise representation reduces the dynamics to a random
affine recursion and yields exact support boundaries, singularity exponents and
Kesten tail indices, confirmed numerically. Rates, moments and closure errors are
governed by the Boolean hierarchy, whereas stationary densities are determined by
the geometry of the jump distribution. For memoryless renewal noise, the
combinatorial structure controlling finite-correlation reductions is Boolean.
\end{abstract}


\section{Introduction\label{sec:intro}}

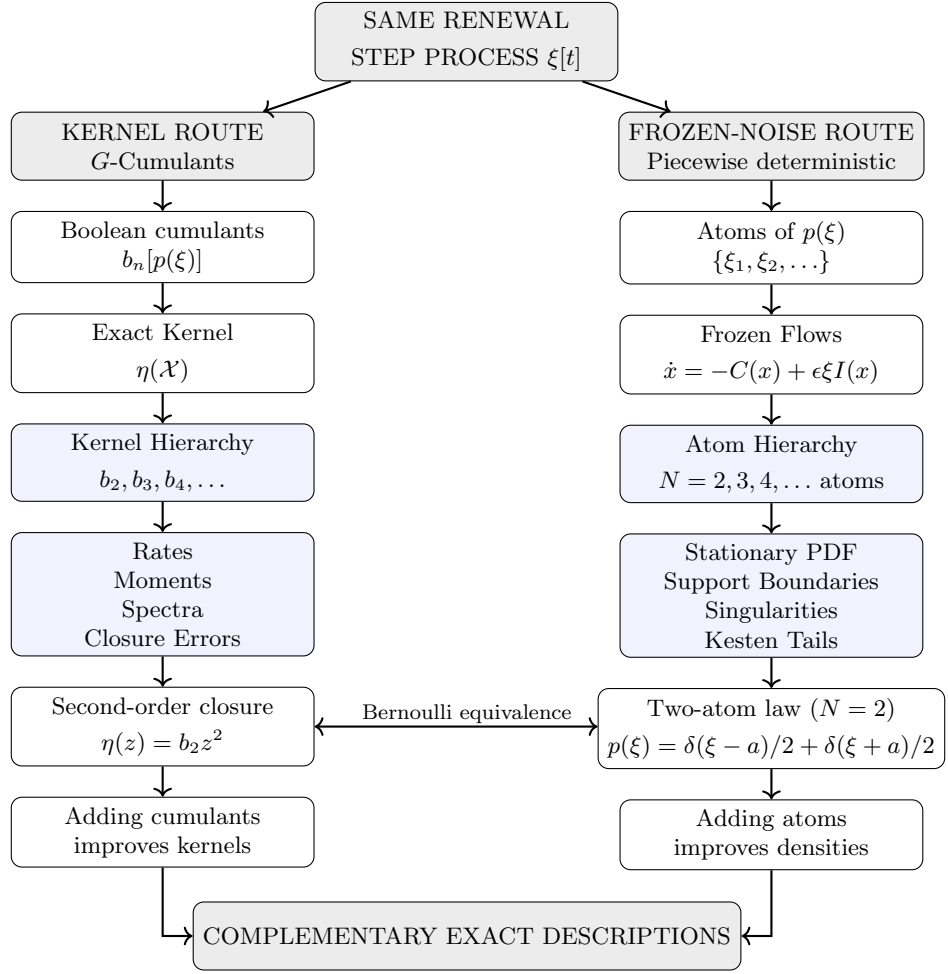
\begin{figure*}[h!]
\centering
\begin{tikzpicture}[
    node distance=0.4cm and 0cm,
    every node/.style={font=\footnotesize},
    box/.style={
        draw,
        rounded corners,
        align=center,
        minimum width=4.cm,
        minimum height=0.9cm
    },
    titlebox/.style={
        draw,
        rounded corners,
        fill=gray!15,
        align=center,
        minimum width=4.cm,
        minimum height=0.9cm
    },
    resultbox/.style={
        draw,
        rounded corners,
        fill=blue!5,
        align=center,
        minimum width=4.cm,
        minimum height=1cm
    },
    arr/.style={->,thick}
]


\node[titlebox] (renewal)
{SAME RENEWAL \\[1mm]STEP PROCESS $\xi[t]$};


\node[titlebox, below left=of renewal]
(kernelroute)
{KERNEL ROUTE\\
$G$-Cumulants};

\node[titlebox, below right=of renewal]
(frozenroute)
{FROZEN-NOISE ROUTE\\
Piecewise deterministic};

\draw[arr] (renewal) -- (kernelroute);
\draw[arr] (renewal) -- (frozenroute);


\node[box, below=of kernelroute]
(cumulants)
{Boolean cumulants\\
$b_n[p(\xi)]$};

\node[box, below=of frozenroute]
(atoms)
{Atoms of $p(\xi)$\\
$\{\xi_1,\xi_2,\ldots\}$};

\draw[arr] (kernelroute) -- (cumulants);
\draw[arr] (frozenroute) -- (atoms);


\node[box, below=of cumulants]
(exactkernel)
{Exact Kernel\\[1mm]
$\eta(\mathcal X)$};

\node[box, below=of atoms]
(frozenflows)
{Frozen Flows\\[1mm]
$\dot x=-C(x)+\epsilon\xi I(x)$};

\draw[arr] (cumulants) -- (exactkernel);
\draw[arr] (atoms) -- (frozenflows);


\node[resultbox, below=of exactkernel]
(kernelhier)
{Kernel Hierarchy\\[1mm]
$b_2,b_3,b_4,\ldots$};

\node[resultbox, below=of frozenflows]
(atomhier)
{Atom Hierarchy\\[1mm]
$N=2,3,4,\ldots$ atoms};

\draw[arr] (exactkernel) -- (kernelhier);
\draw[arr] (frozenflows) -- (atomhier);


\node[resultbox, below=of kernelhier]
(kernelobs)
{Rates\\
Moments\\
Spectra\\
Closure Errors};

\node[resultbox, below=of atomhier]
(densityobs)
{Stationary PDF\\
Support Boundaries\\
Singularities\\
Kesten Tails};

\draw[arr] (kernelhier) -- (kernelobs);
\draw[arr] (atomhier) -- (densityobs);


\node[box, below=of kernelobs]
(secondclosure)
{Second-order closure\\[1mm]
$\eta(z)=b_2z^2$};

\node[box, below=of densityobs]
(twoatom)
{Two-atom law ($N=2$)\\[1mm]
$\displaystyle
p(\xi)=\delta(\xi-a)/2
+\delta(\xi+a)/2
$};

\draw[arr] (kernelobs) -- (secondclosure);
\draw[arr] (densityobs) -- (twoatom);

\draw[<->,thick]
(secondclosure.east) --
(twoatom.west |- secondclosure.east);

\node[font=\scriptsize]
at ($(secondclosure)!0.5!(twoatom)+(0,0.18)$)
{Bernoulli equivalence};


\node[box, below=of secondclosure]
(improvekernel)
{Adding cumulants\\
improves kernels};

\node[box, below=of twoatom]
(improvedensity)
{Adding atoms\\
improves densities};

\draw[arr] (secondclosure) -- (improvekernel);
\draw[arr] (twoatom) -- (improvedensity);


\node[titlebox,
      below=0.9cm of $(improvekernel)!0.5!(improvedensity)$]
(final)
{COMPLEMENTARY EXACT DESCRIPTIONS};

\draw[arr] (improvekernel) |- (final);
\draw[arr] (improvedensity) |- (final);

\end{tikzpicture}

\caption{
Conceptual structure of the present work. Renewal step processes admit two exact and
complementary descriptions. The Boolean route organizes the reduced dynamics in
kernel space through the exact kernel $\eta(\mathcal X)$ and the associated cumulant
hierarchy. The frozen-noise route organizes stationary measures in density
space through frozen flows and the atom hierarchy associated with the jump
distribution.  
}
\label{fig:conceptual_map}
\end{figure*}

Fluctuations in complex systems are routinely represented as stochastic
forcings superimposed on deterministic dynamics, accounting for unresolved
degrees of freedom such as thermal agitation, synaptic inputs, rapid market
fluctuations, or atmospheric forcing. A standard assumption is that such
forcings can be treated as Gaussian processes, an idea traditionally justified
through the central limit theorem (CLT) under the premise that a large number
of independent microscopic fluctuations act on substantially slower
macroscopic variables.

The validity of this picture relies on a strong separation of time scales.
To illustrate the basic setting, consider the stochastic differential equation
\begin{equation}
\label{SDE}
\dot{x} = -C(x) + \epsilon\,I(x)\,\xi[t],
\end{equation}
where $C(x)$ is a generally nonlinear drift, $\xi[t]$ is a stochastic process
with unit variance and correlation time $\tau$, and $I(x)$ is a
state-dependent coupling. Throughout this work, square brackets denote
time-dependent stochastic processes, whereas parentheses refer to particular
realizations. The parameter $\epsilon$ controls the strength of the forcing.

In the singular limit of infinite scale separation, where the correlation time
$\tau$ is much shorter than the characteristic relaxation time of the
deterministic dynamics, the latter can be regarded as locally frozen. A local
form of the CLT then applies, and Eq.~\eqref{SDE} reduces to the familiar
Fokker--Planck equation
\begin{align}
\partial_t P(x\pv t)
\approx
\partial_x\!\bigl[C(x)P\bigr]
+
\epsilon^2\tau\,
\partial_x\!\bigl[I(x)\,\partial_x(I(x)P)\bigr].
\label{ME_secOrder_FPE}
\end{align}

Maintaining a finite diffusive contribution in this limit requires
$\epsilon^2\tau$ to remain constant, effectively replacing the original
forcing by a Gaussian white-noise process.

Many physical, biological, and economic systems, however, operate far from
this asymptotic regime. Correlation times are finite, the statistics of the
unresolved variables are often unknown, and the forcing is frequently more
naturally described as a renewal process consisting of random laminar epochs
rather than as the coarse-grained average of an underlying Gaussian field.
Adopting a Gaussian description under these conditions amounts to introducing
an additional unresolved microscopic layer whose collective action generates
the observed colored noise, as in the Ornstein--Uhlenbeck construction.

The widespread use of Gaussian forcing stems largely from its analytical
tractability~\cite{fFPReport1978}. This tractability is exact only for linear
dynamics. For $C(x)=\gamma x$ and $I(x)=1$, Gaussianity is transferred
linearly from the forcing to the state variable, and the Fokker--Planck
equation provides a complete description. Under nonlinear drift or
state-dependent coupling, this property is lost. Standard projection and
adiabatic-elimination techniques, including the Mori--Zwanzig
formalism~\cite{Zwanzig2001,grigolini1989,chkPNAS}, then yield, for
$t\gg\tau$, an effective evolution equation of the form
\begin{align}
\partial_t P(x\pv t)
=
\partial_x\!\bigl[C(x)P\bigr]
+
\epsilon^2\,
\partial_x\!\bigl[I(x)\,\partial_x(H(x)P)\bigr],
\label{FPE_GEN}
\end{align}
where $H=H_{\rm BFPE}$ at leading order in the perturbation
strength~\cite{bbmJSP191}, or $H=H_{\rm LLA}$ within the local-linear
approximation~\cite{bbmJSP191,tgPRA38}.

The derivation of Eq.~\eqref{FPE_GEN} relies on two key assumptions. First,
the integrated $n$-time correlations of the forcing must scale as $\tau^n$,
with $\tau$ significantly shorter than the deterministic timescale. Second,
either the coupling must remain perturbatively weak ($\epsilon\ll1$), or the
forcing must be sufficiently Gaussian for the dynamics to explore only
locally linear regions of phase space. Under the latter assumption, errors are
ascribed entirely to dynamical nonlinearities, while the statistical
hierarchy is regarded as exact.

This naturally raises a more fundamental question. If the forcing is generated
by a finite-correlation renewal process rather than by an underlying Gaussian
field, which statistical hierarchy governs the reduced dynamics? More
generally, is Gaussianity the appropriate organizing principle for
finite-correlation stochastic reductions, or does a different algebraic
structure emerge naturally from the renewal process itself?

In the present work we address these questions by keeping the dynamics exact
and focusing instead on the statistical structure of the forcing. Using the
generalized-cumulant formalism~\cite{bbJSTAT4,bCSF148,bCSF159} together with
the exact multi-time correlations of renewal processes~\cite{bblmCSF196,
bblmCSF202}, we show that the natural hierarchy governing memoryless
step-renewal noise is not the hierarchy of classical cumulants but the
Boolean hierarchy of non-commutative probability.

The generalized-cumulant construction, originally introduced by
Kubo~\cite{kuboGenCumJPSJ17,kuboGenCumJMP4} and subsequently developed by
Freed and collaborators~\cite{fJCP49,ydfJCP62}, replaces the logarithm of the
characteristic function by a projection operator acting on time-ordered
expansions. Although the formalism was historically debated because of issues
related to time ordering~\cite{fJMP17,fJMP20}, a rigorous formulation was
established in Ref.~\cite{bbJSTAT4}. Building upon that framework, we show
that the totally time-ordered ($G$-)cumulants possess an exact combinatorial
interpretation and coincide with the Boolean cumulants associated with
interval partitions.

A second result emerges when the renewal process is viewed from a different
perspective. Besides the kernel description naturally produced by the
generalized-cumulant formalism, renewal-driven systems admit a complementary
frozen-noise description in which the dynamics is represented as a succession
of deterministic epochs. The two approaches are exact, but they organize
different aspects of the same stochastic process.

The Boolean formulation naturally acts in kernel space. It is organized by the
Boolean cumulants of the jump distribution and by the corresponding exact
kernel, providing direct access to relaxation rates, low-order moments,
spectral observables, and closure criteria. The frozen-noise formulation acts
instead in density space. There the fundamental objects are the frozen flows
generated by the atoms of the jump distribution, from which support
boundaries, singularity structures, and tail properties of the stationary
measure can be inferred.

The relationship between these two descriptions is summarized schematically in
Fig.~\ref{fig:conceptual_map}. The left branch follows the hierarchy of
Boolean cumulants and the associated kernel expansion, while the right branch
follows the hierarchy of atoms underlying the jump distribution and the
corresponding frozen flows. The two routes are complementary: adding
cumulants systematically improves the kernel description, whereas adding atoms
systematically improves the approximation of stationary densities.

A central message of the present work is that these two hierarchies should not
be conflated. Kernel truncations control the accuracy of the reduced dynamics,
whereas stationary densities possess a distinct geometric structure governed
directly by the jump distribution and by the family of frozen deterministic
flows. Consequently, excellent agreement for rates and moments does not
necessarily imply an equally accurate description of the stationary measure.
The two hierarchies are nevertheless related: at the same order they are a
polynomial truncation and a Pad\'e resummation of the same Boolean generator,
and only the latter corresponds to a genuine renewal process
(Section~\ref{sec:atom_expansion}).

\medskip

\noindent
The main results can be summarized as follows:

\begin{description}

\item[\emph{(i) Exact Boolean Structure of Renewal Noise.}]
For memoryless step-renewal processes, the totally time-ordered
($G$-)cumulants coincide exactly with the Boolean cumulants of the jump
distribution. For exponentially distributed waiting times,
\[
\doubleangle{\xi(u_1)\cdots\xi(u_n)}^{(G)}
=
b_n[p(\xi)]\,\phi(u_n-u_1),
\]
at arbitrary time separations. This establishes an explicit bridge between
renewal processes and Boolean probability theory.

\item[\emph{(ii) Exact Kernel Resummation.}]
The Boolean identification allows the generalized-cumulant hierarchy to be
resummed exactly. The resulting memory kernel is expressed through the Boolean
generator $\eta$ evaluated on a resolvent operator and applies to arbitrary
drift and state-dependent coupling.
An independent derivation from the frozen-noise representation extends the
exact kernel to arbitrary waiting-time densities.
\item[\emph{(iii) Exact Closure and Controlled Error Estimates.}]
The hierarchy terminates at second order if and only if the jump distribution
is symmetric Bernoulli. Telegraph noise is therefore exactly solvable because,
within the Boolean kernel hierarchy, the Bernoulli distribution plays the same
algebraic role that the Gaussian distribution plays in the hierarchy of
classical cumulants. For general jump laws, the leading closure error scales
as $(b_4/b_2)\lambda^2$, cleanly separating statistical and dynamical
contributions.

\item[\emph{(iv) Kernel and Density Hierarchies.}]
Kernel localization and truncation govern rates and low-order moments but do
not control stationary probability densities, which are organized instead by
the atoms of the jump distribution. Replacing the jump law with its $N$-point
Gauss quadrature preserves all multi-time correlations up to order $2N-1$ for
any waiting-time law; for Poissonian renewal, the resulting kernel is a Pad\'e
resummation of the Boolean one. The two hierarchies are thus two truncations of
the same object, and only the atomic one yields genuine stationary densities.

\item[\emph{(v) Exact Stationary Statistics.}]
A complementary frozen-noise representation provides direct access to
stationary measures. For the LIMI/CAM model, the dynamics reduce to a random
affine recursion, yielding exact expressions for support boundaries,
singularity exponents, and Kesten tail indices.

\item[\emph{(vi) Non-Perturbative Applicability and Numerical Validation.}]
Because the theory is organized through the Boolean generator rather than a
moment hierarchy, it extends naturally to jump distributions with divergent
moments. Analytical predictions for closure errors, support boundaries,
singular exponents, and Kesten tails are confirmed by extensive numerical
simulations.

\end{description}

Finally, it is important to distinguish step-renewal processes from spike or
shot-noise processes. Although both originate from renewal statistics, they
lead to fundamentally different reduced descriptions. Poissonian spike noise
generates a time-local master equation governed by the ordinary moments of the
jump distribution~\cite{bbmJSTAT2026}, whereas step-renewal processes retain
memory through finite residence times. It is precisely this persistence that
produces a non-local memory kernel and gives rise to the Boolean structure
identified here.

The paper is organized as follows. Sections~\ref{sec:model}--
\ref{sec:nonmarkov} establish the algebraic foundations of the theory,
derive the exact Boolean identification for renewal noise, obtain the
resummed memory kernel, characterize exact closure, and discuss the role of
memorylessness and its non-Markovian extensions.
Sections~\ref{sec:scope} and~\ref{sec:scales} analyze the scope of the theory
and derive quantitative error estimates for finite-order truncations.
Section~\ref{sec:two_routes} introduces the distinction between the kernel and
frozen-noise descriptions and clarifies their complementary domains of
applicability, and Section~\ref{sec:unbounded_expansion} develops the atom
hierarchy, which approximates stationary densities directly in density space
and is related to the kernel hierarchy by a Pad\'e resummation. Sections~\ref{sec:CAM} and~\ref{sec:numres} develop the LIMI/CAM
example, derive exact results for stationary densities, support boundaries,
singularities, and Kesten tails, and validate these predictions numerically.
Finally, Section~\ref{sec:beyond_cam} discusses the extent to which the
resulting picture extends beyond linear drift and coupling before concluding in
Section~\ref{sec:conclusions}.
\section{The model and the cumulant framework\label{sec:model}}
We begin by recalling the generalized-cumulant formalism in a form suitable for renewal-driven
 systems. 
 At this stage, no assumption is made on the statistics of the forcing beyond the existence
  of the relevant correlation functions. The role of this section is twofold. First, it introduces
  the time-ordered cumulant hierarchy governing the reduced dynamics. 
  Second, it identifies the combinatorial structure underlying the corresponding 
  $G$-cumulants. 
  In Section~\ref{sec:renewal}, the specialization to renewal noise will select interval 
  partitions and lead to their exact identification with Boolean cumulants.

\subsection{Stochastic Liouville equation and $G$-cumulants}

For a given realization of the noise $\xi(t)$, the conditional probability density $P(x,\xi(t)\pv t)$ satisfies the stochastic Liouville equation associated with Eq.~\eqref{SDE}:
\begin{equation}
\label{stochLiouv}
\partial_t P(x,\xi(t)\pv t)=\big[\mathcal L_a+\epsilon\,\xi(t)\mathcal L_I\big]
P(x,\xi(t)\pv t),
\qquad
\mathcal L_a:=\partial_xC(x),\quad \mathcal L_I:=\partial_xI(x).
\end{equation}
Introducing the interaction representation, $\tilde P=e^{-\mathcal L_a t}P$, one obtains $\partial_t\tilde P=\epsilon\,\xi(t)\tilde{\mathcal L}_I(t)\tilde P$, where
\begin{equation}
\label{LIt}
\tilde{\mathcal L}_I(t):=e^{-\mathcal L_at}\mathcal L_Ie^{\mathcal L_at}
=e^{-\mathcal L_a^\times t}[\mathcal L_I]
\end{equation}
represents the adjoint Lie evolution of $\mathcal L_I$ generated by $\mathcal L_a$~\cite{bJMP59}. Averaging over the noise realizations yields the reduced density in the form of a generalized characteristic function of the stochastic operator
 $\mathscr S(t):=\int_0^t\Omega(u)\mathrm du$, with $\Omega(u):=\xi(u)\tilde{\mathcal L}_I(u)$:
\begin{equation}
\label{pExp}
\tilde P(x\pv t)=\Big\langle\exp_O\Big[\epsilon\int_0^t\mathrm du\,
\Omega(u)\Big]\Big\rangle P(x\pv 0)
:=\Big\langle\exp_O\Big[\epsilon\int_0^t\mathrm du\,
\tilde{\mathcal L}_I(u)\xi(u)\Big]\Big\rangle P(x\pv 0),
\end{equation}
where $\exp_O$ denotes the exponential ordered under the partial time-ordering (PTO) map $O$
(i.e., the standard time-ordered exponential). 
Following Ref.~\cite{bbJSTAT4}, generalized ($M$-)cumulants are defined by expressing the characteristic function as a generalized exponential of a cumulant generator:
\begin{align}
\label{pExp_2}
\tilde P(x\pv t)=\exp_{M_O}\Big[
\sum_{n\ge1}\epsilon^n\!\!\int_0^t\!\!\mathrm du_{n}\!\!\int_0^{u_{n}}\!\!\!
\mathrm du_{n-1}\cdots\!\!\int_0^{u_2}\!\!\mathrm du_1\;
\doubleangle{\Omega(u_1)\Omega(u_2)\cdots\Omega(u_n)}
\Big]P(x\pv 0).
\end{align}
Among the admissible ordering maps $M_O$, two play a distinguished role. Choosing PTO ($M_O=O$) yields a time-local master equation, whereas total time-ordering (TTO, $M_O=G$) yields a master equation governed by a memory kernel~\cite[\S4.4.3]{bbJSTAT4}\footnote{While PTO and TTO coincide for single-time arguments as in Eq.~\eqref{pExp}, TTO acts non-trivially on multi-time objects, like the argument of the exponential in Eq.~\eqref{pExp_2}.}. Under $G$-ordering, the reduced density obeys
\begin{equation}
\label{MEG}
\partial_t\tilde P(x\pv t)=\int_0^t\mathrm du\;
G\big(-\I\epsilon\tilde{\mathcal L}_I(\cdot)\pv t,u\big)\,\tilde P(x\pv u),
\end{equation}
with the Green function expanded in the corresponding $G$-cumulants:
\begin{align}
\label{GCum_x_L}
G\big(-\I\epsilon\tilde{\mathcal L}_I(\cdot)\pv t,u\big)
&=\sum_{n\ge1}\epsilon^n\!\!\int_u^t\!\!\mathrm du_{n-1}\!\!\int_u^{u_{n-1}}\!\!\!
\mathrm du_{n-2}\cdots\!\!\int_u^{u_3}\!\!\mathrm du_2\;
\tilde{\mathcal L}_I(t)\tilde{\mathcal L}_I(u_{n-1})\cdots\tilde{\mathcal L}_I(u)
\nonumber\\
&\hspace{4.2cm}\times
\doubleangle{\xi(u)\xi(u_2)\cdots\xi(t)}^{(G)} .
\end{align}
Here, $G$ designates both the total-ordering map and the generated kernel, reflecting the fact that total time-ordering generates the memory kernel. Returning to the Schr\"odinger picture, Eq.~\eqref{MEG} becomes
\begin{equation}
\label{MEG_}
\partial_t P(x\pv t)={\mathcal L}_aP(x\pv t)+
\int_0^t\mathrm du\;
\mathcal G\big( t,u\big)\, P(x\pv u),
\end{equation}
where
\begin{equation}
\label{G_}
\mathcal G(t,u):=e^{\mathcal L_at}G\big(-\I\epsilon\tilde{\mathcal L}_I(\cdot)\pv t,u\big)e^{-\mathcal L_au}.
\end{equation}

\subsection{Properties of total time ordering}

Two key properties distinguish the TTO map from other admissible choices.

First, $G$-ordering exhibits an \emph{inheritance property}: every cumulant of the composite operator $\Omega(u)=\xi(u)\tilde{\mathcal L}_I(u)$ factorizes into an ordered product of Liouvillians multiplied by the scalar noise cumulant~\cite{bCSF159},
\begin{align}
\label{G-cumulant_Omega}
\doubleangle{\Omega(u_n)\cdots\Omega(u_1)}^{(G)}
&=\tilde{\mathcal L}_I(u_n)\cdots\tilde{\mathcal L}_I(u_1)\;
\doubleangle{\xi(u_1)\cdots\xi(u_n)}^{(G)} .
\end{align}
Consequently, any vanishing-cumulant condition satisfied by the scalar noise is inherited directly by the operator. Under the PTO map, by contrast, non-commuting factors mix cumulants of different orders, preventing higher-order operator cumulants from vanishing even when scalar cumulants do.

Second, $G$-cumulants enter directly into the memory kernel appearing in Eq.~\eqref{GCum_x_L}. Consequently, the entire problem of deriving the reduced dynamics is reduced to determining the scalar $G$-cumulants of the driving noise.

\subsection{Combinatorial identity of $G$-cumulants}

The moment--cumulant relations governing the $G$-map possess a clear combinatorial interpretation. Denoting by $\mathbb P:=\rangle\langle$ the projection operator that factorizes averages, the inversion formula from Ref.~\cite[Eq.~(95)]{bbJSTAT4} takes the form:
\begin{align}
\label{k_nVSm_n}
\doubleangle{\xi(u_1)\xi(u_2)\cdots\xi(u_n)}^{(G)}&=
\langle \xi(u_1)(1-\mathbb P)\xi(u_2)(1-\mathbb P)\cdots(1-\mathbb P)\xi(u_n)\rangle
\nonumber \\
&=\sum_{\pi\in\mathcal I(n)}(-1)^{|\pi|-1}\ \prod_{B\in\pi}
\langle\xi(u_{b_1})\cdots\xi(u_{b_{|B|}})\rangle.
\end{align}
Here, the M\"obius function simplifies to an alternating sign $(-1)^{|\pi|-1}$. The sum ranges over the $2^{n-1}$ compositions of the ordered sequence $u_1\le\cdots\le u_n$ into consecutive blocks, where $|\pi|$ denotes the number of blocks in partition $\pi$. The inverse relation expresses moments in terms of $G$-cumulants:
\begin{equation}
\label{m_vs_G}
\langle\xi(u_1)\cdots\xi(u_n)\rangle
=\sum_{\pi\in\mathcal I(n)}\ \prod_{B\in\pi}
\doubleangle{\xi(u_{b_1})\cdots\xi(u_{b_{|B|}})}^{(G)},
\end{equation}
where the sum runs over the lattice of interval partitions $\mathcal I(n)$ (Appendix~\ref{app:boolean}).

Equations~\eqref{k_nVSm_n} and~\eqref{m_vs_G} are precisely the moment--cumulant relations defining Boolean cumulants through M\"obius inversion on the lattice of interval partitions $\mathcal I(n)$~\cite{SpeicherWoroudi1997,NicaSpeicher2006}. Thus, the generalized cumulants introduced in the context of time-ordered stochastic dynamics possess an intrinsically Boolean combinatorial structure. At this stage, the correspondence is purely algebraic and makes no reference to a specific stochastic process. The crucial observation developed in the next section is that renewal noise naturally generates interval partitions as its underlying correlation structure. As a consequence, the $G$-cumulants of renewal processes become literally the Boolean cumulants of the jump distribution.
Their place within the broader classification of independence notions~\cite{Muraki2003} is investigated in a companion work.

\section{Renewal Step Noise: the Exact Boolean Identification
\label{sec:renewal}}

\subsection{Process Definition and Interval-Partition Support
\label{definition_and_partition}}

We consider the step (Kubo--Anderson) renewal process $\xi[t]$, in which the
noise remains constant over random laminar intervals and is resampled at
renewal events. During each epoch, the value of $\xi$ is drawn independently
from a jump distribution $p(\xi)$, while the epoch duration is drawn from a
waiting-time density $\psi(\theta)$. Throughout this work, we assume
$\overline{\xi}=0$ and, unless stated otherwise, unit variance $\overline{\xi^2}=1$, while leaving the
higher moments $\overline{\xi^n}$ of $p(\xi)$ arbitrary.

For exponential waiting times,
$\psi(\theta)=e^{-\theta/\tau}/\tau$,
the process is Markovian, stationary, and free of aging. Its two-time
correlation function takes the simple form

\[
\langle\xi(u_1)\xi(u_2)\rangle
=
e^{-|u_2-u_1|/\tau}.
\]

The resulting process constitutes a standard model of finite-correlation
forcing, extending the familiar two-state telegraph process to arbitrary
discrete or continuous jump distributions.

Although non-exponential waiting times lead to aging effects and to memory
kernels that are no longer of convolution type (see
Section~\ref{sec:nonmarkov}), the underlying combinatorial structure is much
more general and does not depend on the choice of $\psi(\theta)$.

Consider a collection of ordered observation times
$u_1\le\cdots\le u_n$. We define a \emph{block} as a subset of observation
times belonging to the same renewal epoch and therefore sharing the same noise
realization. Because the process is piecewise constant, if two times $u_i$ and
$u_k$ belong to the same epoch, then every intermediate time $u_j$ satisfying
$i<j<k$ must belong to that epoch as well. Consequently, the admissible blocks
are necessarily contiguous and form interval partitions of the ordered index
set.

Configurations involving crossing or nested partitions therefore have
identically zero probability. The correlation functions of step-renewal noise
are supported exclusively on interval partitions, independently of the
particular waiting-time distribution $\psi(\theta)$.

This observation is the key structural ingredient underlying the general
results of Ref.~\cite{bblmCSF202}. In the present context, it provides the
crucial link between the combinatorics of renewal processes and the Boolean
cumulant structure identified in Section~\ref{sec:model}. Indeed, once the
support of the correlation functions is restricted to interval partitions, the
Boolean nature of the resulting cumulant hierarchy becomes a direct
consequence of the renewal structure itself.

\subsection{Exact evaluation of the generalized characteristic function\label{sec:frozen}}

The Boolean-cumulant description developed below is not the only exact route to the reduced dynamics. For step-renewal noise, the piecewise-constant nature of the forcing allows the generalized characteristic function  in Eq.~\eqref{pExp} to be evaluated directly, yielding an exact pathwise representation that will play a central role in the analysis of stationary statistics.

 Within a laminar interval where $\xi[t]$ remains constant, the generator in the Schr\"odinger picture reduces to the time-independent operator
\begin{equation}
\mathcal L_\xi:=\mathcal L_a+\epsilon\,\xi\,\mathcal L_I ,
\label{eq:frozen_generator}
\end{equation}
rendering the ordered exponential equivalent to a standard operator exponential.

Let $t_0=0<t_1<\cdots<t_N<t$ denote, respectively, the initial time and the renewal instants in a given realization, $\xi_k$ the value held during $[t_{k-1},t_k]$, and $\theta_k:=t_k-t_{k-1}$ the duration of the $k$-th laminar interval, with $\theta_{N+1}:=t-t_N$ representing the residual duration. The full propagator factorizes into a product of static-noise propagators:
\begin{equation}
\exp_O\Big[\int_0^t\!\mathrm du\,
\big(\mathcal L_a+\epsilon\,\xi(u)\,\mathcal L_I\big)\Big]
=\prod_{k}^{\longleftarrow}\exp\big[\mathcal L_{\xi_k}\,\theta_k\big],
\label{eq:frozen_factorization}
\end{equation}
where the arrow indicates chronological ordering (earlier times to the right).

Since the jump values $\xi_k$ are independent draws from $p(\xi)$ and are independent of interval durations, averaging over $\xi_k$ yields the effective propagator
\begin{equation}
\mathcal E(\theta):=\big\langle e^{\mathcal L_\xi \theta}\big\rangle_\xi
=\big\langle e^{(\mathcal L_a+\epsilon\xi\mathcal L_I)\theta}\big\rangle_\xi ,
\label{eq:frozen_propagator}
\end{equation}
which is the moment-generating function of the jump distribution evaluated on the frozen operator $\mathcal L_\xi$. Averaging over the renewal events yields the exact representation:
\begin{empheq}[box=\fbox]{align}
P(x\pv t)=\sum_{N\ge0}\ \int\!\mathrm d\theta_1\cdots\mathrm d\theta_{N+1}\;
\delta\Big(t-\sum_{k=1}^{N+1}\theta_k\Big)\,
\psi(\theta_1)\cdots\psi(\theta_N)\,\Psi(\theta_{N+1})
\prod_{k=1}^{N+1}\mathcal E(\theta_k)\;P(x\pv 0),
\label{eq:frozen_series}
\end{empheq}
where $\Psi(\theta)=\int_\theta^\infty\psi(\theta')\,\mathrm d\theta'$ is the survival probability of the final interval. Equation~\eqref{eq:frozen_series} is exact for arbitrary $C(x)$, $I(x)$, and $\psi(\theta)$.

\subsection{Complementary Exact Descriptions of Renewal Dynamics \label{sec:two_descriptions}}

Equation~\eqref{eq:frozen_series} and the resummed master equation derived in
Section~\ref{sec:kernel} should not be viewed as competing representations.
They provide two exact descriptions of the same process, but give access to 
complementary classes of observables and analytical results.

\paragraph{Pathwise information and stationary statistics.}Equation~\eqref{eq:frozen_series} is fundamentally a pathwise description.
The dynamics are represented as a sequence of deterministic evolutions under
successive frozen values of the noise, making the geometry of individual
trajectories directly accessible. This representation naturally exposes
non-perturbative properties of the stationary state. The support of the
invariant measure, the local behavior around frozen fixed points, and the tail
indices arising in unstable regimes all follow directly from
Eq.~\eqref{eq:frozen_series} (Sections~\ref{sec:stationary}
and~\ref{sec:numres}). 
These quantities are controlled by global features of the jump 
distribution, such as its support and atomic structure, and therefore 
lie beyond the reach of any finite cumulant truncation.

\paragraph{Closure structure and reduced dynamics.}Although exact, Eq.~\eqref{eq:frozen_series} is generally not explicit.
Practical evaluation requires computing the operator exponential
$e^{\mathcal L_\xi \theta}$, averaging over the jump distribution, and performing
the renewal summation. For general drift and coupling functions, these steps
rarely admit closed-form expressions. 
Even when these operations can be carried out analytically, 
the resulting representation does not reveal transparently which 
features of the jump distribution control the reduced dynamics.

By contrast, the master-equation formulation makes the governing statistical structure explicit.
Its kernel is governed by the Boolean generating function
$\eta$, exact closure is linked to the vanishing of higher Boolean cumulants,
and the leading closure error is controlled by the ratio $b_4/b_2$. These
structural properties are not visible directly in
Eq.~\eqref{eq:frozen_series}, despite the fact that both descriptions are
exactly equivalent.

A further distinction is that the factorization
\eqref{eq:frozen_factorization} relies crucially on the piecewise-constant
character of the forcing. It therefore applies specifically to step-renewal
noise and does not extend to shot noise, continuously varying colored noises,
or more general stochastic perturbations. The generalized-cumulant framework,
by contrast, remains applicable far beyond the renewal setting.

In what follows, Eq.~\eqref{eq:frozen_series} will serve three purposes. First,
it provides an independent derivation of the resummed kernel
(Section~\ref{sec:second_derivation}). Second, it forms the basis of the
stationary-state analysis developed in
Sections~\ref{sec:stationary} and~\ref{sec:numres}. Third, in the Poissonian
case it yields an exact simulation scheme free of time-discretization errors.
In summary, the memory-kernel formulation is the natural tool for studying closure,
relaxation rates, moment dynamics, and spectral observables, whereas the frozen-noise
representation is the natural tool for studying stationary probability densities, support
boundaries, singular structures, and tail behavior. The two descriptions are therefore
complementary rather than redundant, and both are needed for a complete characterization of
renewal-driven systems.

\subsection{Multi-time correlations and Boolean identification}

For an arbitrary waiting-time density $\psi$, the exact $n$-time correlation function $C_n(u_1,\dots,u_n):=\langle\xi(u_1)\cdots\xi(u_n)\rangle$ can be expressed as a sum over 
the $2^{n-1}$ compositions of the ordered sequence $u_1\le\dots\le u_n$~\cite{bblmCSF202}:
\begin{equation}
C_n(u_1,\dots,u_n)=\sum_{p=1}^{n}\ \sum_{\{m_i\}:\sum_{i=1}^p m_i=n}
\overline{\xi^{m_1}}\,\overline{\xi^{m_2}}\cdots\overline{\xi^{m_p}}\;
\mathcal C\big(\{m_i\};u_1,\dots,u_n\big),
\label{eq:renewal_general}
\end{equation}
where $p:=|\pi|$ is the number of blocks in the partition, $\overline{\xi^{m_i}}$ are the moments of the jump distribution, and $\mathcal C$ represents the joint probability that each block falls within a single laminar region while adjacent blocks are separated by at least one renewal event (it replaces the notation with triangular brackets
exploited in 
Ref.~\cite{bblmCSF202}).

Equation~\eqref{eq:renewal_general} is the analytic
counterpart of the structural statement given in Section~\ref{definition_and_partition}: the support is $\mathcal I(n)$
whatever $\psi$, and the per-block coefficients are the moments of the scalar random variable $\xi$.
The emergence of Boolean cumulants requires an additional ingredient, namely the memoryless property of exponential waiting times.

For exponential waiting times, the renewal process acquires an additional structure. Memorylessness causes the probability associated with each interval partition to factorize into independent gap contributions, allowing Eq.~\eqref{eq:renewal_general} to be rewritten in a form adapted to Boolean combinatorics.
In fact, defining the lags $d_i:=u_{i+1}-u_i$ between consecutive observation times and $x_i:=e^{-d_i/\tau}$, Eq.~\eqref{eq:renewal_general} reduces 
to~\cite{bblmCSF202}\footnote{Writing
the composition as $n=n_1+n_2+\dots+n_p$, from 
 Ref.~\cite[Eq.~(41)]{bblmCSF202}, Eq.~\eqref{eq:renewal_general} for exponential waiting time
 becomes explicit:
\begin{align}
\label{eq:corrGen}
&C_n(u_1,\dots,u_n)
= \sum_{p=1}^{n}\ \sum_{\{n_i\}:\sum_i n_i=n}
\overline{\xi^{n_1}}\;\overline{\xi^{n_2}}\cdots\overline{\xi^{n_p}}
\nonumber\\
&\qquad\times
\underbrace{e^{-(u_{n_1}-u_1)/\tau}}_{\text{no renewal inside block }1}
\underbrace{\Big(1-e^{-(u_{n_1+1}-u_{n_1})/\tau}\Big)}_{\text{at least one renewal between blocks }1,2}
\underbrace{e^{-(u_{n_1+n_2}-u_{n_1+1})/\tau}}_{\text{no renewal inside block }2}
\cdots
\end{align}
Each block contributes the survival probability over its own internal gaps,
each junction between consecutive blocks the probability that at least one
renewal occurred there. In terms of  $d_i$ 
and $x_i$, and labelling a composition by
its set $S\subseteq\{1,\dots,n-1\}$ of cut gaps, Eq.~\eqref{eq:corrGen} reads
compactly as Eq.~\eqref{eq:renewal_Cn}}:
\begin{equation}
C_n(u_1,\dots,u_n)
=\sum_{S\subseteq\{1,\dots,n-1\}}\ \Big[\prod_{i\notin S}x_i\Big]\Big[\prod_{i\in S}(1-x_i)\Big]
\prod_{B\in\pi_S}\overline{\xi^{|B|}} ,
\label{eq:renewal_Cn}
\end{equation}
where $S\subseteq\{1,\dots,n-1\}$ denotes the set of cut gaps (i.e. the set of gaps between consecutive blocks) and $\pi_S$ is the corresponding interval partition.

Equation~\eqref{eq:renewal_Cn} can be mapped onto a multiplicative Boolean structure by expanding the junction factors $(1-x_i)$. The mechanism is already visible at fourth order, where the conversion from moments to Boolean cumulants can be followed explicitly.
In this simple situation (and assuming $\overline\xi=0$) Eq.~\eqref{eq:renewal_Cn} contains two non-vanishing terms ($S=\emptyset$, i.e.\ no cut and a single block, and $S=\{2\}$, i.e.\ one cut at the central gap and two blocks):
\begin{equation}
C_4(u_1,u_2,u_3,u_4)=\overline{\xi^4}\,x_1x_2x_3 + \big(\overline{\xi^2}\big)^2\,x_1(1-x_2)x_3 .
\label{eq:C4_moment_form}
\end{equation}
Expanding $(1-x_2)$ and regrouping terms yields:
\begin{equation}
C_4(u_1,u_2,u_3,u_4)=\underbrace{\big[\overline{\xi^4}-\big(\overline{\xi^2}\big)^2\big]}_{=b_4}
x_1x_2x_3 + \underbrace{\overline{\xi^2}}_{=b_2}\,x_1\cdot\underbrace{\overline{\xi^2}}_{=b_2}\,x_3 ,
\label{eq:C4_boolean_form}
\end{equation}
where $b_2=\overline{\xi^2}$ and $b_4=\overline{\xi^4}-(\overline{\xi^2})^2$ are the second and fourth Boolean cumulants of $p(\xi)$. This transformation removes the explicit junction factors $(1-x_2)$, expressing the correlation function as a sum over interval partitions weighted purely by Boolean cumulants and internal survival probabilities.

        \begin{figure*}[h!]  
                \centering
                \includegraphics[width=\textwidth]{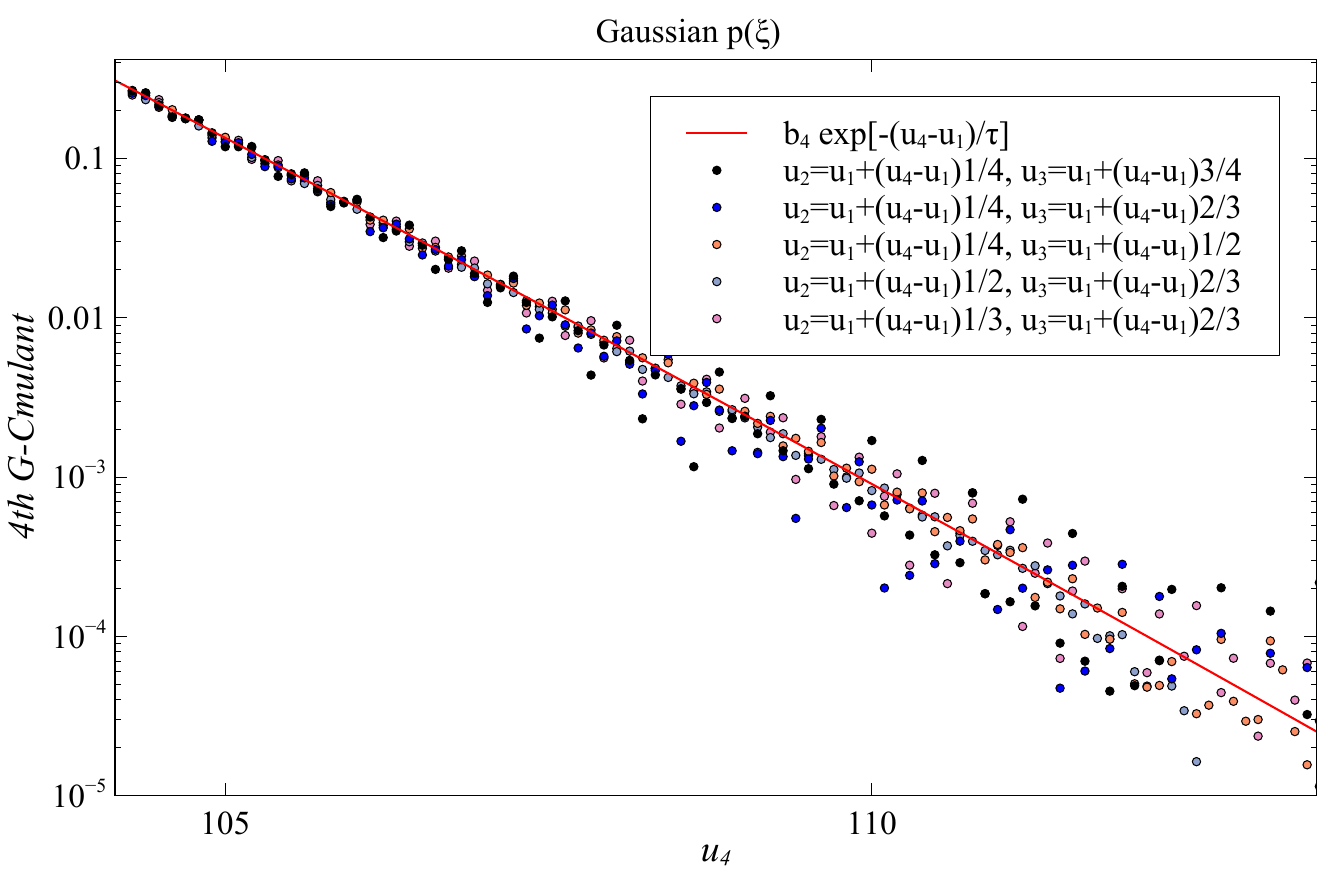}
                \caption{
                        Log-plots of the 4-time  $G$-cumulant 
                        $\doubleangle{\xi(u_1)\xi(u_2)\xi(u_3)\xi(u_4)}^{(G)}$ with $u_1=100$, in the case of exponential waiting time with $\tau = 1$, for the 
                        Gaussian PDF $p(\xi) =  \exp[-\xi^2/2]/\sqrt{2\pi}$ and different intermediate times.
                        Circles represent the results of numerical simulations.
                        Solid line is the theoretical result  in Eq.~\eqref{eq:n_G=bphi}.
                         }
                \label{gaue_T1_t1_100}
        \end{figure*}

The cancellation observed at fourth order is not accidental. Extending the same algebraic rearrangement to arbitrary order (Appendix~\ref{app:cancellation}) yields
\begin{equation}
C_n(u_1,\dots,u_n)=\sum_{\pi\in\mathcal I(n)}\prod_{B\in\pi}
b_{|B|}\,\phi\big(\mathrm{span}(B)\big),
\qquad \mathrm{span}(B):=u_{\max B}-u_{\min B},
\label{eq:renewal_boolean_form}
\end{equation}
where $\phi(u)=e^{-u/\tau}$, and the product of internal factors telescopes to $\prod_{i\in\text{int}(B)}x_i=\phi(\mathrm{span}(B))$.

Equation~\eqref{eq:renewal_boolean_form} has exactly the same structure as the moment--cumulant relation defining $G$-cumulants, Eq.~\eqref{m_vs_G}. The identification is therefore immediate.

\begin{proposition}[Boolean cumulants as exact $G$-cumulants]
\label{prop:main}
For a Poissonian step renewal process with jump distribution $p(\xi)$, the totally time-ordered cumulants are given exactly for all time separations by:
\begin{equation}
\doubleangle{\xi(u_1)\cdots\xi(u_n)}^{(G)}
=b_n\big[p(\xi)\big]\;\phi(u_n-u_1),
\label{eq:n_G=bphi}
\end{equation}
where $b_n[p(\xi)]$ are the Boolean cumulants of $p(\xi)$, and $\phi(u)=e^{-u/\tau}$.
\end{proposition}
Proposition~\ref{prop:main} provides the key link between the algebraic structure developed
in Section~\ref{sec:model} and the dynamics of renewal processes. The interval-partition
combinatorics generated by the renewal structure translates exactly into the Boolean
cumulants of the jump distribution. Consequently, the coefficients governing the
generalized-cumulant expansion are determined entirely by the Boolean cumulant hierarchy of
$p(\xi)$. For Poissonian renewal noise, the $G$-cumulants are therefore not merely analogous
to Boolean cumulants: they coincide exactly with the Boolean cumulants of the jump
distribution, modulated only by the survival probability of a single renewal epoch.

\begin{remark}[Probabilistic interpretation]
Conditional on a specific realization of renewal times, a block corresponds to a set of sampling points residing within a single epoch. Since jump values across different epochs are independent, any mixed $G$-cumulant spanning multiple epochs vanishes identically. Averaging over all renewal trains yields:
\begin{equation}
\doubleangle{\xi(u_1)\cdots\xi(u_n)}^{(G)}
=b_n\big[p(\xi)\big]\times
\Pr\big[\text{all } u_1,\dots,u_n \text{ lie within the same renewal epoch}\big].
\label{eq:kappa_as_probability}
\end{equation}
Thus, the connected $G$-cumulant at order $n$ isolates the single-epoch contribution, whereas multi-epoch configurations contribute exclusively to disconnected moments.
\end{remark}

In Fig.~\ref{gaue_T1_t1_100}, we compare the theoretical and numerical estimates of the fourth $G$-cumulant of 
a step-renewal process with Gaussian amplitude distribution $p(\xi)$. The cumulant is plotted as a function of 
the latest time $u_4$ for different choices of the intermediate times. Symbols represent the
results from the numerical simulations 
of the step-noise process $\xi[t]$, while solid lines correspond to the analytical prediction of 
Proposition~\ref{prop:main}, Eq.~\eqref{eq:n_G=bphi}. For symmetric jump distributions, the fourth $G$-cumulant can be expressed as 
\[ \doubleangle{\xi(u_1)\xi(u_2)\xi(u_3)\xi(u_4)}^{(G)} = \left\langle \xi(u_1)\xi(u_2)\xi(u_3)\xi(u_4)\right\rangle - \left\langle \xi(u_1)\xi(u_2)\right\rangle \left\langle \xi(u_3)\xi(u_4)\right\rangle , \]
while the fourth Boolean cumulant is given by $b_4= \overline{\xi^4}-\left(\overline{\xi^2}\right)^2$.

\begin{remark}[Span dependence and memorylessness]
In Eq.~\eqref{eq:n_G=bphi}, the intermediate time points enter the $G$-cumulant strictly through the total span $u_n-u_1$. 

This reduction to the total span relies entirely on the memoryless property \[ \phi(a+b)=\phi(a)\phi(b). \]

For non-exponential waiting times, this factorization fails, modifying the hierarchy as analyzed in Section~\ref{sec:nonmarkov}.
\end{remark}

\begin{remark}[Comparison with moment-based approximations]
\label{rem:moments_vs_cumulants}
Approximations that replace $(1-x_i)\to 1$ across cut gaps in Eq.~\eqref{eq:renewal_Cn}~\cite{bblmCSF202} recover the exponential factor $\phi(u_n-u_1)$ multiplied by ordinary moments $\overline{\xi^n}$ rather than Boolean cumulants $b_n$. This simplification is valid only in the large-gap limit ($d_i \gg \tau$). In contrast, Eq.~\eqref{eq:n_G=bphi} is exact across all time separations. This distinction is critical for evaluating the memory kernel in Eq.~\eqref{GCum_x_L}, where integration is dominated by short time lags ($t-u\lesssim\tau$).
\end{remark}

\section{The Boolean cumulants of the jump distribution\label{sec:four}}

Proposition~\ref{prop:main} reduces the full noise statistics entering the reduced dynamics to the sequence of Boolean cumulants $\{b_n\}$. The next step is therefore to understand the structure of these cumulants for representative jump laws. In particular, we wish to identify when the resulting hierarchy terminates, when it remains infinite, and how rapidly the higher-order terms grow.

Setting all times to coincide in Eq.~\eqref{eq:renewal_Cn} yields the standard moment--cumulant relation $\overline{\xi^n} = \sum_{\pi \in \mathcal I(n)} \prod_{B \in \pi} b_{|B|}$. In terms of ordinary (rather than exponential) generating functions, $1 + M_o(z) := \sum_{n \ge 0} \overline{\xi^n} z^n = \langle (1 - \xi z)^{-1} \rangle$, the cumulant generating function $\eta(z) = \sum_{n \ge 1} b_n z^n$ satisfies
\begin{equation}
\eta(z) = \frac{M_o(z)}{1 + M_o(z)} = 1 - \frac{1}{1 + M_o(z)} = 1 - \left\langle \frac{1}{1 - \xi z} \right\rangle^{-1},
\label{eq:boolean_ogf}
\end{equation}
which is the defining relation for Boolean cumulants via M\"obius inversion on interval partitions $\mathcal I(n)$~\cite{SpeicherWoroudi1997}.

To illustrate the variety of possible behaviors, we now evaluate Eq.~\eqref{eq:boolean_ogf} for four representative jump distributions, ranging from the uniquely truncating Bernoulli case to heavy-tailed laws with divergent moments, all assumed symmetric and normalized to unit variance. Definitions and standard properties of Boolean cumulants, along with their physical applications, are gathered in Appendix~\ref{app:boolean}. Note that while the Taylor coefficients of $M_o(z)$ and $\eta(z)$ exist only when all moments are finite, the Cauchy-type average in Eq.~\eqref{eq:boolean_ogf} is well defined for any probability measure $p(\xi)$, since the integrand decays as $|\xi|^{-1}$ for complex $z$ off the real axis.

\paragraph{(i) Symmetric Bernoulli, $\xi = \pm a$.} Here $M_o(z) = \sum_{n \ge 1} a^{2n} z^{2n} = a^2 z^2 / (1 - a^2 z^2)$, which gives
\begin{equation}
\eta(z) = a^2 z^2, \qquad b_2 = a^2, \qquad b_n = 0 \quad (n \ge 3).
\label{eq:bernoulli_eta}
\end{equation}
All Boolean cumulants beyond second order vanish identically. This is the unique finite-variance distribution with this property. In non-commutative probability, this property identifies the symmetric Bernoulli distribution as the \emph{Boolean Gaussian}~\cite{Muraki2003}---the counterpart for interval hierarchies of the classical Gaussian or the free semicircle distribution.

Combined with Proposition~\ref{prop:main}, Eq.~\eqref{eq:bernoulli_eta} provides a direct, one-line proof of the result in Ref.~\cite{bCSF159} that telegraph noise is an exact $G$-Gaussian process ($\doubleangle{\xi(u_1)\cdots\xi(u_n)}^{(G)} \propto b_n = 0$ for $n \ge 3$).

\paragraph{(ii) Uniform distribution on $[-a,a]$.} With moments $\overline{\xi^{2k}} = a^{2k}/(2k+1)$, the moment generating function reads $M_o(z) = \frac{1}{2az} \log\left(\frac{1+az}{1-az}\right) - 1$, leading to
\begin{equation}
\eta(z) = 1 - \frac{2az}{\log\left(\frac{1+az}{1-az}\right)}.
\end{equation}
Expanding $\eta(z)=1-az/\operatorname{artanh}(az)$ gives $b_{2m+1} = 0$ and even-order cumulants
\begin{equation}
b_{2m} = c_m\, a^{2m},
\qquad
1-\frac{y}{\operatorname{artanh}y}=\sum_{m\ge1}c_m\,y^{2m}
=\frac{y^2}{3}+\frac{4y^4}{45}+\frac{44y^6}{945}+\frac{428y^8}{14175}+\cdots .
\end{equation}
For unit variance ($a = \sqrt{3}$), the explicit coefficients are $b_2 = 1$, $b_4 = 4/5$, $b_6 = 44/35$, $b_8 = 428/175, \dots$ (numerically $1, 0.80, 1.26, 2.45, 5.30, \dots$). The hierarchy remains infinite, but the growth of the higher-order Boolean cumulants is comparatively moderate.

\paragraph{(iii) Gaussian, $\mathcal N(0,\sigma^2)$.} With moments $\overline{\xi^{2k}} = \sigma^{2k}(2k-1)!!$, one obtains
\begin{equation}
b_{2k} = \sigma^{2k} I_k, \qquad I_k \in \{1, 2, 10, 74, 706, 8162, 110410, \dots\},
\label{eq:gaussian_boolean}
\end{equation}
where $I_k$ are the Boolean cumulants of $\mathcal N(0,1)$, generated by
\begin{equation}
I(z) = \sum_{k \ge 1} I_k z^k = 1 - \left( \sum_{k=0}^{\infty} (2k-1)!! \, z^k \right)^{-1}.
\end{equation}
Combinatorially, $I_k$ counts the number of \emph{indecomposable} pair partitions of $2k$ points---pairings in which no proper non-empty subset of blocks covers $\{1, \dots, 2j\}$ for any $j < k$. 

This combinatorial feature links the present construction to field-theoretic diagrammatics. Equation~\eqref{eq:boolean_ogf} matches the ordinary-generating-function form of Dyson's equation for a propagator, $G = G_0 + G_0 \Sigma G = G_0 / (1 - \Sigma G_0)$, with the self-energy $\Sigma$ played by $\eta$. For Gaussian jumps, this connection is exact: Gaussian moments enumerate all pairings, whereas the Boolean transformation organizes moments into chains of consecutive interval blocks. Extracting the cumulants thus decomposes each pairing into its indecomposable factors. From this perspective, for a Gaussian jump distribution, \emph{the Boolean cumulants represent the one-particle-irreducible (1PI) content of the moments}, and Eq.~\eqref{eq:boolean_ogf} acts as Dyson's resummation of irreducible chains.

As consistency checks, the semicircle distribution (whose moments count non-crossing pairings) has Boolean cumulants given by the Catalan numbers $1,1,2,5,14,42,\dots$, which enumerate indecomposable non-crossing pairings; conversely, the Bernoulli distribution (which supports only a single interval structure) truncates at $b_2$.

\paragraph{(iv) Power-law tails, $p(\xi) \propto (1 + \xi^\beta)^{-1}$.} For even $\beta$, moments $\overline{\xi^n} \propto [\sin\pi(n+1)/\beta]^{-1}$ exist only for $n < \beta - 1$, meaning that only finitely many Boolean cumulants can be defined:
\begin{center}
\renewcommand{\arraystretch}{1.2}
\begin{tabular}{clc}
\hline
$\beta$ & Defined Boolean cumulants & First divergent term \\
\hline
$6$  & $b_2=1, \ b_4=3$ & $b_6$ \\
$8$  & $b_2=1, \ b_4=\sqrt2, \ b_6=6+3\sqrt2$ & $b_8$ \\
$10$ & $b_2=1, \ b_4=\tfrac{\sqrt5}{2}, \ b_6=\tfrac52+\tfrac{\sqrt5}{2}, \ b_8=\tfrac{123}{7}+\tfrac{\sqrt{13449}}{7}$ & $b_{10}$ \\
\hline
\end{tabular}
\end{center}
The hierarchy terminates here for the opposite reason to case (i): not because higher terms vanish, but because the next term diverges. As detailed in Section~\ref{sec:heavy}, this limitation restricts perturbative cumulant \emph{expansions}, rather than the non-perturbative validity of the formalism.

\medskip

In summary, among these distributions, only the symmetric Bernoulli hierarchy truncates. The Gaussian distribution is the furthest from closing: $b_4/b_2 = 2$ (versus $4/5$ for the uniform law), and its Boolean cumulants asymptotically saturate the moments, with the ratio of indecomposable to total pairings $b_{2k}/\overline{\xi^{2k}}$ ($1, 0.67, 0.67, 0.70, 0.75, 0.79, \dots$) approaching unity. Taken together, these examples illustrate the qualitative landscape of the Boolean hierarchy. The symmetric Bernoulli distribution is uniquely characterized by exact termination of the expansion, the uniform law provides an infinite but slowly growing hierarchy, and the Gaussian distribution produces the largest closure corrections among the finite-variance examples considered here. Heavy-tailed laws introduce a qualitatively different limitation, associated with the divergence of moments rather than with the underlying combinatorial structure. Thus, while the Gaussian distribution occupies a distinguished position within the classical cumulant hierarchy, the symmetric Bernoulli distribution plays the corresponding role within the Boolean hierarchy. As will be shown below, this distinction is not merely combinatorial: it translates directly into the structure of the reduced dynamics, the exact closure criterion, and the size of truncation errors.
\section{The memory kernel and its exact resummation\label{sec:kernel}}

\subsection{Series expansion and operator resummation}

Proposition~\ref{prop:main} reduces the statistical content of the memory-kernel expansion to the Boolean cumulants of the jump distribution. Substituting Eq.~\eqref{eq:n_G=bphi} into Eq.~\eqref{GCum_x_L} and exploiting the inheritance property~\eqref{G-cumulant_Omega} yields a remarkable simplification: all stochastic information enters solely through the Boolean cumulant hierarchy $\{b_n\}$, while the remaining structure is determined by the deterministic dynamics.

For a zero-mean jump distribution ($\overline\xi=0$), the memory kernel becomes
\begin{align}
G\big(-\I\epsilon\tilde{\mathcal L}_I(\cdot)\pv t,u\big)
&=\epsilon^2b_2\,e^{-(t-u)/\tau}\,\tilde{\mathcal L}_I(t)\tilde{\mathcal L}_I(u)
\nonumber\\
&\;+e^{-(t-u)/\tau}\tilde{\mathcal L}_I(t)
\left[\sum_{n\ge3}\epsilon^nb_n\!\!\int_u^t\!\!\mathrm du_{n-1}\cdots
\!\!\int_u^{u_3}\!\!\mathrm du_2\,
\tilde{\mathcal L}_I(u_{n-1})\cdots\tilde{\mathcal L}_I(u_2)\right]
\tilde{\mathcal L}_I(u).
\label{eq:kernel_boolean}
\end{align}

At this point, the infinite hierarchy of multi-time correlations has been reduced to the 
sequence $\{b_n\}$ characterizing the jump distribution $p(\xi)$, while the operators 
and the global exponential envelope are determined entirely by the unperturbed dynamics.
The remaining task is therefore purely operational: resum the corresponding operator series.

Transforming to the Schr\"odinger picture via Eq.~\eqref{G_} (i.e., getting rid of the interaction
representation), the $n$-th term of the series $\mathcal G := \sum_{n\ge 2} \mathcal G_n$ depends solely on the time difference $\Delta:=t-u$ and takes the form
\begin{equation}
\mathcal G_n(\Delta)=\epsilon^nb_n\,e^{-\Delta/\tau}\!\!\!\int\limits_{\substack{d_i>0\\ \sum d_i=\Delta}}
\!\!\!\mathcal L_I\,e^{\mathcal L_ad_1}\,\mathcal L_I\,e^{\mathcal L_ad_2}
\cdots e^{\mathcal L_ad_{n-1}}\,\mathcal L_I\;\prod_i\mathrm dd_i .
\label{eq:kernel_chain}
\end{equation}
Here, as in Section~\ref{sec:renewal}, $d_i:=u_{i+1}-u_i$ are the lags
between consecutive insertion times, with $u_1=u$ and $u_n=t$.
The key simplification again stems from memorylessness. Since \[ \phi(\Delta)=e^{-\Delta/\tau} \]
depends only on the total span $\Delta=\sum_i d_i$, the survival factor factorizes as \[
\phi(\Delta)=\prod_i \phi(d_i). \] Equation~\eqref{eq:kernel_chain} is therefore an $(n-1)$-fold
convolution, and its Laplace transform reduces to a product of elementary resolvents,
\begin{equation}
\int_0^\infty\!\mathrm dd\;e^{-sd}e^{-d/\tau}e^{\mathcal L_ad}
=\Big(s+\tfrac1\tau-\mathcal L_a\Big)^{-1}=:P_a^{-1}(s),
\end{equation}
yielding the compact operational expression
\begin{equation}
\hat{\mathcal G}_n(s)=\epsilon^nb_n\,\mathcal L_I\big(P_a^{-1}\mathcal L_I\big)^{n-1}
=b_n\,P_a\mathcal X^n,
\qquad \mathcal X:=\epsilon P_a^{-1}(s)\,\mathcal L_I .
\label{eq:Kn_resolvent}
\end{equation}

Summing over $n \ge 2$ with $b_1 = 0$ evaluates the memory kernel in closed form:
\begin{empheq}[box=\fbox]{align}
\hat{\mathcal G}(s):=\sum_{n\ge2}\hat{\mathcal G}_n(s)=P_a(s)\,\eta(\mathcal X)
=\Big(s+\tfrac1\tau-\mathcal L_a\Big)\,
\eta\!\Big(\epsilon\big(s+\tfrac1\tau-\mathcal L_a\big)^{-1}\mathcal L_I\Big).
\label{eq:kernel_resummed}
\end{empheq}

Equation~\eqref{eq:kernel_resummed} constitutes the central dynamical result of the Boolean-cumulant construction. The entire dependence on the jump distribution is encoded in the scalar Boolean generating function $\eta$, whereas the dynamical information is contained in the resolvent operator $\mathcal X$, with non-commutativity between $\mathcal L_a$ and $\mathcal L_I$ encoded in the operator ordering within $\eta$. The reduced dynamics is therefore governed by the same object that generates the Boolean cumulants of the jump statistics.

Equation~\eqref{eq:kernel_resummed} admits a natural interpretation as an operator-valued Dyson resummation where $\eta(\mathcal X)$ serves as the self-energy operator.

\paragraph{Context and preceding formulations.}
Averaged resolvent techniques for linear systems driven by Kubo--Anderson or kangaroo processes were introduced by Brissaud and Frisch~\cite{BrissaudFrisch1974}, who identified the Kubo number as the effective perturbation parameter. Subsequently, Yoon, Deutch, and Freed~\cite{ydfJCP62} used TTO to formulate resolvent expansions for memory kernels, \textit{implicitly} employing M\"obius inversion on interval partitions $\mathcal I(n)$ [Eq.~(3.10) in Ref.~\cite{ydfJCP62}].
The present derivation reveals the algebraic structure underlying these earlier formulations. Total time-ordered cumulants coincide with Boolean cumulants and factorize as $b_n[p(\xi)]\,\phi(\mathrm{span})$. As a consequence, the entire hierarchy can be resummed into the Boolean generator $\eta(\mathcal X)$. This identification is not merely formal. It leads directly to an explicit criterion for hierarchy termination, singles out symmetric Bernoulli noise as the unique exactly closing jump process, and provides quantitative estimates for the error associated with finite-order closures.

\subsection{Validation checks}
The resummed kernel \eqref{eq:kernel_resummed} admits several independent consistency checks. In particular, it can be recovered both from exact Markovian embeddings of the renewal process and from limiting cases whose solutions are already known. These comparisons provide non-trivial validation of the Boolean-kernel construction.
\paragraph{Joint Markov elimination.}
The first check compares Eq.~\eqref{eq:kernel_resummed} with the exact Markovian representation of the renewal process. For exponential waiting times, the composite pair $(x,\xi)$ forms a Markov process governed by
\begin{equation}
\partial_tP(x,\xi,t)=\mathcal L_\xi P
+\frac1\tau\Big[p(\xi)\!\int\!P(x,\xi',t)\,\mathrm d\xi'-P\Big],
\qquad \mathcal L_\xi=\mathcal L_a+\epsilon\xi\mathcal L_I .
\end{equation}
Integrating out the noise variable $\xi$ in Laplace space yields the exact reduced density $\hat P(s) = \big[R(s)^{-1} - \tau^{-1}\big]^{-1} P(x\pv 0)$, where $R(s) := \big\langle(s + \tau^{-1} - \mathcal L_\xi)^{-1}\big\rangle_\xi$ is the averaged frozen resolvent. The total generator is $\hat{\mathcal G}_{\rm tot}(s) = (s + \tau^{-1}) - R(s)^{-1}$. 

Subtracting the deterministic drift contribution $\mathcal L_a$ recovers Eq.~\eqref{eq:kernel_resummed} exactly, since $R = [1 + M_o(\mathcal X)]P_a^{-1}$ and $(1 + M_o)^{-1} = 1 - \eta$ together imply $\hat{\mathcal G}_{\rm tot} - \mathcal L_a = P_a\,\eta(\mathcal X)$. The Boolean-kernel representation is therefore fully consistent with the exact Markovian embedding of the renewal process.

\paragraph{Free L\'evy walk.}
The second check considers a limiting case whose exact solution is already known and widely used in continuous-time random-walk theory. In the absence of drift ($\mathcal L_a = 0$), the distinction between interaction and Schr\"odinger pictures vanishes, and the resolvent operator becomes scalar in Fourier space ($\mathcal L_I = \partial_x \to \I k$). Using $1 + M_o(w) = \langle(1 - \xi w)^{-1}\rangle$, Eq.~\eqref{eq:kernel_resummed} reduces to
\begin{equation}
\hat P(k,s)=\frac{\tau A}{\tau-A},
\qquad A:=\Big\langle\frac{1}{s+\tau^{-1}-\I k\xi}\Big\rangle ,
\end{equation}
which is the Montroll--Weiss formula $\langle\hat\Psi\rangle / (1 - \langle\hat\psi\rangle)$ for continuous-time random walks with exponential waiting times and arbitrary velocity distributions~\cite{zdkRMP87}. For symmetric binary jumps ($\xi = \pm a$), this recovers the telegrapher's equation~\cite{bCSF159}. 
Structurally, the Montroll--Weiss denominator is precisely the Boolean $\eta$-transform of the velocity distribution. Thus, the classical Montroll--Weiss theory appears as a particular realization of the general Boolean-kernel framework developed here.

\subsection{Alternative derivation and non-exponential generalization\label{sec:second_derivation}}

The resummed kernel admits a second, completely independent derivation based on the frozen-noise representation of Section~\ref{sec:frozen}. This derivation is important for two reasons. First, it reveals the origin of the resolvent structure directly at the level of the renewal process. Second, it extends naturally to arbitrary waiting-time distributions, whereas the Boolean identification of Proposition~\ref{prop:main} relies on memorylessness. The two approaches therefore play complementary roles. The frozen-noise route establishes the exact kernel and its generalization beyond Poissonian renewal statistics, while the cumulant route identifies the algebraic structure of its coefficients and yields explicit closure and error criteria.

Taking the Laplace transform of the renewal series \eqref{eq:frozen_series} yields the geometric representation $\hat P(s) = [1 - \hat\psi_{\mathcal E}(s)]^{-1} \hat\Psi_{\mathcal E}(s) P(x\pv 0)$, where
\begin{equation}
\hat\psi_{\mathcal E}(s):=\int_0^\infty\!\!\mathrm d\theta\;e^{-s\theta}\,\psi(\theta)\,
\mathcal E(\theta),
\qquad
\hat\Psi_{\mathcal E}(s):=\int_0^\infty\!\!\mathrm d\theta\;e^{-s\theta}\,\Psi(\theta)\,
\mathcal E(\theta),
\end{equation}
and $\mathcal E(\theta) = \langle e^{\mathcal L_\xi \theta}\rangle_\xi$ is the averaged frozen propagator [Eq.~\eqref{eq:frozen_propagator}]. Comparing this with the Laplace-transformed equation of motion $s\hat P - P_0 = \hat{\mathcal G}_{\rm tot}\hat P$ gives the total generator for an \textit{arbitrary waiting-time density} $\psi(\theta)$:
\begin{equation}
\hat{\mathcal G}_{\rm tot}(s)=s-\frac{1-\hat\psi_{\mathcal E}(s)}
{\hat\Psi_{\mathcal E}(s)} .
\label{eq:kernel_general_WT}
\end{equation}
Equation~\eqref{eq:kernel_general_WT} is the most general kernel expression obtained in the present work. Unlike Eq.~\eqref{eq:kernel_resummed}, it does not rely on exponential waiting times and therefore remains valid whenever the renewal process is characterized by a single waiting-time density $\psi(\theta)$.
The memory kernel corresponds to the non-drift contribution $\hat{\mathcal G} = \hat{\mathcal G}_{\rm tot} - \mathcal L_a$.

Equation~\eqref{eq:kernel_general_WT} applies whenever all waiting times are drawn from a single distribution $\psi(\theta)$, rendering the renewal train a pure convolution.  For a general equilibrated process the
first waiting time is drawn from $\psi_{\rm eq}(\theta)=\Psi(\theta)/\langle
\theta\rangle$ and the chain acquires a distinct first factor, but the kernel
remains a function of $t-u$ alone. What breaks this structure is aging. In that case, the distribution of the first waiting time depends on the elapsed age of the process and therefore on the initial time itself. The renewal chain loses its convolution structure and the kernel becomes a genuine two-time object $G(t,u)$ rather than a function of the time difference alone.

For exponential waiting times ($\psi(\theta) = e^{-\theta/\tau}/\tau$ and $\Psi(\theta) = e^{-\theta/\tau}$), the relation $\hat\psi_{\mathcal E} = \hat\Psi_{\mathcal E}/\tau$ holds, and the integral reduces to the averaged frozen resolvent
\begin{equation}
\hat\Psi_{\mathcal E}(s)
=\int_0^\infty\!\!\mathrm d\theta\;e^{-(s+1/\tau)\theta}\big\langle e^{\mathcal L_\xi \theta}\big\rangle_\xi
=\Big\langle\big(s+\tfrac1\tau-\mathcal L_\xi\big)^{-1}\Big\rangle_\xi = R(s).
\end{equation}
Equation~\eqref{eq:kernel_general_WT} then simplifies to
\begin{equation}
\hat{\mathcal G}_{\rm tot}(s)=\Big(s+\tfrac1\tau\Big)-R(s)^{-1},
\end{equation}
in exact agreement with Eq.~\eqref{eq:kernel_resummed}.

The two derivations therefore converge to the same kernel from complementary directions. The frozen-noise representation explains the origin of the resolvent structure, whereas the cumulant construction reveals why the resulting coefficients are governed by the Boolean cumulants of the jump distribution. In
Eq.~\eqref{eq:kernel_general_WT} the jump statistics enters only through
$\mathcal E(\theta)=\langle e^{\mathcal L_\xi \theta}\rangle_\xi$, the \emph{moment}
generating function evaluated on the frozen generator: the moments are the
native ingredient of the renewal chain, and the Boolean cumulants appear only
after the renewal structure has been resummed, which requires the exponential
form of $\psi$.

In this sense, Eq.~\eqref{eq:kernel_general_WT} shows that the Boolean construction derived for Poissonian renewal noise is embedded within a broader renewal framework. Exponential waiting times do not merely simplify the kernel: they reveal the specific algebraic structure that makes the Boolean hierarchy emerge.

Since the waiting-time density enters Eq.~\eqref{eq:kernel_general_WT} only
through $\hat\psi_{\mathcal E}$ and $\hat\Psi_{\mathcal E}$, the small-$s$
behaviour of those transforms determines the structure of the kernel directly:
for $\psi(\theta)\sim(T/\theta)^\mu$ with $1<\mu<2$ the numerator
$1-\hat\psi_{\mathcal E}$ vanishes as $s^{\mu-1}$ rather than as $s$, and the
master equation acquires a fractional time derivative of order $\mu-1$. This is
the familiar structure of subdiffusive continuous-time random walks, obtained
here for arbitrary drift and state-dependent coupling; a proper treatment,
including the aging that accompanies that regime, is deferred to separate work.

\section{Exact closure and finite-order criteria\label{sec:when_exact}}

Since each power of $\mathcal X$ in $\eta(\mathcal X)$ carries one factor
$\mathcal L_I=\partial_xI(x)$, hence one power of $\epsilon$, finite-order closure of the memory kernel is equivalent to finite-order termination of the Boolean generating function $\eta$.

Thus, the memory kernel terminates at finite order in the coupling if and only if $\eta$ is a polynomial.

Polynomial termination occurs exclusively for symmetric Bernoulli noise, for which $b_n = 0$ for $n \ge 3$ [Eq.~\eqref{eq:bernoulli_eta}].\footnote{The converse follows from the analytic structure of $\eta$ (Appendix~\ref{sec:analytic_zeta}). If $\eta$ is a polynomial of degree $d$, the reciprocal Cauchy transform $F(y)=y\,[1-\eta(1/y)]$ reads $F(y)=y-b_1-\sum_{k=2}^{d}b_k\,y^{1-k}$. For any probability measure with finite variance, $F$ admits the Nevanlinna representation $F(y)=y-\overline{\xi}+\int\mathrm d\sigma(t)/(t-y)$, with $\sigma$ a finite positive measure~\cite{SpeicherWoroudi1997}. Hence $\int\mathrm d\sigma(t)/(t-y)=-\sum_{k=2}^{d}b_k\,y^{1-k}$ is analytic on $\mathbb C\setminus\{0\}$, and the Stieltjes inversion formula forces $\sigma=b_2\,\delta_0$, i.e.\ $b_k=0$ for $k\ge3$. With $b_1=\overline{\xi}=0$, $\eta(z)=b_2z^2$ is the Boolean generator of the symmetric Bernoulli law with $a^2=b_2$.} In this case, Eq.~\eqref{eq:kernel_resummed} collapses to
\begin{equation}
\hat{\mathcal G}(s)=\epsilon^2a^2\,\mathcal L_I
\Big(s+\tfrac1\tau-\mathcal L_a\Big)^{-1}\mathcal L_I ,
\label{eq:kernel_dichotomous}
\end{equation}
which transforms back to time domain as
\begin{equation}
\label{MEG_2}
\partial_t P(x\pv t)={\mathcal L}_aP(x\pv t)+\epsilon^2a^2\,\mathcal L_I
\int_0^t\mathrm du\; e^{-u/\tau} e^{{\mathcal L}_a u}
\, \mathcal L_I P(x\pv u).
\end{equation}
No higher powers of $\mathcal L_I$ survive. The infinite hierarchy generated by Eq.~\eqref{eq:kernel_resummed} collapses exactly to a single second-order term. The second-order closure is therefore exact at arbitrary coupling strength $\epsilon$, correlation time $\tau$, and state-dependent coupling $I(x)$.

Although Eq.~\eqref{eq:kernel_dichotomous} is second order in $\epsilon$, it is generally of infinite order in spatial derivatives $\partial_x$, because $P_a^{-1}$ involves arbitrary powers of $\mathcal L_a = \partial_x C(x)$. The operator remains strictly second order in spatial derivatives (an exact Fokker--Planck equation with memory) if and only if $\mathcal L_a$ preserves the polynomial degree of the operand. This occurs only for linear drift $C(x) = \gamma x$, where the adjoint of $\mathcal L_a$ acts on degree-$n$ terms as the scalar $-n\gamma$, giving $P_a^{-1} \to (s + \tau^{-1} + n\gamma)^{-1}$.
The preceding observations can be summarized in a precise uniqueness statement.
\begin{theorem}
\label{thm:uniqueness}
For a step renewal process with exponential waiting times, the memory kernel of the master equation for $P(x\pv t)$ terminates at second order in the coupling $\epsilon$ if and only if the jump distribution $p(\xi)$ is a symmetric Bernoulli distribution. If, in addition, the drift $C(x)$ is linear, the master equation reduces to an exact Fokker--Planck equation with memory.
\end{theorem}

That dichotomous Markovian noise admits exact treatment is well known and
simple; for the step-renewal setting it is established, for instance, in
Ref.~\cite{bblmCSF202}. In the free case, $\mathcal L_a=0$, Eq.~\eqref{MEG_2}
reduces to the familiar telegrapher's equation, and we shall accordingly refer
to it as the \emph{extended telegraph equation}. 
What the Boolean identification adds is both a reason and a uniqueness statement. Exact solvability
 is not a peculiar computational feature of two-state noise, nor a special property of the
  telegraph process itself. Rather, it reflects a deeper algebraic fact: among all finite-variance
jump distributions, the symmetric Bernoulli law is the unique distribution whose Boolean hierarchy
terminates exactly at second order. In this precise sense, the Bernoulli distribution plays,
within the Boolean hierarchy, the same role that the Gaussian distribution plays within the
classical cumulant hierarchy.

The contrast with classical Gaussian reductions is instructive. For classical Gaussian processes, second-order closure of the memory hierarchy is generally obtained only after a local linearization of the dynamics, since Gaussianity is preserved exclusively under linear transformations. By contrast, Poissonian dichotomous noise achieves exact second-order closure for arbitrary nonlinear drift $C(x)$ and arbitrary state-dependent coupling $I(x)$. The underlying process is not Gaussian in the classical sense, but $G$-Gaussian in the sense that all higher $G$-cumulants vanish identically. For every other jump distribution, the Boolean generating function $\eta(z)$ contains infinitely many non-vanishing terms. Consequently, the memory kernel retains powers of $\mathcal L_I$ of arbitrary order and the second-order closure becomes an approximation whose accuracy must be quantified. This is the problem addressed in the following sections.

\section{Heavy-tailed distributions: Non-perturbative applicability\label{sec:heavy}}

An important consequence of the resummed formulation is that its validity does not depend on the existence of a complete moment hierarchy. The operator-valued kernel in Eq.~\eqref{eq:kernel_resummed} remains well defined even for jump distributions for which conventional cumulant expansions diverge. The reason is that $\eta$ is not defined by its moments. From
$1+M_o(w)=\langle(1-\xi w)^{-1}\rangle$,
\begin{equation}
\eta(w)=1-\Big\langle\frac{1}{1-\xi w}\Big\rangle^{-1},
\label{eq:eta_analytic}
\end{equation}
that is, $\eta$ is built from the Cauchy transform of $p(\xi)$, equivalently
$\eta(w)=1-w\,F(1/w)$ with $F$ the reciprocal Cauchy transform (Appendix~\ref{sec:analytic_zeta}). The expectation in Eq.~\eqref{eq:eta_analytic} converges for any probability measure $p(\xi)$ because the integrand decays as $|\xi|^{-1}$. Taylor-series coefficients, by contrast, require the existence of finite moments and therefore of the corresponding Boolean cumulants. Similarly, the operator resolvent norm $\|(s + \tau^{-1} - \mathcal L_a - \epsilon\xi\mathcal L_I)^{-1}\|$ remains bounded as $O(|\xi|^{-1})$.

For example, if $p(\xi) \propto (1 + \xi^6)^{-1}$, the sixth moment diverges and $b_6$ is undefined, causing the cumulant series to fail beyond fourth order. However, the failure occurs at the level of the perturbative hierarchy, not at the level of the exact kernel. In fact, $\eta(\I y)$ remains smooth and bounded for all real $y$. This non-perturbative stability reflects the fact that every probability measure is infinitely divisible in Boolean probability theory~\cite{SpeicherWoroudi1997}, as $\eta(z)/n$ defines a valid measure for all $n \ge 1$ without requiring L\'evy--Khintchine integrability conditions. A particularly striking example is provided by the Cauchy distribution. Although none of its ordinary moments exist, one finds, for real $y$, $1 + M_o(\I y) = (1 + |y|)^{-1}$, yielding $\eta(\I y) = -|y|$ exactly, despite the absence of defined moments.

The practical implication is that the range of applicability of Eq.~\eqref{eq:kernel_resummed}
is considerably broader than that of the Boolean-cumulant expansion itself. Consequently, superdiffusive
L\'evy walks driven by heavy-tailed velocity distributions fall directly within the scope of
Eq.~\eqref{eq:kernel_resummed}. Evaluating $\eta$ on
the \emph{operator} $\mathcal X$ requires $\eta$ to be analytic on a domain containing
the spectrum of $\mathcal X$: in the presence of a drift  the condition must be checked case by case. 

For unbounded jump distributions with multiplicative coupling, $\eta(u)$ may acquire an imaginary
component, signalling a genuine dynamical instability in which state-dependent forcing overwhelms
the restoring drift. Thus, the limitations of the formalism arise not from the divergence of
moments, but from the dynamical properties of the underlying stochastic system.

\section{Role of memorylessness and non-Markovian deformations\label{sec:nonmarkov}}

The interval partition support in Eq.~\eqref{eq:renewal_general} depends solely on $\xi[t]$ remaining constant within each renewal epoch, holding for any waiting-time density $\psi(\theta)$. However, the reduction of block weights to Boolean cumulants requires an additional ingredient, namely the multiplicative factorization of gap probabilities,
\begin{equation}
\prod_i\phi(d_i)=\phi\Big(\sum_i d_i\Big),
\label{eq:memoryless}
\end{equation}
which is satisfied exclusively by exponential waiting times.
Equation~\eqref{eq:memoryless} is the key property responsible for all subsequent simplifications. It is the same multiplicative structure that collapses the renewal correlations onto the total span, produces the Boolean identification of Proposition~\ref{prop:main}, and ultimately allows the kernel to be resummed into the operator-valued transform $\eta(\mathcal X)$.

When memorylessness fails, the interval-partition support survives, but gap
probabilities retain non-multiplicative correlations across blocks. At fourth
order, and for a symmetric jump law, the $G$-cumulant deforms to
\begin{equation}
\doubleangle{\xi(u_1)\xi(u_2)\xi(u_3)\xi(u_4)}^{(G)}
=b_4\,P_1(u_1,u_4)
+b_2^2\Big[A_{\{12\}\{34\}}-P_1(u_1,u_2)\,P_1(u_3,u_4)\Big],
\label{eq:renewal_deformation}
\end{equation}
where $P_1(u,u'):=\Pr[\text{no renewal in }[u,u']]$ is the probability that
$u$ and $u'$ lie in the same renewal epoch, which reduces to $\phi(u'-u)$ in the
Poissonian case, and
\[A_{\{12\}\{34\}}:=\Pr[\text{no renewal in }[u_1,u_2]\text{ and in }[u_3,u_4]].\]
Equation~\eqref{eq:renewal_deformation} shows explicitly how the Boolean
structure is deformed once the multiplicative property~\eqref{eq:memoryless} is
lost. The bracket is the covariance of the no-renewal events on the two
disjoint intervals $[u_1,u_2]$ and $[u_3,u_4]$. It vanishes identically for
Poissonian renewal, where Eq.~\eqref{eq:renewal_deformation} reduces to
Proposition~\ref{prop:main}, but not otherwise. It also vanishes at vanishing
time lags, where all no-renewal probabilities tend to one, so that the
$G$-cumulant reduces to $b_4$ for every waiting-time density, and for large
intermediate gaps ($d_2\to\infty$), where the renewal events on the two
intervals decorrelate and the Boolean form $b_4\,P_1(u_1,u_4)$ is recovered.

For power-law waiting times, single-block configurations dominate at long time
lags~\cite{bblmCSF202}. In Eq.~\eqref{eq:renewal_deformation} this means that
$A_{\{12\}\{34\}}$ is dominated by the single-epoch configuration,
$A_{\{12\}\{34\}}\simeq P_1(u_1,u_4)$, while the product
$P_1(u_1,u_2)P_1(u_3,u_4)$ becomes negligible, so that the $G$-cumulant
approaches $(b_4+b_2^2)\,P_1(u_1,u_4)=\overline{\xi^4}\,P_1(u_1,u_4)$. More
generally, one expects
$\doubleangle{\xi(u_1)\cdots\xi(u_n)}^{(G)}\simeq\overline{\xi^n}\,P_1(u_1,u_n)$:
the per-block coefficients become the ordinary moments rather than the Boolean
cumulants.

In this regime, the interval hierarchy effectively collapses to a single dominant block. The geometric chain whose resummation generates $\eta=M_o/(1+M_o)$ is therefore no longer active, and the Boolean hierarchy degenerates back to the ordinary moment hierarchy.
The asymptotic replacement \[ \eta \longrightarrow M_o \] is therefore not a new statistical structure, but a return to the native renewal description encoded in $\mathcal E(\theta)=\langle e^{\mathcal L_\xi \theta}\rangle_\xi$ of Section~\ref{sec:second_derivation}.

Beyond the coefficients, non-exponential waiting times change the problem
structurally. A power-law renewal process ages: for $\mu>2$ the dependence on
the initial time is a decaying transient, while for $1<\mu<2$ the mean waiting
time diverges, stationarity is never reached, and the dependence on $t_0$ never
disappears. The first epoch then carries a different waiting-time density, the
renewal chain \eqref{eq:frozen_series} is no longer a convolution, and the
memory kernel becomes a genuine two-time object. The explicit Boolean resummations and closure results derived in Sections~\ref{sec:kernel}--\ref{sec:when_exact} no longer survive in their present form. We
therefore restrict the present work to the memoryless case, where the process is
stationary from any origin and aging is absent.

\paragraph{Outlook: non-exponential waiting times.}
The structure displayed by Eq.~\eqref{eq:renewal_deformation} is not specific
to fourth order. For an arbitrary waiting-time density, the $G$-cumulants still
admit an exact expansion over interval partitions in which the coefficients
depend only on the Boolean cumulants of the jump law, whereas the weights of
the partitions depend only on the renewal process. The Poissonian case of
Proposition~\ref{prop:main} is the distinguished point at which all weights
except that of the single block vanish identically. This separation makes it
possible to define indices that measure the distance of a given renewal process
from memorylessness independently of $p(\xi)$, and to characterize how the
per-block coefficients interpolate between the Boolean cumulants $b_n$ and the
ordinary moments $\overline{\xi^n}$ as the time lags grow. Preliminary
numerical evidence indicates that the interpolation towards $\overline{\xi^n}$
is a genuine signature of long memory (power-law waiting times), and not merely
of a departure from the exponential law: waiting-time densities with an
exponential tail or with an increasing hazard rate behave differently. These
developments, together with the aging regime $1<\mu<2$ and the resulting
two-time memory kernels, will be presented in a forthcoming work.
The memoryless case should therefore be viewed not as a generic renewal process, but as the distinguished point at which the interval-partition structure becomes exactly Boolean.

\section{Scope and model-dependent assumptions\label{sec:scope}}

At this stage, it is useful to distinguish between results that depend solely on the renewal structure of the forcing and those that require additional assumptions on the underlying dynamics. The Boolean identification, the resummed memory kernel, and the exact closure criterion are consequences of the stochastic architecture of the renewal process itself. Subsequent quantitative results, including error estimates, stationary densities, and support bounds, require increasingly specific assumptions on the drift and coupling operators.

Constructing quantitative estimates or closed-form densities requires additional physical assumptions. Table~\ref{tab:scope_summary} categorizes these results according to their underlying structural requirements.

\begin{table*}[h!]
\centering
\renewcommand{\arraystretch}{1.25}
\begin{tabular}{>{\raggedright\arraybackslash}p{0.40\textwidth}>{\raggedright\arraybackslash}p{0.33\textwidth}>{\raggedright\arraybackslash}p{0.19\textwidth}}
\hline
Result & Underlying requirement & Reference \\
\hline
\multicolumn{3}{c}{\textbf{Renewal-driven structural results}}\\
\hline
$\doubleangle{\xi(u_1)\cdots\xi(u_n)}^{(G)}=b_n[p(\xi)]\,\phi(\mathrm{span})$
& Exponential waiting times
& Section~\ref{sec:renewal}
\\
$\hat{\mathcal G}=P_a\,\eta(\mathcal X)$
& Exponential waiting times
& Section~\ref{sec:kernel}
\\
$\hat{\mathcal G}_{\rm tot}$ for arbitrary $\psi$
& Pure renewal structure
& Section~\ref{sec:second_derivation}
\\
Termination iff $p(\xi)$ is symmetric Bernoulli
& Exponential waiting times
& Theorem~\ref{thm:uniqueness}
\\
  Correlation preservation by $N$-atom surrogates (orders $\le 2N-1$)
  & Pure renewal structure
  & Proposition~\ref{prop:corr_preservation}
  \\
Resummed kernel for heavy-tailed $p(\xi)$
& Analyticity of $\eta(\mathcal X)$ on the spectrum of $\mathcal X$
& Section~\ref{sec:heavy}
\\
\hline
\multicolumn{3}{c}{\textbf{Dynamical reduction and closure estimates}}\\
\hline
Discarded kernel operator
$\mathcal D=(b_4/b_2)\mathcal X^2$
& Exponential waiting times
& Eq.~\eqref{eq:Delta}
\\
Operator norm estimate
$\|\mathcal D\|\simeq(b_4/b_2)\lambda^2$
& Isolated dominant relaxation rate
& Section~\ref{sec:scales}
\\
Rate parameter
$\lambda=\epsilon\Lambda\tau$
& Linear drift $C(x)$
& Section~\ref{sec:worked}
\\
\hline
\multicolumn{3}{c}{\textbf{Stationary-state and model-specific results}}\\
\hline
Compact support between frozen fixed points
& Ergodicity and flow stability
& Section~\ref{sec:stationary}
\\
Affine recursion and Beta invariant density
& Linear drift $C(x)$ and affine coupling $I(x)$
& Sections~\ref{sec:worked},~\ref{sec:stationary}
\\
Support boundary
$x^*_{\max}/\sigma$
& Linear drift and constant coupling
& Section~\ref{sec:fpe_test}
\\
Singularities generated by discrete atoms
& Discrete jump distribution $p(\xi)$
& Section~\ref{sec:stationary}
\\
\hline
\end{tabular}
\caption{
Hierarchy of results according to their structural assumptions.
The first block contains properties determined solely by the renewal
architecture of the forcing. The second block concerns quantitative closure
estimates and requires dynamical information about relaxation rates. The third
block contains stationary-state properties that depend on specific features of
the underlying drift, coupling, or jump distribution.
}
\label{tab:scope_summary}
\end{table*}

The strongest assumptions enter only when quantitative bounds are sought. In particular, estimates based on a single effective relaxation rate require the existence of a dominant spectral scale and therefore depend on the global connectivity of phase space.
If non-linear drift creates isolated basins that the noise perturbation cannot bridge, the system lacks a single dominant eigenvalue, and bounds must be evaluated locally on the operator level $\mathcal D = (b_4/b_2)\mathcal X^2$. 

Bounded jump distributions guarantee ergodic exploration of the interval between the extreme frozen fixed points whenever the maximal forcing $\epsilon\max|\xi|$ exceeds the relevant basin-separation threshold. In this regime, phase-space connectivity produces a unique dominant relaxation scale, allowing the operator-level bound $\mathcal D=(b_4/b_2)\mathcal X^2$ to be replaced by the effective scalar parameter $\lambda$ used in the quantitative closure theory developed below.

\section{Quantifying the Discarded Operator: Scale Parameters\label{sec:scales}}

Successive terms in the kernel expansion \eqref{eq:Kn_resolvent} differ by a constant operator factor. For symmetric jump distributions ($b_3=0$), the leading-order correction to the second-cumulant kernel reads
\begin{equation}
\hat K_4(s)=\epsilon^4b_4\,\mathcal L_IP_a^{-1}\mathcal L_IP_a^{-1}
\mathcal L_IP_a^{-1}\mathcal L_I
=\hat K_2(s)\cdot\frac{b_4}{b_2}\,\mathcal X^2 ,
\label{eq:K4}
\end{equation}
meaning that retaining only the second $G$-cumulant discards the relative operator contribution
\begin{equation}
\mathcal D=\frac{b_4}{b_2}\,\mathcal X^2 ,
\qquad \mathcal X=\epsilon\,P_a^{-1}(s)\,\mathcal L_I ,
\qquad P_a(s)=s+\tfrac1\tau-\mathcal L_a .
\label{eq:Delta}
\end{equation}
Equation~\eqref{eq:Delta} is exact and exhibits a complete separation between statistical
and dynamical contributions. The factor $b_4/b_2$ depends solely on the jump distribution,
whereas all dependence on the underlying dynamics is contained in the operator $\mathcal X$. 

Standard
second-order closures, including Mori--Zwanzig projections, second-order cumulant
expansions, and Bourret-type resummations, are usually justified by formal power counting in
the coupling strength $\epsilon$. Equation~\eqref{eq:Delta}, by contrast, identifies
explicitly the leading discarded operator. The approximation error is therefore
characterized not only by its perturbative order, but also by two independent quantities: a
purely statistical factor, $b_4/b_2$, and a purely dynamical factor, $\mathcal X^2$.

\subsection{Physical scale parameters\label{sec:phys_par}}

Evaluating $\|\mathcal D\|$ as a scalar estimate requires mapping the operator expressions onto physical scales. This mapping involves three parameters: the resolvent evaluation point $\Gamma$, the effective coupling strength $\Lambda$, and the deterministic relaxation rate $\gamma$.

\paragraph{Evaluation point ($\Gamma$).}
The resolvent $P_a^{-1}(s)$ must be evaluated at the spectral point corresponding to the physical quantity of interest. For exponentially decaying observables, this is the dominant pole $s=-\Gamma$, where $\Gamma$ is the relaxation rate. For stationary properties, it is $s=0$. 

Whenever the reduced density relaxes to equilibrium, $\Gamma$ corresponds to the leading non-zero eigenvalue of the generator. In monostable systems, $\Gamma$ scales with the local potential curvature at the minimum. In multistable systems with separated time scales, fast intrabasin relaxation versus slow interbasin Kramers transitions, the relevant scale for the memory kernel (whose temporal range is bounded by $\tau$) is the fast intrabasin rate. The kernel does not resolve interbasin dynamics directly.

\paragraph{Effective coupling ($\Lambda$).}
The operator $\mathcal L_I$ is represented by its projection onto the subspace of the chosen observable.
Formally, if $\Gamma$ is the decay rate of interest and
$\varphi_\Gamma$ the corresponding eigenfunction of the adjoint generator,
$\mathcal A_a\varphi_\Gamma=-\Gamma\varphi_\Gamma$, then
\begin{equation*}
\Lambda:=\frac{\langle\varphi_\Gamma|\,\mathcal A_I\,|\varphi_\Gamma\rangle}
{\langle\varphi_\Gamma|\varphi_\Gamma\rangle} ,
\end{equation*}
with $\mathcal A_I$ the adjoint of $\mathcal L_I$. Consequently, $\Lambda$ depends on both the coupling operator and the specific observable being tracked.

For example, in the LIMI/CAM model of Section~\ref{sec:affine_recursion_subsec}, if
 $\langle x\rangle$ is the quantity of interest, one has $\varphi_\Gamma=x$ and
$\mathcal A_Ix=1+\beta x$, whose projection on $x$ is $\beta x$: the additive
part is orthogonal to the sector and drops out, leaving $\Lambda=\beta$. Thus, in this case
the first moment is insensitive to the additive component of the coupling.

\paragraph{Deterministic drift rate ($\gamma$).}
The rate $\gamma$ characterizes the unperturbed linear drift $C(x)=\gamma x$. For non-linear drifts, an effective value $\gamma_{\rm eff}\simeq\big\langle C'(x)\big\rangle_{P_{\rm st}}$ (or $\langle|C'|\rangle$ in unstable regions) serves as a local estimate over the explored phase space
over a time window of duration
$\sim\tau$.

\paragraph{Effective expansion parameter ($\lambda$).}
Projecting $\mathcal X$ onto the spectral sector associated with the observable under consideration yields the scalar estimate
\begin{equation}
\|\mathcal D\|\simeq\frac{b_4}{b_2}\,\lambda^2 ,
\qquad
\lambda=\frac{\epsilon\Lambda}{\tau^{-1}-\Gamma-\mathcal L_a} .
\label{eq:lambda_eff}
\end{equation}
In the absence of (or for a very slow) drift ($\mathcal L_a=0$), noise-induced relaxation gives $\Gamma=O(\epsilon^2)$, so $\lambda \to \epsilon\Lambda\tau$ recovers the classical Kubo number (the ratio of noise correlation time to system response time).

At first sight one might expect stronger deterministic relaxation to improve the accuracy of the second-order closure. Remarkably, this is not the case for relaxation rates. However, although deterministic drift shortens the kernel memory duration to $\tau/(1+\gamma\tau)$, it also shifts the resolvent pole by an equal amount. In Eq.~\eqref{eq:lambda_eff}, the drift operator $\mathcal L_a$ in the denominator is offset by the corresponding shift in $\Gamma$. Hence, the expansion parameter governing decay rates remains the Kubo number $\lambda\simeq\epsilon\Lambda\tau$:
  the error of the closure is \emph{rigorously independent} of $\gamma$ for a
linear drift, and expected to be approximately so, with $\gamma_{\rm eff}$ in
place of $\gamma$, for a nonlinear one.

Explicitly, concerning the behavior of $\langle x\rangle$, the adjoint generator acts as
$\mathcal A_ax=-\gamma x$, so $\mathcal L_a\to-\gamma$ there, while the decay
rate of the perturbed system may be written $\Gamma=\gamma-\mu$, with $\mu>0$
the reduction of the relaxation rate caused by the noise. The denominator of
Eq.~\eqref{eq:lambda_eff} is therefore
\begin{equation*}
\tau^{-1}-\Gamma-\mathcal L_a=\tau^{-1}-(\gamma-\mu)+\gamma=\tau^{-1}+\mu ,
\end{equation*}
in which $\gamma$ has cancelled identically: what the drift does is to shift the
pole by exactly the amount by which it shortens the kernel. That $\mu$ is itself independent
of $\gamma$  for a linear drift is shown in Section~\ref{sec:worked}, where $\mu$ is computed in
closed form.

We observe that in practice $\lambda$ need not be computed from Eq.~\eqref{eq:lambda_eff} in
full. For a relaxation rate the denominator is $\tau^{-1}+\mu$ with
$\mu=O(\epsilon^2)$, so $\lambda\simeq\epsilon\Lambda\tau$; for a stationary
quantity it is $\tau^{-1}+\gamma_{\rm eff}$, which can only reduce it. The Kubo
number
\begin{equation}
\lambda\lesssim\epsilon\Lambda\tau
\end{equation}
is therefore a bound valid for every observable considered here, and the
condition $\epsilon\Lambda\tau\ll1$ is the practical criterion for truncating at
second order. 

It is worth noting that for a stationary moment there is no decay, so  $\Gamma=0$ and the resolvent is
evaluated at $s=0$ instead of at a pole; the denominator is then
$\tau^{-1}+\gamma$, the drift has nothing to cancel against, and
$\lambda\simeq\epsilon\Lambda\tau/(1+\gamma\tau)$ for a linear drift. The two
cases are the same expression evaluated at two different points, and they behave
oppositely: the error on a rate is independent of $\gamma$, that on a stationary
moment falls as $(1+\gamma\tau)^{-2}$. Both are verified in
Section~\ref{sec:worked}.

This distinction is essential. Equation~\eqref{eq:Delta} controls the error of the memory-kernel description and of the dynamical observables derived from it; \textit{it does not uniformly bound non-perturbative features of the stationary density, such as support boundaries or local singular points, which depend on the full jump distribution} $p(\xi)$. These are therefore most naturally analyzed through the frozen-noise representation.
\subsection{Reduction to the Fokker--Planck Description\label{sec:fpe_emergence}}

The question raised in the Introduction can now be answered quantitatively. Second-order Fokker--Planck descriptions often perform well far outside the formal regime in which they are usually derived. The exact kernel developed above explains both why this happens and where the approximation ultimately breaks down.
Second-order
Zwanzig projection and the local-linear approximation are formally justified
only under very weak coupling or extreme time-scale separation, conditions that
are rarely met; yet they are used, and they appear to work. The exact results
above explain why, and, more usefully, say exactly where they stop working.

The master equation with a second-cumulant kernel [Eq.~\eqref{MEG_2}] reduces to a Fokker--Planck equation through two distinct approximations, both introducing errors of order $\lambda^2$:

\begin{enumerate} 
\item \emph{Kernel localization.} 
Approximating  $e^{\mathcal L_a u}P(x, t-u) \approx P(x, t) + O(\epsilon^2 u)$ inside the convolution
removes temporal memory and reduces the non-local kernel to a local operator. The induced relative
error is $O(\lambda^2)$. 
\item \emph{Boolean-hierarchy truncation.} Retaining only the second
Boolean cumulant neglects all higher orders and produces a relative error of $(b_4/b_2)\lambda^2$.
\end{enumerate}

Combining both approximations yields the localized Fokker--Planck equation
\begin{equation}
\label{eq:FPE}
\partial_tP(x\pv t)=\partial_xC(x)P(x\pv t)
+\epsilon^2b_2\,\mathcal L_I\!\!\int_0^\infty\!\!\mathrm d\Delta\;e^{-\Delta/\tau}\,
\tilde{\mathcal L}_I(-\Delta)\,P(x\pv t) .
\end{equation}
The meaning of Eq.~\eqref{eq:FPE} is often misunderstood. Small $\lambda$ guarantees that the Fokker--Planck equation reproduces the memory kernel with relative accuracy $O(\lambda^2)$ and therefore accurately captures the relaxation rates and low-order moments governed by that kernel. It does \emph{not} guarantee that the stationary probability density itself is accurate.

When $\gamma\tau\ll1$, this simplifies to Eq.~\eqref{ME_secOrder_FPE} with diffusion coefficient $D=\epsilon^2b_2\tau$. The noise statistics enter $D$ exclusively through the second  cumulant $b_2$ (note that the second cumulant is the same in the Boolean or any other  cumulant approach, including the standard one).

The Fokker--Planck equation \eqref{eq:FPE} cannot tell a Gaussian jump law from a uniform or a dichotomous one. Once the reduction is performed, all three are represented only through the common coefficient $b_2$. This is a statement about the closure, not about the stochastic process itself. The exact dynamics retains information on the higher Boolean cumulants and on the support of the jump distribution, whereas the Fokker--Planck description does not. The practical consequence is that two distinct error mechanisms coexist. For rates and low-order moments the neglected information contributes only through $(b_4/b_2)\lambda^2$. For stationary densities and tail properties its impact is controlled instead by $\gamma\tau$, as discussed below.

\paragraph{Parameter distinction: $\lambda$ versus $\gamma\tau$.}
Two independent dimensionless parameters control two fundamentally different aspects of the approximation:
\begin{itemize}
    \item $\lambda \lesssim \epsilon\Lambda\tau$ governs weak-coupling accuracy for relaxation rates and low-order moments.
    \item $\gamma\tau$ governs the number of renewal epochs within a relaxation time ($N_{\rm eff} \sim 1/\gamma\tau$), determining bulk Gaussianity and tail behavior.
\end{itemize}

The distinction becomes particularly transparent in the simplest solvable case of additive noise and linear drift, i.e. for  $I(x)=1$ and  $C(x)=\gamma x$. In this case the ratio of the maximum support boundary $x^*_{\max}=\epsilon\max|\xi|/\gamma$ to the stationary standard deviation $\sigma$ is
\begin{empheq}[box=\fbox]{equation}
\frac{x^*_{\max}}{\sigma}=\max|\xi|\,\sqrt{\frac{1+\gamma\tau}{b_2\,\gamma\tau}} .
\label{eq:support_in_sigma}
\end{empheq}
This ratio is independent of the coupling and, for a given jump law, controlled by $\gamma\tau$
alone. As $\gamma\tau \to 0$, the exact compact support boundary is pushed into the far tails (e.g., about $12\sigma$ at $\gamma\tau=0.02$ for uniform jumps). Consequently, while the Fokker--Planck approximation accurately describes the bulk distribution when $\gamma\tau \ll 1$, it fails in the far tails for any $\gamma\tau$, as it assigns non-zero probability beyond $x^*_{\max}$.

\paragraph{Why it works, and where it fails silently.} This answers the question
raised in the Introduction. The classical justifications invoke extreme
time-scale separation, but what is actually required is only $\lambda\ll1$ for
the closure and $\gamma\tau\ll1$ for the density to look Gaussian; neither
demands that the noise be fast in an absolute sense, and for dichotomous jumps
the closure is exact at any $\lambda$ (Theorem~\ref{thm:uniqueness}). This is why
the Fokker--Planck route succeeds more often than its derivation would suggest.
Its failure is equally well localized, and it is silent. The support does not
move as $\gamma\tau$ decreases: $x^*_{\max}$ is fixed, and it is $\sigma$ that
contracts within it. Thus, the Fokker--Planck solution never reproduces it, at
any $\lambda$. The error is one of domain rather than accuracy, and it is in the
unsafe direction: extreme events
are systematically overestimated and  the theory predicts a finite first-passage time into regions that are in fact kinematically inaccessible. What $\gamma\tau\ll1$ buys is  that
this region is pushed far in the tail, i.e., beyond the scale one is likely to probe; since recurrence
times of extreme events are controlled by exactly that region, it is there that
the approximation fails without announcing itself.

\paragraph{Comparison with spike noise and Pawula's theorem.}
For spike renewal noise, the master equation is local in time, but the operational jump factor $\hat p\big(\mathrm i\partial_xI(x)\big)-1$ contains infinite spatial derivatives $\partial_x$. Truncating this series corresponds to truncating the Kramers--Moyal expansion in standard moments. By Marcinkiewicz's theorem, no non-degenerate characteristic function $\hat p$ can be a polynomial, meaning the Kramers--Moyal series cannot terminate at finite order (Pawula's theorem). In contrast, step renewal noise allows exact finite-order termination because its expansion is governed by Boolean cumulants, for which the symmetric Bernoulli distribution forms a valid finite generator.
The contrast is therefore structural rather than technical. Pawula's theorem prevents finite-order closure in the moment hierarchy associated with spike noise, whereas the Boolean hierarchy associated with step renewal noise admits a genuine finite generator. The exact closure of dichotomous renewal noise is a consequence of this algebraic distinction and reflects the unique role of the symmetric Bernoulli distribution as the Boolean analogue of the Gaussian law.

\section{Two Exact Descriptions: Kernel Hierarchies and Atom Hierarchies \label{sec:two_routes}}

The first main result we have obtained, i.e., the \emph{resummed master equation}, Eq.~\eqref{eq:kernel_resummed} is valuable  because
it yields the exact memory kernel as the Boolean generator $\eta$ of the jump distribution
evaluated on a resolvent. It  holds for arbitrary drift and
coupling, and, being a function of $P_a^{-1}(s)$, it is a spectral object. Thus,  this
result concerns spectral data: decay rates, located by the pole
condition on the relevant sector, and the low-order moments those rates control.
Its expansion in Boolean cumulants is what makes the accuracy of a truncation
quantifiable, through $\mathcal D=(b_4/b_2)\mathcal X^2$.

The second main result is the \emph{frozen-noise representation}, Eq.~\eqref{eq:frozen_series}:
the propagator factorizes exactly into a renewal chain of frozen-noise
propagators, because the noise is piecewise constant. It is equally exact, and
it is a pathwise object. What follows from it are the properties of the
stationary PDF, such as support, singular points, tail behaviour, which are
determined by $\max|\xi|$ and by the structure of $p(\xi)$, quantities no
cumulant of any order can see.

Therefore, the termination criterion of
Theorem~\ref{thm:uniqueness} is a statement concerning the first main result: for symmetric
Bernoulli jumps the cumulant series stops, and the second-cumulant closure is
exact. However, it is worth stressing that it should \textit{not} be read as saying that the reduced \emph{density} is then
correct for other jump laws, or that truncating brings one closer to it. 
Indeed, the second-cumulant closure is not an approximation in density space: it is exactly
the renewal process generated by the symmetric Bernoulli law having the same variance. This
fact implies that its
stationary density has the support and the singularity
structure of a two-valued law $\xi=\pm a$: two edge divergences, and a support fixed by
$a=\sqrt{b_2}$ rather than by $\max|\xi|$. For a jump law with many values, or a
continuous one, that is not an approximation of the truth but a different
object. The cumulant expansion converges in the space of \emph{kernels}, not in
the space of densities.

Thus, for rates,
relaxation times and low-order moments, we have to use the kernel-expansion approach: truncate at second
cumulant, and read the error off $b_4/b_2$ and the appropriate $\lambda$. On the other hand,
for
the equilibrium distribution, and for anything controlled by its tails, extreme events, first-passage times, we cannot truncate at all: we have to use the
frozen-noise representation, which gives the answer in closed form and without
an expansion.

\section{Frozen Flows, Atoms, and Stationary-Density Geometry
\label{sec:unbounded_expansion}}

The discussion of the previous section highlights a fundamental distinction.
The Boolean hierarchy naturally organizes the reduced kernel, whereas the
stationary density is controlled by different geometric objects.

For jump laws of unbounded support this distinction becomes particularly
transparent. Consider, for example, Gaussian jumps. Their Boolean cumulants
grow factorially,
$b_{2k}=\sigma^{2k}I_k$, 
with $I_k$ the indecomposable pairings and
$I_k\sim (2k-1)!!$.
The resulting series $\eta(u)=\sum_n b_nu^n$
has zero radius of convergence and is therefore only asymptotic.
Adding cumulants improves the approximation up to an optimal order and
degrades it thereafter.

This behaviour should not be interpreted as a failure of the kernel
description. The exact kernel remains perfectly well defined through the
analytic continuation provided by the Cauchy-transform representation
Eq.~\eqref{eq:eta_analytic}. The resummed expression
Eq.~\eqref{eq:kernel_resummed} therefore remains valid even for jump laws
whose cumulant expansion diverges.

The more important lesson is conceptual. Truncating the Boolean hierarchy does
not produce a hierarchy of stationary densities. A second-order truncation
does not generate a density close to that of the original noise; it generates
the exact density of a different stochastic process, namely the Bernoulli
renewal process with the same variance.

This observation motivates a different route. If the goal is to approximate
stationary measures, then the natural object to approximate is not the kernel
but the jump distribution itself. The relevant expansion variable is therefore
not the cumulant order but the number of atoms used to represent the jump law.

In the frozen-noise formulation, each atom of $p(\xi)$ generates a
corresponding frozen flow and, for bounded dynamics, a corresponding frozen
fixed point. The support, singularity structure, and tail properties of the
stationary density are determined by these geometric objects. Consequently,
progressive refinements of the jump distribution lead naturally to a hierarchy
of stationary-density approximations.

This hierarchy is developed below.

\subsection{An Expansion that Converges in Density Space
\label{sec:atom_expansion}}

Whereas cumulants organize successive approximations of
the reduced kernel, atoms organize successive approximations of the stationary
density itself.

The obstruction is therefore not truncation per se, but the object being
truncated. The support, the location of singular points, and the asymptotic
tail structure are controlled primarily by the geometry of the jump
distribution. The natural way to approximate these quantities is to
approximate $p(\xi)$ directly rather than a series constructed from its
cumulants.

\paragraph{The construction.} Replace $p(\xi)$ by a discrete law with $N$ atoms
chosen so as to reproduce its first $2N-1$ moments: the nodes and weights of
the Gauss quadrature associated with the measure $p$. Each surrogate is a
discrete jump law, hence exactly solvable by the machinery of
Section~\ref{sec:stationary}: its support is the interval spanned by its frozen
fixed points, its singular points are those fixed points, and their exponents
follow from Eq.~\eqref{eq:edge_exponents}. One obtains in this way a hierarchy of
exactly solvable problems indexed by $N$, converging to the original as $N$
grows.

\paragraph{Correlation preservation.} The atom hierarchy is more than a
heuristic device: at each level it reproduces \emph{exactly} a finite portion
of the multi-time statistics of the forcing, for any waiting-time law.

\begin{proposition}[Correlation preservation by atom surrogates]
\label{prop:corr_preservation}
Let $p(\xi)$ have finite moments up to order $2N-1$ and at least $N$ support
points. Let $p_N(\xi)=\sum_{i=1}^{N}w_i\,\delta(\xi-\xi_i)$ be its $N$-point Gauss
quadrature, so that $\sum_iw_i\xi_i^k=\overline{\xi^k}$ for $k\le2N-1$.
Consider two step-renewal processes with the same waiting-time density $\psi$
and the same initial time $t_0$, with jump laws $p$ and $p_N$ respectively, and
initial ensembles either both equal to a common $p_0$ or equal to $p$ and $p_N$
respectively. Then, for arbitrary $\psi$, arbitrary ordered times
$t_0\le u_1\le\dots\le u_n$ and every $n\le2N-1$, the two processes have
identical $n$-time correlation functions and identical $G$-cumulants.
\end{proposition}

\noindent\emph{Proof.} By Eq.~\eqref{eq:renewal_general}, and by its
generalization to an arbitrary initial ensemble in Ref.~\cite{bblmCSF202},
$C_n$ is a linear combination of products
$\overline{\xi_0^{m_1}}\,\overline{\xi^{m_2}}\cdots\overline{\xi^{m_p}}$ with
all $m_i\le n$, where $\overline{\xi_0^{m}}$ denotes a moment of the initial
ensemble. Since the jump values are independent of the renewal times, the
coefficients of this combination depend only on $\psi$, on $t_0$ and on the
times $u_1,\dots,u_n$, and not on the jump law. Gauss quadrature is exact for
all moments of order $\le2N-1$. The $G$-cumulants of order $n$ are polynomials
in the correlation functions of order $k\le n$ [Eq.~\eqref{k_nVSm_n}].
$\square$

\begin{corollary}
\label{cor:kernel_preservation}
For arbitrary drift $C(x)$, coupling $I(x)$ and waiting-time density $\psi$,
the $G$-cumulant expansion~\eqref{GCum_x_L} of the memory kernel generated by
$p_N$ coincides with that generated by $p$ up to and including order
$\epsilon^{2N-1}$. In particular, the Boolean cumulants of $p_N$ and $p$
coincide up to order $2N-1$, so that the Boolean generator of $p_N$,
$\eta_N(z)$, satisfies $\eta_N(z)=\eta(z)+O(z^{2N})$.
\end{corollary}

\begin{remark}[The atom hierarchy as a Pad\'e resummation of the kernel]
\label{rem:pade}
The function
\[
1+M_o^{(N)}(z)=\sum_{i=1}^{N}\frac{w_i}{1-\xi_iz}
\]
is rational of type $[N-1/N]$ in $z$ and agrees with $1+M_o(z)$ through order
$z^{2N-1}$. It is therefore the $[N-1/N]$ Pad\'e approximant of $1+M_o(z)$: this
is the classical correspondence between Gauss quadrature and Pad\'e
approximants of Markov functions, i.e.\ of Cauchy--Stieltjes transforms of
positive measures~\cite{BakerGravesMorris}. Since the reciprocal of an $[L/M]$
Pad\'e approximant is the $[M/L]$ approximant of the reciprocal function,
$\eta_N=1-[1+M_o^{(N)}]^{-1}$ is precisely the $[N/N-1]$ Pad\'e approximant of
the Boolean generator $\eta$. These statements concern the jump law alone and
hold for any waiting-time density. For exponential waiting times the surrogate
kernel is $P_a\,\eta_N(\mathcal X)$, a Pad\'e resummation of the Boolean kernel
series $\sum_nb_n\,P_a\mathcal X^n$; no ordering ambiguity arises, since
$\eta_N(\mathcal X)$ involves powers of a single operator. The two hierarchies of
Fig.~\ref{fig:conceptual_map} therefore truncate the \emph{same} object in two
different ways at the same order: a polynomial truncation of $\eta$ (kernel
hierarchy) and a Pad\'e resummation (atom hierarchy). They agree in kernel
space through $O(\epsilon^{2N-1})$, but only the latter is the exact kernel of a
genuine renewal process: a polynomial of degree larger than two is not the
Boolean generator of any probability measure (Section~\ref{sec:when_exact}).
This is why only the atom hierarchy produces legitimate stationary densities,
with the correct type of support and singularities. For $N=2$ and a symmetric
law both reduce to $\eta(z)=b_2z^2$: the ``Bernoulli equivalence'' of
Fig.~\ref{fig:conceptual_map} is the lowest instance of this correspondence.
For an arbitrary waiting-time density the same mechanism operates on the exact
kernel~\eqref{eq:kernel_general_WT}: replacing $p$ by $p_N$ amounts to
evaluating the average $\langle\cdot\rangle_\xi$ in
$\mathcal E(\theta)=\langle e^{(\mathcal L_a+\epsilon\xi\mathcal L_I)\theta}\rangle_\xi$
by $N$-point Gauss quadrature. Because the term of order $\epsilon^k$ of the
exponential is a polynomial of degree $k$ in $\xi$, this is exact through
$O(\epsilon^{2N-1})$, in agreement with Corollary~\ref{cor:kernel_preservation}.
\end{remark}

\paragraph{Advantages of the atom-based hierarchy.}
Three features distinguish this construction from the cumulant expansion.

First, it converges in the space of densities. Each surrogate contains a larger
number of atoms and therefore a larger number of singular points and a wider
support. Both approach those of the target distribution monotonically as $N$
increases.

Second, its first non-trivial member coincides exactly with the
second-cumulant closure. Two atoms of equal weight reproducing the variance are
precisely the symmetric Bernoulli law, so the truncation studied throughout
this paper is nothing but the $N=2$ element of the present hierarchy. This
provides a geometric interpretation of the exact closure theorem:
Theorem~\ref{thm:uniqueness} reflects the fact that the Bernoulli distribution
is the lowest-order member of a hierarchy defined directly in density space.
Consequently, exact closure for dichotomous noise and the impossibility of
improving the stationary density by adding cumulant orders are two aspects of
the same phenomenon. \emph{For stationary densities, improvement comes from
adding atoms rather than from adding cumulant orders.}

Third, every member of the hierarchy is itself an exact solution of a genuine
renewal problem. The approximate density therefore remains a legitimate
probability density at every stage: correctly normalized, non-negative, and
supported on a finite interval determined by the corresponding frozen fixed
points. Nothing analogous holds for the Fokker--Planck approximation, whose
support remains unbounded regardless of how accurately its kernel reproduces
the underlying dynamics.

\paragraph{Scope and limitations.} The construction uses nothing of the drift or
of the coupling: replacing $p(\xi)$ by an $N$-atom law is a statement about the
noise alone, and the hierarchy is defined for any $C$ and $I$. What a linear
drift and coupling buy is that each surrogate is solvable in closed form; for
nonlinear ones the surrogates remain exactly defined renewal problems, and their
supports still follow from the frozen fixed points, but their densities must be
obtained numerically. The approximation is variational rather than perturbative 
(there is no small parameter and $N$ may be increased at will) and its
convergence is correspondingly slow, so it is useful for a qualitatively correct
density rather than for high precision, the residual error being concentrated
near the edges where the surrogate support still falls short.

For a law of unbounded support the same construction behaves as it should: the
quadrature nodes are finite in number but unbounded in magnitude, those of a
Gaussian law extend as $\sqrt N$, so the surrogate support grows with $N$
instead of converging, the true support being infinite. Section~\ref{sec:CAM}
quantifies the convergence on the LIMI/CAM model.

The hierarchy therefore preserves the qualitative geometry of the stationary measure at 
every truncation level, unlike cumulant-based approximations, which preserve only 
the low-order dynamical information contained in the kernel.

\section{A worked example: the LIMI/CAM model\label{sec:CAM}}

The previous sections
developed the general framework. We now consider a model for which both exact descriptions
introduced in this work can be carried through analytically and compared in detail. 

The
LIMI/CAM model is not merely an illustrative example: it is the simplest renewal-driven
system in which the Boolean-kernel formulation and the frozen-noise formulation remain fully
non-trivial while still admitting extensive analytical treatment. This makes the model an
ideal laboratory for understanding the distinct roles played by the two hierarchies
developed above. On the one hand, it permits an explicit evaluation of the exact kernel, the
corresponding closure approximations, and the error estimate $(b_4/b_2)\lambda^2$. On the
other hand, it allows the stationary density to be constructed directly from frozen flows,
making it possible to investigate support boundaries, singular structures, and tail
behaviour in closed form. 

The discussion therefore proceeds from kernel space to density
space. Section~\ref{sec:worked} derives the exact relaxation rate, while
Section~\ref{sec:closure_cost} quantifies the cost of the second-cumulant truncation and
verifies the estimate $(b_4/b_2)\lambda^2$. Section~\ref{sec:fpe_test} introduces the
additional approximation of temporal localization and shows that the resulting
Fokker--Planck equation accurately reproduces rates and low-order moments while failing to
reproduce the geometry of the stationary density. 

We then turn to the complementary
density-space description. Section~\ref{sec:atoms_cam} develops the atom-based approximation
hierarchy introduced in Section~\ref{sec:atom_expansion}, illustrating how stationary
densities improve through the addition of atoms rather than cumulant orders.
Sections~\ref{sec:stationary} and~\ref{sec:numres} finally abandon approximations
altogether, deriving the exact stationary density from the frozen-noise representation and
validating the resulting predictions numerically. 

The LIMI/CAM model is defined by
 
\begin{equation} 
\dot x=-\gamma x+\epsilon\,\xi[t]\,(1+\beta
x). 
\label{eq:worked_sde} 
\end{equation}

We refer to Eq.~\eqref{eq:worked_sde} as the \emph{linear system with linear multiplicative
interaction} (LIMI) model~\cite{bmCHAOS34}. The terminology is chosen deliberately. In much of the
literature, especially when $\xi[t]$ is assumed Gaussian or white, the same equation is
commonly described as a correlated additive--multiplicative (CAM) noise model. Since the
present work allows for a much broader class of renewal forcings, the acronym LIMI better
emphasizes the structural ingredients of the problem: a linear restoring drift and a linear
state-dependent coupling.
Because the acronym CAM is widely used in the literature, we shall use the combined designation LIMI/CAM when referring to this model. The term LIMI emphasizes the underlying dynamical structure, whereas CAM preserves continuity with the existing literature. 

The model combines \[ C(x)=\gamma x, \qquad I(x)=1+\beta x, \] that is, a linear
deterministic drift and the lowest-order nonlinear correction to purely additive forcing.
Despite this simple structure, it is capable of generating highly non-Gaussian stationary
states, intermittent fluctuations, heavy tails, and anomalous transport, depending on the
value of the multiplicative parameter $\beta$~\cite{lhuEPJB78,wkmCHAOS27,bcmmCHAOS28}.

When $\xi[t]$ is specialized to Gaussian or, more commonly, white noise, the model has been
employed in a variety of physical contexts. In climate and ocean dynamics it has been used
as a reduced description of large-scale variability, including conceptual studies of the El
Ni\~no Southern Oscillation (ENSO) and extreme climate fluctuations
~\cite{sOD60,sAR101,scpJC28,bGRL43,bcmmA2018,bcmmCHAOS28,bmCHAOS31,bmCHAOS34}. Related forms
also arise in particle transport and active-matter models, including descriptions of active
Brownian particles interacting with complex environments~\cite{zzlzSR6}.

Previous analytical studies of the LIMI/CAM model have largely relied on Gaussian,
white-noise, or perturbative assumptions on the driving process. The present work approaches
the problem from a different perspective. The forcing $\xi[t]$ is modeled as a step-renewal
process with arbitrary jump distribution and finite correlation time, allowing the entire
hierarchy developed in Sections~\ref{sec:renewal}--\ref{sec:two_routes} to be tested on a
concrete system.

The LIMI model is particularly useful in this respect because its algebraic structure
permits exact evaluation of both the resummed kernel and the frozen-noise dynamics, making
it possible to compare directly exact results, Boolean-cumulant truncations, Fokker--Planck
reductions, and density-space approximations.

The reason this model is analytically tractable becomes apparent at the operator level. The
interaction-picture coupling operator reads \begin{equation} \tilde{\mathcal L}_I(t) =
\partial_x\big[e^{\gamma t}+\beta x\big] = e^{\gamma t}A+B, \qquad A:=\partial_x, \qquad
B:=\beta\partial_xx, \label{eq:LI_tilde_worked} \end{equation} and satisfies \[ [A,B]=\beta
A, \] thus generating a two-dimensional solvable Lie algebra.

The operator $B$ acts diagonally, \[ Bx^n=\beta(n+1)x^n, \] whereas $A$ lowers the
polynomial degree, \[ Ax^n=nx^{n-1}. \] Since the drift generator is proportional to $B$, \[
\mathcal L_a = \gamma\partial_xx = (\gamma/\beta)B, \] the resolvent 
$P_a^{-1} =(s+\tau^{-1}-\mathcal L_a)^{-1}$ is diagonal as well, acting on $x^n$ as multiplication by $[s+\tau^{-1}-\gamma(n+1)]^{-1}$.

Every term of Eq.~\eqref{eq:Kn_resolvent} is therefore a product of operators that either
act diagonally on the polynomial degree or shift it by one. As a result, all contributions
can be evaluated explicitly within each finite-degree sector.
Equivalently, the commutation relation allows every
$A$ to be moved to the right of every function of $B$ at the cost of a shift of
argument, $A\,g(B)=g(B+\beta)\,A$, bringing any product to the form
(function of $B$)$\times A^k$.

This algebraic closure is the key reason why the LIMI/CAM model is particularly well 
suited as a
worked example: the abstract operator relations developed in the previous sections become
explicitly computable, allowing direct comparison between exact results and the various
approximation schemes.

Note that operator non-commutativity is genuinely relevant only when $\gamma\beta\neq0$. The
additive-noise limit ($\beta=0$) and the zero-drift limit ($\gamma=0$) are effectively
commutative and provide useful benchmark cases against which the fully coupled problem can
be compared.

\subsection{Exact relaxation rate\label{sec:worked}}

For the LIMI/CAM model in Eq.~\eqref{eq:worked_sde} the moment hierarchy closes
exactly, and the closure provides a stringent test of the whole construction.
The reason is the algebraic structure noted after
Eq.~\eqref{eq:LI_tilde_worked}: the drift generator is diagonal on powers of
$x$ and the coupling lowers the degree by one, so the sector of degree $n$ feeds
only on the sector of degree $n-1$. Thus, for a discrete jump PDF $p(\xi)=\sum_iw_i\,
\delta(\xi-\xi_i)$, define
\begin{equation}
v_i^{(n)}(t):=\big\langle x^n(t)\,\chi_i(t)\big\rangle ,
\qquad
\chi_i(t):=\begin{cases}1 & \text{if }\xi[t]=\xi_i,\\ 0 & \text{otherwise,}\end{cases}
\end{equation}
the mean of $x^n$ restricted to the $i$-th noise state, so that
$\langle x^n\rangle=\sum_iv_i^{(n)}$ and $\langle\chi_i\rangle=w_i$ in the stationary
state, the frozen flow
$\dot x=-\lambda_ix+\epsilon\xi_i$ gives
\begin{equation}
\dot v_i^{(n)}=-n\lambda_iv_i^{(n)}+n\epsilon\xi_i\,v_i^{(n-1)}
+\frac1\tau\Big[w_i\sum_jv_j^{(n)}-v_i^{(n)}\Big],
\label{eq:moment_hierarchy}
\end{equation}
a triangular system solvable order by order from $v_i^{(0)}=w_i$. We develop the
case $n=1$ in full, since it is the one for which the relaxation rate can be
compared directly with the resummed kernel; the higher sectors are used in
Sections~\ref{sec:fpe_test} and~\ref{sec:numres}, where the stationary
moments are needed.

 Since $(x,\xi)$ is jointly Markovian, the $v_i:=v_i^{(1)}$ obey
\begin{equation}
\dot v_i=-\lambda_iv_i+\epsilon\xi_iw_i
+\frac1\tau\Big[w_i\sum_jv_j-v_i\Big],
\qquad \lambda_i=\gamma-\epsilon\beta\xi_i ,
\label{eq:first_moment_system}
\end{equation}
a closed linear system whose slowest decay rate governs $\langle x\rangle=\sum_iv_i$: the first term is transport along the frozen flow, the
second the drive from the constant part of $I(x)$, and the bracket the gain and
loss through resampling of the noise at rate $1/\tau$.

Two consequences follow immediately.
First, the matrix governing Eq.~\eqref{eq:first_moment_system} can be written as
\begin{equation}
M=-\gamma\, \mathbbm{1}
+\Big[\mathrm{diag}(\epsilon\beta\xi_i)-\tfrac1\tau\, \mathbbm{1}
+\tfrac1\tau\,w\,\mathbf 1^{\!\top}\Big],
\label{eq:M_matrix}
\end{equation}
where $\mathbbm{1}$ is the identity matrix, $w:=(w_1,\dots,w_N)^{\!\top}$ the vector of
weights and $\mathbf 1:=(1,\dots,1)^{\!\top}$, so that $w\,\mathbf 1^{\!\top}$ is
the rank-one matrix with entries $(w\mathbf 1^{\!\top})_{ij}=w_i$. 
Because the drift enters solely through the shift 
$-\gamma\mathbbm{1}$, the noise-induced modification $\mu$ to the total relaxation rate $\Gamma = \gamma - \mu$ satisfies:
\begin{equation}
\boxed{\ \mu \text{ is strictly independent of }\gamma\ }
\label{eq:mu_indep}
\end{equation}
The second consequence is a closed form for $\mu$. The decay rate sits at a
\emph{pole} of the resolvent. From Eq.~\eqref{eq:kernel_general_WT} with
exponential waiting times, $\hat{\mathcal G}_{\rm tot}=(s+\tau^{-1})-R(s)^{-1}$,
so the Laplace-transformed evolution equation
$s\widehat{\langle x\rangle}-x_0=\hat{\mathcal G}_{\rm tot}\widehat{\langle x\rangle}$
has a pole where $s-\hat{\mathcal G}_{\rm tot}=0$, that is where
\begin{equation}
R_{xx}(s)=\tau ,
\qquad
R_{xx}(s)=\Big\langle\big(s+\tau^{-1}+\gamma-\epsilon\beta\xi\big)^{-1}\Big\rangle_\xi ,
\end{equation}
$R_{xx}$ being the $x$-component of the averaged resolvent: on that sector the
adjoint generator acts as $\mathcal A_\xi x=(-\gamma+\epsilon\beta\xi)x+
\epsilon\xi$, whose coefficient of $x$ is $-\lambda(\xi)$.

Setting $s=-\Gamma=-(\gamma-\mu)$ makes the drift cancel, $s+\tau^{-1}+\gamma=
\tau^{-1}+\mu$, and the condition becomes
$\langle(\tau^{-1}+\mu-\epsilon\beta\xi)^{-1}\rangle_\xi=\tau$. Factoring
$\tau^{-1}+\mu$ out of the bracket,
\begin{equation}
\frac{1}{\tau^{-1}+\mu}\Big\langle\frac{1}{1-\xi u}\Big\rangle_\xi=\tau ,
\qquad
u:=\frac{\epsilon\beta}{\tau^{-1}+\mu}=\frac{\epsilon\beta\tau}{1+\mu\tau} ,
\end{equation}
and since $\langle(1-\xi u)^{-1}\rangle=1+M_o(u)$ this reads
$1+M_o(u)=\tau(\tau^{-1}+\mu)=1+\mu\tau$, i.e.\ $M_o(u)=\mu\tau$. Using
$M_o=\eta/(1-\eta)$,
\begin{empheq}[box=\fbox]{equation}
\mu\,\tau=\frac{\eta(u)}{1-\eta(u)},
\qquad
u=\frac{\epsilon\beta\,\tau}{1+\mu\tau} ,
\label{eq:exact_rate}
\end{empheq}
a self-consistent pair determining $\mu$, with $\eta$ the Boolean cumulant
generating function of $p(\xi)$ and $M_o=\eta/(1-\eta)$ the ordinary moment
generating function. The effective coupling is $\epsilon\beta$: for the first
moment only the multiplicative part of $\mathcal L_I$ contributes, the additive
part being annihilated by $\partial_x$ acting on a constant. Note that $u$ is
fixed by $\mu$ itself and \emph{not} by $\gamma$; in the perturbative
regime, $\mu\tau\ll1$, one has $u\simeq\epsilon\beta\tau$.

\subsection{Closure error analysis\label{sec:closure_cost}}
The general error estimate derived in Section~\ref{sec:scales} can be tested explicitly in the LIMI model. Truncating the Boolean (equivalently, the $G$-cumulant) expansion at second order amounts to replacing $\eta$ by its leading term, $\eta(u)\to b_2u^2$, in Eq.~\eqref{eq:exact_rate}. Two concrete predictions follow, and both are borne out exactly.

\paragraph{The truncated pair, in closed form.} Truncating the cumulant series
at second order means solving Eq.~\eqref{eq:exact_rate} with $\eta\to b_2u^2$;
call $\mu_2$ its root. That truncated pair is elementary: writing $m=\mu_2\tau$
and $v^2=b_2(\epsilon\beta\tau)^2$, it reduces to $m^2+m-v^2=0$, whence
\begin{equation}
\mu_2\tau=\frac{\sqrt{1+4v^2}-1}{2},
\qquad v=\sqrt{b_2}\,\epsilon\beta\tau .
\label{eq:mu_truncated}
\end{equation}

Of course, according to  Theorem~\ref{thm:uniqueness} for symmetric Bernoulli jumps
$\eta(u)=b_2u^2$ \emph{identically}, so the truncation is no truncation at all
and Eq.~\eqref{eq:mu_truncated} gives the true $\mu$, at any coupling and with
$b_2=a^2$.

For any other jump law
Eq.~\eqref{eq:mu_truncated} is an approximation, and expanding both roots in
$u$, with $\eta/(1-\eta)=\eta+\eta^2+O(\eta^3)$ and
$\eta=b_2u^2+b_4u^4+O(u^6)$,
\begin{equation}
\mu\tau=b_2u^2+\big(b_4+b_2^2\big)u^4+O(u^6),
\qquad
\mu_2\tau=b_2u^2+b_2^2u^4+O(u^6),
\end{equation}
the two differing first at order $u^4$ and by $b_4u^4$ exactly. Since
$u=\epsilon\beta\tau+O(\epsilon^3)$, the shift in $u$ between the two
self-consistent solutions contributes only at $O(u^6)$, and
\begin{equation}
\frac{\mu_2-\mu}{\mu}\simeq-\frac{b_4}{b_2}\,u^2 ,
\qquad u\simeq\epsilon\beta\tau .
\label{eq:closure_error}
\end{equation}
This is precisely the scalar realization of the operator estimate $\mathcal D=(b_4/b_2)\mathcal X^2$ derived
in Section~\ref{sec:scales}: the closure error factorizes into a purely statistical contribution,
$b_4/b_2$, and a purely dynamical scale, $u^2$.

 For uniform jumps,
$b_4/b_2=0.8$ (with $\epsilon=\tau=1$):

\begin{center}
\renewcommand{\arraystretch}{1.2}
\begin{tabular}{cccccc}
\hline
$\beta$ & $u$ & $\mu$ exact & $\mu$ second order & measured error & $(b_4/b_2)u^2$\\
\hline
$0.15$ & $0.1467$ & $0.022411$ & $0.022015$ & $1.76\%$ & $1.72\%$\\
$0.20$ & $0.1924$ & $0.039703$ & $0.038516$ & $2.99\%$ & $2.96\%$\\
$0.25$ & $0.2355$ & $0.061763$ & $0.059017$ & $4.45\%$ & $4.44\%$\\
$0.30$ & $0.2756$ & $0.088464$ & $0.083095$ & $6.07\%$ & $6.08\%$\\
$0.40$ & $0.3463$ & $0.155179$ & $0.140312$ & $9.58\%$ & $9.59\%$\\
$0.50$ & $0.4038$ & $0.238444$ & $0.207107$ & $13.14\%$ & $13.04\%$\\
\hline
\end{tabular}
\end{center}

The agreement between the last two columns is to two significant figures over
an eightfold range in the error, and remains within $1\%$ relative even at
$u=0.4$, where the expansion parameter is no longer small. Equation
\eqref{eq:closure_error} is therefore not merely an asymptotic estimate but a
usable formula. 

At least within the parameter range relevant for the present model, the leading Boolean correction already captures essentially the entire truncation error.

A potentially counterintuitive consequence now emerges. The reader should appreciate that because  $\mu$ is independent of $\gamma$, so is $u$, and
so is the error \eqref{eq:closure_error}: a fast deterministic relaxation does
\emph{not} improve the accuracy of the second $G$-cumulant closure for the first
moment of this model. What does decrease with $\gamma$ is the error relative to
the \emph{total} rate $\Gamma=\gamma-\mu$, but only because $\Gamma$ grows while
the absolute error stays fixed, which is not a statement about the quality of
the closure.

The relevant expansion parameter is therefore $u\simeq\epsilon\beta\tau$, the
coupling strength times the noise correlation time: the Kubo number, with the
multiplicative part of the coupling as the relevant amplitude.

\paragraph{Moment existence versus closure accuracy.} Moment existence and closure accuracy are controlled by two entirely different mechanisms and should not be confused. Equation~\eqref{eq:exact_rate} governs
the first moment. As shown by the hierarchy~\eqref{eq:moment_hierarchy}, higher-order
moments follow the same structure, with the $n$-th sector characterized by the coupling
$n\epsilon\beta$ and the same resampling rate. This immediately leads to two distinct
regimes. If $\lambda_i>0$ for every $\xi_i$ in the support of the noise distribution, then
$n\lambda_i+\tau^{-1}>\tau^{-1}$ for all $n$. Consequently, every sector remains stable and
all moments are finite. Conversely, if $\lambda_i<0$ for some $\xi_i$, the quantity
$n\lambda_i+\tau^{-1}$ changes sign once $n>1/(\tau|\lambda_{\min}|)$. Moments of that order
and higher then diverge, yielding precisely the same threshold identified by the Kesten
exponent in Section~\ref{sec:kesten_subsec}. The corresponding stability criterion is
therefore \begin{equation} \epsilon\beta\max|\xi|<\gamma , \label{eq:uniform_criterion}
\end{equation} or, equivalently, $\lambda(\xi)>0$ throughout the support of the
distribution. Condition~\eqref{eq:uniform_criterion} admits a simple physical
interpretation. It guarantees that the frozen fixed point
$x^*(\xi)=\epsilon\xi/\lambda(\xi)$ remains regular for all $\xi$, namely that the
multiplicative fluctuations never overcome the restoring drift. This is also precisely the
condition under which the stationary probability density has compact support; see
Eq.~\eqref{eq:compact_support}. Equation~\eqref{eq:uniform_criterion} determines which
moments exist, whereas the accuracy of the second-cumulant closure is controlled by the
parameter $(b_4/b_2)u^2$, with $u\simeq\epsilon\beta\tau$. The latter is therefore
independent of both Eq.~\eqref{eq:uniform_criterion} and the drift strength $\gamma$. Under
the unit-variance normalization adopted throughout this work, $\max|\xi|=1$ for the
symmetric Bernoulli distribution and $\max|\xi|=\sqrt{3}$ for the uniform distribution,
whereas $\max|\xi|$ is unbounded for a Gaussian distribution. As a result,
Eq.~\eqref{eq:uniform_criterion} cannot be satisfied in the Gaussian case, and high-order
moments diverge regardless of how small the coupling strength is. This conclusion is fully consistent with the non-perturbative discussion of Section~\ref{sec:heavy} and with the Kesten-type tail analysis of Section~\ref{sec:kesten_subsec}. Divergence of sufficiently high moments is therefore a property of the stationary distribution itself and should not be interpreted as a failure of the Boolean-kernel description.

\subsection{Localization and truncation: which parameter controls what\label{sec:fpe_test}}

So far the cost of the closure has been quantified on a spectral observable, namely the
relaxation rate, where the prediction $(b_4/b_2)\lambda^2$ was verified quantitatively. This
analysis concerns the memory kernel and the quantities controlled directly by it. A
different question is how accurately the localized Fokker--Planck description reproduces the
\emph{stationary density}. In that case, the Boolean-kernel expansion alone is no longer
sufficient. As discussed in Section~\ref{sec:fpe_emergence}, two independent parameters
enter: $\lambda$, which governs the error associated with truncation and localization of the
kernel, and $\gamma\tau$, which controls the relative position of the compact support with
respect to the bulk of the stationary distribution. The distinction is important. Reducing
$\lambda$ improves the accuracy of relaxation rates and low-order moments, but does not by
itself guarantee a better approximation of the stationary density. The latter is controlled
by the independent parameter $\gamma\tau$.

\paragraph{Stationary distribution comparison.}
Applying both localization and truncation to the LIMI model yields the Stratonovich Fokker--Planck stationary solution:
\begin{equation}
P^{\rm FP}_{\rm st}(x)\ \propto\ \big(1+\beta x\big)^{-(\nu+1)}
\exp\!\Big[-\frac{\nu}{1+\beta x}\Big] ,
\qquad
\nu=\frac{\gamma}{D\beta^{2}},
\qquad
D=\frac{\epsilon^{2}b_2\tau}{1+\gamma\tau} .
\label{eq:fpe_stationary}
\end{equation}
A direct comparison between Eq.~\eqref{eq:fpe_stationary} and exact numerical sampling of the renewal map confirms that $\lambda$ dictates closure error in observables, while $\gamma\tau$ determines the spatial domain over which the Fokker--Planck approximation remains valid.

\begin{table}[h!]
\centering
\renewcommand{\arraystretch}{1.25}
\begin{tabular}{lccccccc}
\hline
& $\epsilon$ & $\gamma\tau$ & $\lambda$ & $(b_4/b_2)\lambda^2$ &
error on $\sigma$ & support in $\sigma$ & misplaced mass\\
\hline
(a) & $0.5$  & $1$   & $0.150$ & $1.80\%$ & $0.15\%$ & $[-2.0,\,+3.1]$ & $1.46\%$\\
(b) & $0.2$  & $1$   & $0.060$ & $0.29\%$ & $0.02\%$ & $[-2.3,\,+2.7]$ & $1.43\%$\\
(c) & $0.15$ & $0.1$ & $0.045$ & $0.16\%$ & $0.14\%$ & $[-3.2,\,+24.9]$ & $0.00\%$\\
\hline
\end{tabular}
\caption{Comparison between exact renewal statistics and the localized Fokker--Planck approximation. We use
uniform jumps with $\beta=0.3$ and $\tau=1$ throughout, and vary $\epsilon$ and
$\gamma$.}
\label{tab:fpe_comparison}
\end{table}
The results illustrate the practical distinction between kernel accuracy and density accuracy. Reducing $\lambda$ improves the accuracy of observables such as the variance, in agreement with the closure estimate $(b_4/b_2)\lambda^2$, whereas reducing $\gamma\tau$ pushes the exact compact support farther into the tails and therefore improves the density approximation over the physically relevant region.

\subsection{The atom hierarchy at work\label{sec:atoms_cam}}

The Fokker--Planck description of Section~\ref{sec:fpe_test} reproduces the stationary
variance to $0.02\%$ while misplacing $1.43\%$ of the probability outside the true support.
Halving the coupling again would improve the first number while leaving the second
essentially unchanged. This is not a failure of accuracy but of expansion variable: the
cumulant hierarchy controls the kernel, whereas the support is not a property of the kernel.

Section~\ref{sec:atom_expansion} argued that a controlled approximation to the stationary
density exists, but in the number of atoms retained in $p(\xi)$ rather than in the number of
cumulants. The LIMI model provides a natural test case because every member of the resulting
hierarchy remains analytically tractable.

We consider uniform jumps with $\gamma=\tau=1$, $\epsilon=0.5$, and $\beta=0.4$, the same
parameters used for the discrete examples of Section~\ref{sec:numres}. The multiplicative
coupling makes the stationary supports asymmetric. We then replace $p(\xi)$ by the $N$-point
Gauss--Legendre quadrature on $[-\sqrt3,\sqrt3]$, which reproduces the first $2N-1$ moments
of the uniform distribution.

Each surrogate is therefore a discrete jump law that can be solved exactly through the machinery
of Section~\ref{sec:stationary} and compared directly against numerical sampling:
\begin{center} \renewcommand{\arraystretch}{1.2} \begin{tabular}{cccc} \hline $N$ & moments
matched & support & Kolmogorov distance\\ \hline $2$ & $3$ & $[-0.417,\,0.625]$ & $0.092$\\
$3$ & $5$ & $[-0.529,\,0.917]$ & $0.056$\\ $4$ & $7$ & $[-0.574,\,1.063]$ & $0.034$\\ $6$ &
$11$ & $[-0.610,\,1.193]$ & $0.017$\\ \hline exact & --- & $[-0.643,\,1.317]$ & ---\\ \hline
\end{tabular} \end{center}

Several features emerge from these results.

First, convergence is monotonic in both the support and the Kolmogorov distance, the latter
decreasing approximately as $N^{-3/2}$ over the range considered.

Second, the $N=2$ surrogate coincides exactly with the second-cumulant closure. Two equally
weighted atoms reproducing the variance are precisely the symmetric Bernoulli distribution,
and therefore generate the same support, $[-0.417,\,0.625]$, together with the same
Beta-type stationary density obtained previously.

Third, the residual discrepancy is concentrated near the support boundaries. For example, at
$N=6$ the upper edge still underestimates the exact value by approximately $9\%$, whereas
the bulk distribution is reproduced to better than $2\%$.

These observations make the distinction between the cumulant hierarchy and the atomic
hierarchy completely transparent. The cumulant expansion acts in kernel space: increasing
the number of retained cumulants modifies the reduced dynamics but leaves unchanged the
underlying two-state description associated with the $N=2$ Bernoulli law. Consequently, the
stationary density retains the same compact support and the same pair of edge singularities.

This is not merely a matter of convergence rate. For $N=2$ the jump law is exactly the
symmetric Bernoulli distribution, for which the Boolean hierarchy terminates identically at
second order. The $N=2$ density is therefore not the first term of a cumulant hierarchy
waiting to be improved: it is already the exact stationary density of a different stochastic
process.

The atomic hierarchy behaves differently. Increasing $N$ changes the jump law itself and
therefore modifies the support, the number and location of singular points, and the overall
geometry of the stationary density. In this sense, improvement of the density comes from
adding atoms rather than from adding cumulant orders.

The atom hierarchy thus provides in density space the analogue of what the Boolean-cumulant
expansion provides in kernel space: a systematic sequence of controlled approximations,
organized around the geometry of the stationary measure rather than around the dynamics of
the reduced kernel.

\subsection{Exact Stationary-distribution approach for the LIMI model \label{sec:stationary}} We now abandon the kernel description altogether and turn to the second exact route developed in this work, namely the frozen-noise representation of Section~\ref{sec:frozen}. Unlike the Boolean-kernel construction, which is naturally adapted to rates, moments, and closure estimates, the frozen-noise representation provides direct access to stationary densities, support boundaries, singular points, and tail behaviour. Although the analysis below is carried out for the LIMI model \eqref{eq:worked_sde}, several of the underlying structural features are more general. Within each renewal interval the dynamics is deterministic, for arbitrary drift $C(x)$ and coupling $I(x)$, and the stationary state is therefore built from a sequence of frozen deterministic evolutions. What makes the LIMI model special is the linearity of the frozen drift. In that case the interval-to-interval dynamics reduces to a random affine recursion, allowing the machinery of random affine maps to be applied. This additional algebraic structure makes it possible to obtain explicit results for the stationary density and to identify its support, singularities, and tail properties in closed form.

\subsubsection{Mapping the LIMI Model to a Random Affine Recursion
\label{sec:affine_recursion_subsec}} 
The frozen-noise representation becomes particularly transparent in the LIMI model because the dynamics within each renewal interval is linear. The continuous-time process can therefore be reduced exactly to a discrete-time recursion between successive renewal events.

Within a laminar interval during which the noise
remains frozen at the value $\xi$, Eq.~\eqref{eq:worked_sde} reduces to the linear equation
\begin{equation} 
\dot x = -\gamma x + \epsilon\xi(1+\beta x) =
-\big(\gamma-\epsilon\beta\xi\big)x + \epsilon\xi. 
\end{equation} 
The dynamics can therefore
be written as 
\begin{equation} 
\dot x = -\lambda(\xi)\big[x-x^*(\xi)\big], \qquad
\lambda(\xi):=\gamma-\epsilon\beta\xi, \qquad 
x^*(\xi):=\frac{\epsilon\xi}{\lambda(\xi)},
\label{eq:frozen_flow} 
\end{equation}
 where $\lambda(\xi)$ is an effective relaxation rate
and $x^*(\xi)$ denotes the corresponding frozen fixed point. 

Whenever $\epsilon\beta\xi>\gamma$, the effective relaxation rate becomes negative and the
frozen fixed point turns repulsive. This is precisely the local instability responsible for
the Kesten-type heavy tails discussed later in Section~\ref{sec:kesten_subsec}.

Integrating Eq.~\eqref{eq:frozen_flow} yields \(
x(t)=x^*+[x(0)-x^*]e^{-\lambda t} \). Sampling the process at successive renewal times
$t_k$, the evolution during the $k$-th epoch, characterized by duration $\theta_k$, noise
value $\xi_k$, and initial condition $x_k$, is described exactly by 
\begin{equation} 
x_{k+1}
= x^*(\xi_k) + \big[x_k-x^*(\xi_k)\big]e^{-\lambda(\xi_k)\theta_k} = A_k x_k + B_k,
\end{equation} 
with 
\begin{equation} 
A_k = e^{-\lambda(\xi_k)\theta_k}, \qquad B_k =
x^*(\xi_k)\big[1-e^{-\lambda(\xi_k)\theta_k}\big]. 
\label{eq:affine_recursion}
\end{equation} 
The pairs $(\xi_k,\theta_k)$ are independently drawn from $p(\xi)$ and
$\psi(\theta)$, respectively. 

The continuous stochastic dynamics has thus been reduced exactly to the iteration of random affine maps.

Equation~\eqref{eq:affine_recursion} belongs to the class of random affine recursions originally studied by Kesten~\cite{Kesten1973}. Although $A_k$ and $B_k$ are statistically dependent within a given epoch, being generated by the same pair $(\xi_k,\theta_k)$, successive pairs $(A_k,B_k)$ are independent and identically distributed. Consequently, the standard theory of random affine maps applies. In particular, under the usual integrability conditions, a unique invariant measure exists whenever $\mathbb E[\log|A|]<0.$

 Equations~\eqref{eq:affine_recursion} and~\eqref{eq:frozen_series} are the pathwise and distributional manifestations of the same renewal structure. Both arise directly from the piecewise-constant nature of the noise. The affine recursion propagates a single realization exactly from renewal event to renewal event, whereas Eq.~\eqref{eq:frozen_series} describes the evolution of the corresponding probability law after averaging over the renewal ensemble.  Unlike the former,
Eq.~\eqref{eq:frozen_series} contains the survival factor $\Psi(\theta_{N+1})$ associated with
the incomplete final epoch. Apart from this term, neither representation requires
assumptions on $\psi(\theta)$ beyond the renewal property itself.
 
The affine structure is therefore a special consequence of the linear LIMI dynamics. For nonlinear drifts or nonlinear coupling functions, the renewal mapping generally becomes a non-affine iterated function system. The extent to which the present construction survives beyond the affine case is discussed in Section~\ref{sec:beyond_cam}.

\subsubsection{Compact Support Under Bounded Jumps}

When $\lambda(\xi)>0$ throughout the support of $p(\xi)$, corresponding to the uniform
stability condition $\epsilon\beta\max|\xi|<\gamma$ of Eq.~\eqref{eq:uniform_criterion},
the multiplier satisfies $A_k\in(0,1)$ for every realization and the affine recursion is
uniformly contractive. Every frozen map therefore sends the interval bounded by the extrema
of the fixed points $x^*(\xi)$ into itself.

As a consequence, the stationary measure is confined to the compact interval

\begin{equation} \mathrm{supp}\,P_{\rm st} = \big[\min_\xi x^*(\xi),\,\max_\xi
x^*(\xi)\big]. \label{eq:compact_support} \end{equation}

This support depends only on the frozen fixed points and therefore only on the range of the
jump distribution. It is not determined by the memory kernel and cannot be inferred from any
finite collection of cumulants.

For dichotomous jumps with parameters $\gamma=\tau=1$, $\epsilon=0.5$, $\beta=0.3$, and
$a=1$, Eq.~\eqref{eq:compact_support} gives

\begin{equation} \left[ -\frac{\epsilon a}{\gamma+\epsilon\beta a}, \, \frac{\epsilon
a}{\gamma-\epsilon\beta a} \right] = [-0.4348,\,0.5882]. \end{equation}

The existence of a compact support provides a concrete illustration of the distinction
emphasized throughout Sections~\ref{sec:fpe_emergence} and \ref{sec:two_routes}. The support
is a property of the frozen-noise dynamics, not of the reduced kernel. Consequently,
reducing the closure parameter $\lambda$ improves rates and low-order moments but leaves the
support unchanged.

This behaviour differs fundamentally from the Gaussian white-noise limit, where the
stationary density develops algebraic tails extending over the whole real axis. Under step
renewal noise with bounded jumps, the invariant measure remains confined to a finite
interval and any singular or algebraic behaviour can only occur near its boundaries.

From this perspective, the heavy tails encountered in Gaussian LIMI/CAM models should not be
attributed solely to the multiplicative coupling. They arise because the driving process
itself has unbounded support. Multiplicative amplification controls how the stationary
measure is reshaped, whereas the existence of an infinite support ultimately originates from
the possibility of arbitrarily large noise realizations.

\subsubsection{Analytical Singularities and Edge Exponents}

For discrete jump distributions, the frozen-noise representation allows not only the support but also the singular structure of the stationary density to be obtained analytically. The resulting exponents are determined entirely by local properties of the frozen dynamics. 

Let $p(\xi)$ be supported on the
discrete set $\{\xi_i\}$ with associated probabilities $\{w_i\}$. Denoting by
$P_i(x):=P_{\rm st}(x,\xi=\xi_i)$ the stationary joint density conditioned on
the noise state $\xi_i$, the stationary marginal density is given by
$P_{\rm st}(x)=\sum_i P_i(x)$.

For exponentially distributed waiting times, the process $(x,\xi)$ is
Markovian. The functions $P_i(x)$ therefore satisfy the stationary forward
Kolmogorov equations
\begin{equation}
\frac{\mathrm d}{\mathrm dx}
\big[v_i(x)P_i(x)\big]
=
\frac1\tau
\Big[
w_i\sum_j P_j(x)-P_i(x)
\Big],
\qquad
v_i(x)
=
-\lambda_i\big(x-x_i^*\big),
\label{eq:stationary_system}
\end{equation}
which balance deterministic transport against stochastic transitions between
noise states.
 Equation~\eqref{eq:stationary_system} reveals that singularities can arise only at the frozen fixed points, where the deterministic velocity $v_i(x)$ vanishes.

Summing Eq.~\eqref{eq:stationary_system} over $i$ eliminates the
transition terms and yields
$\mathrm d/\mathrm dx\sum_i v_iP_i=0$. Imposing a zero-flux boundary condition
then gives

\begin{equation}
\sum_i v_i(x)P_i(x)=0.
\end{equation}

Near a fixed point $x_i^*$ the velocity field vanishes linearly, $v_i(x)\simeq-\lambda_i(x-x_i^*)$, causing the $i$-th component  $P_i(x)$ to dominate the balance in Eq.~\eqref{eq:stationary_system}. Matching the leading singular terms then gives
\begin{equation}
\boxed{
\sigma_i
=
\frac{1-w_i}{\lambda_i\tau}
-1
}
\label{eq:edge_exponents}
\end{equation}
for the local exponent governing the behavior of $P_{\rm st}(x)$ near
$x_i^*$.
The exponent depends only on two quantities: 
the local contraction rate $\lambda_i$ and the probability $w_i$ of remaining in the corresponding
noise state. 
Neither the global structure of the density nor the behavior of the other branches enters explicitly.

Equation~\eqref{eq:edge_exponents} applies to every fixed point $x_i^*$,
including those located in the interior of the support. At interior points the
singularity is two-sided, whereas at the outer boundaries the density is
defined only on one side. The singularity corresponds to an integrable
divergence whenever

\begin{equation}
\lambda_i\tau>1-w_i.
\label{eq:divergence_criterion}
\end{equation}

This criterion admits a simple interpretation. The quantity $\lambda_i\tau$
measures the typical relaxation achieved during a single epoch, whereas
$1-w_i$ is the probability of leaving state $i$ at the next renewal event.
Whenever relaxation dominates escape, probability accumulates near the fixed
point, generating a local divergence.

For the three-state distribution
$\xi\in\{-\sqrt2,0,\sqrt2\}$ with weights
$\{1/4,1/2,1/4\}$ and parameters
$\gamma=\tau=1$, $\epsilon=0.5$, and $\beta=0.4$,
Eq.~\eqref{eq:edge_exponents} predicts the value
$\sigma_++1=1.0458$ at the upper boundary. Numerical fits of the cumulative
distribution over shrinking intervals $d<3\times10^{-2}$ and
$d<10^{-2}$ yield 1.0427 and 1.0391, respectively, in good agreement with the
theoretical prediction.
 The agreement confirms that the singular structure of the stationary density is entirely encoded in the frozen dynamics and can be predicted analytically without solving the full invariant measure.

\subsubsection{Closed-Form Solutions and Multi-State Extensions}

For dichotomous noise ($N=2$), the stationary system \eqref{eq:stationary_system} can be
solved explicitly. Imposing the zero-flux condition $v_+P_+ = -v_-P_-$, and integrating
the resulting first-order equation yields
\begin{equation} 
P_{\rm st}(x) \propto (1+\beta x) \big(x^*_+ -
x\big)^{\frac{1}{2\tau\lambda_+}-1} \big(x - x^*_-\big)^{\frac{1}{2\tau\lambda_-}-1}, \qquad
\lambda_\pm = \gamma\mp\epsilon\beta a, 
\label{eq:beta_pdf} 
\end{equation}
namely a modified Beta distribution supported on $[x^*_-,x^*_+]$. The prefactor $(1+\beta
x)$ originates from the multiplicative coupling and reduces to unity in the additive-noise
limit. Equation~\eqref{eq:beta_pdf} therefore provides the exact stationary density of the
dichotomous LIMI model.

The special role of the dichotomous case mirrors the exact closure theorem of
Section~\ref{sec:when_exact}. For $N=2$, the stationary density is available in closed form
and the Boolean hierarchy terminates identically at second order. The same Bernoulli law is
therefore distinguished from both the dynamical and the stationary-distribution viewpoints.

For $N\ge3$, the situation becomes more intricate. Equation \eqref{eq:stationary_system}
defines an $N$-dimensional Fuchsian system with regular singular points at the frozen fixed
points $x_i^*$ and at infinity. The singularity structure remains completely explicit, but
the global solution is generally no longer elementary. In particular, the case $N=3$ reduces
to an equation of Heun type, while a continuous jump distribution $p(\xi)$ leads to a
Fredholm integral equation.

Despite the loss of a closed-form density, the geometric information obtained in the
previous subsections survives unchanged. The support formula \eqref{eq:compact_support} and
the exponent formula \eqref{eq:edge_exponents} remain exact for arbitrary discrete jump laws
and continue to characterize the support and singular structure of the stationary measure.

The three-state distribution considered above provides the simplest non-trivial example
beyond the Bernoulli case. It is the smallest discrete law for which the fourth Boolean
cumulant does not vanish, yielding $b_4=1$. It therefore constitutes the simplest setting in which
the effects of non-closure can be investigated directly at the level of the stationary
density, providing a natural bridge between the exactly solvable Bernoulli case and the
generic multi-state situation.

\subsubsection{Power-Law Tails and the Kesten Exponent
\label{sec:kesten_subsec}}

The compact-support scenario of the previous subsection relies on the uniform stability condition
$\lambda(\xi)>0$ throughout the support of the jump distribution. Once this condition is violated,
the stationary measure changes qualitatively.

Whenever the noise amplitude is sufficiently large that
$\lambda(\xi)<0$, i.e.\
$\epsilon\beta\xi>\gamma$, the corresponding frozen fixed point becomes
repulsive. In this regime, the random multiplier $A_k$ exceeds unity with
non-zero probability, and the affine recursion is no longer uniformly
contractive.

Under these conditions, Kesten's theorem~\cite{Kesten1973} implies that the
stationary density develops a power-law tail of the form
$P_{\rm st}(x)\sim x^{-1-\kappa}$, where $\kappa>0$ is the unique positive
solution of

\begin{equation}
\mathbb{E}[A^\kappa]=1.
\end{equation}

For exponentially distributed waiting times,

\begin{equation}
\mathbb{E}[A^\kappa\mid\xi]
=
\frac{1}{1+\kappa\lambda(\xi)\tau},
\end{equation}

which leads to the characteristic equation

\begin{empheq}[box=\fbox]{equation}
\left\langle
\frac{1}{1+\kappa\,\lambda(\xi)\,\tau}
\right\rangle_\xi
=
1,
\qquad
\lambda(\xi)
=
\gamma-\epsilon\beta\xi.
\label{eq:kesten}
\end{empheq}
Equation~\eqref{eq:kesten} shows that the tail exponent is determined entirely by the frozen dynamics. The memory kernel no longer enters explicitly; only the distribution of the local contraction rates $\lambda(\xi)$ matters. A positive solution exists if and only if the system admits locally unstable epochs, namely if $\lambda(\xi)<0$ on a set of non-zero probability.  Indeed, the left-hand side diverges as
$\kappa\rightarrow 1/(\tau|\lambda_{\min}|)$ whenever
$\lambda_{\min}<0$.

Table~\ref{tab:kesten_comparison} compares the theoretical prediction from
Eq.~\eqref{eq:kesten} with numerical estimates obtained using Hill's method for
dichotomous noise with $a=1$ and $\gamma=\tau=1$.

\begin{table}[h!]
\centering
\renewcommand{\arraystretch}{1.2}
\begin{tabular}{cccc}
\hline
$\beta$ & $\epsilon$ &
Analytical exponent $\kappa$ from Eq.~\eqref{eq:kesten} &
Hill estimate \\
\hline
$1.5$ & $1.0$ & $0.800$ & $0.823$ \\
$2.0$ & $1.0$ & $0.333$ & $0.329$ \\
$1.5$ & $1.2$ & $0.446$ & $0.433$ \\
\hline
\end{tabular}
\caption{Comparison between the analytical Kesten exponent obtained from
Eq.~\eqref{eq:kesten} and numerical Hill estimates.}
\label{tab:kesten_comparison}
\end{table}

The agreement confirms that the tail exponent is controlled by the affine multipliers alone and is accurately predicted by the Kesten condition \eqref{eq:kesten}.

Equation~\eqref{eq:kesten} applies to both discrete and continuous jump distributions $p(\xi)$.
The condition $\epsilon\beta\max\xi>\gamma$ coincides with the loss of uniform stability expressed
by Eq.~\eqref{eq:uniform_criterion} and is therefore equivalent to three apparently different
statements: loss of uniform contractivity of the affine recursion, existence of a positive
Kesten exponent, and the appearance of an imaginary component in the Boolean generator $\eta(u)$
discussed in Section~\ref{sec:heavy}. These are simply different manifestations of the same
underlying dynamical instability.
These equivalent criteria identify the onset of the
heavy-tailed regime and the breakdown of the globally stable fixed-point
structure.

If $p(\xi)$ has unbounded support, the event $\lambda(\xi)<0$ occurs with
non-zero probability independently of the asymptotic form of the tails.
Consequently, compact support is lost and the system generically enters the
Kesten regime. In that case, the tail exponent is determined from the more
general condition
\begin{equation}
\mathbb{E}
\!\left[
e^{-\kappa\lambda(\xi)\theta}
\right]
=
1,
\end{equation}
provided the expectation exists. 
This requirement is not automatic. For sufficiently broad jump distributions, the existence of the Kesten exponent becomes a question of integrability of the joint renewal measure rather than of the affine recursion alone.
The onset of heavy tails is therefore controlled by the frozen dynamics, whereas their precise exponent depends on the statistical properties of the renewal process through Eq.~\eqref{eq:kesten}.

\section{Numerical Validation of Stationary-Density Predictions \label{sec:numres}} This
section tests the exact analytical predictions derived in Section~\ref{sec:stationary}.
Unlike the weak-coupling closure approximations examined in Section~\ref{sec:fpe_test}, the
results of Section~\ref{sec:stationary} are non-perturbative and follow directly from the
frozen-noise representation. The objective is therefore not to assess an approximation, but
to verify the predicted support boundaries, singularities, and tail behaviour of the
stationary measure.

Two discrete jump distributions are considered in order to probe different structural
features of the stationary state. The first is a symmetric three-state distribution,
representing the simplest jump law with non-vanishing fourth Boolean cumulant ($b_4\neq0$).
The second is an asymmetric four-state distribution, included to test the generality of the
analytical predictions beyond symmetric cases.

Throughout this section, the system parameters are those used in Section~\ref{sec:atoms_cam},
\[ \gamma=\tau=1, \qquad \epsilon=0.5, \qquad \beta=0.4. \] Each trajectory consists of
$8\times10^6$ renewal steps following an equilibration stage of $2\times10^5$ steps. The
jump distributions are

\begin{equation} 
\begin{aligned} 
\text{(a)}\quad & \xi \in \{-\sqrt2,\,0,\,\sqrt2\}, &&
w=\{\tfrac14,\tfrac12,\tfrac14\}, \\ \text{(b)}\quad & \xi \in \{-1.3,-0.4,0.4,1.3\}, &&
w=\{0.2,0.3,0.3,0.2\}. 
\end{aligned} 
\end{equation}

Both distributions are centered, $\langle\xi\rangle=0.$ The corresponding Boolean
cumulants are $(b_2,b_4)=(1,1)$ for law~(a) and $(b_2,b_4)=(0.772,0.562)$ for
law~(b).

A distinction should be made between the theoretical and numerical sampling procedures. The
analytical results of Section~\ref{sec:stationary} describe the continuous-time stationary
measure, whereas the simulations naturally generate a sequence of states observed at renewal
events through the affine recursion \eqref{eq:affine_recursion}.

For exponentially distributed waiting times, these two sampling procedures are equivalent
because of the PASTA property (Poisson Arrivals See Time Averages): the distribution
observed at renewal epochs coincides with the continuous-time stationary distribution. The
numerical histograms can therefore be compared directly with the analytical predictions
derived from Eq.~\eqref{eq:stationary_system}.

This equivalence is specific to the memoryless case. For non-exponential waiting-time
distributions, renewal events are sampled with a length bias and the distribution observed
at renewal times generally differs from the continuous-time stationary measure. In that
setting, the analytical framework developed in Section~\ref{sec:stationary} would need to be
supplemented by the appropriate age-weighted sampling prescription.

\subsection{Numerical Accuracy and Protocol}

Equilibration was assessed by comparing simulations initialized at $x_0=0$
with simulations started from the extreme boundaries of the support. Both sets
of initial conditions converged to identical stationary profiles within
statistical uncertainty, indicating that the sampling procedure reached the
same invariant measure.
Error bars were estimated by partitioning each dataset into ten independent
blocks. In all figures, the resulting uncertainties are smaller than the symbol
size. Consequently, residual discrepancies between measured and predicted
exponents are dominated by systematic effects, primarily the contribution of
the regular background density within the fitting window, rather than by
sampling fluctuations.

Since the numerical scheme integrates the linear dynamics exactly within each
renewal interval, no time-discretization error is introduced. This allows the
density to be sampled arbitrarily close to singular points and permits
distances as small as $10^{-4}$ from a frozen fixed point to be resolved,
avoiding the step-size limitations of conventional numerical integrators.

Special care is required in the unstable regime considered in
Figs.~\ref{fig:stationary_exponents}(b) and
\ref{fig:stationary_exponents}(d), where the invariant measure develops a
power-law tail and the variance diverges. For these cases, trajectories consisted of
$3\times10^6$ renewal steps. The tail exponent was estimated using
Hill's method applied to the upper $2\%$ of positive values, choosing the
largest threshold for which the estimate remained stable under successive
halving. This provides a numerical counterpart to the analytical prediction of
the Kesten exponent derived in Section~\ref{sec:kesten_subsec}.

\subsection{Support Boundaries and Singular Structure}

The first prediction tested is the exact support formula
\eqref{eq:compact_support}. According to the frozen-noise theory, the support
of the invariant measure is determined solely by the extreme frozen fixed
points and is therefore independent of any kernel approximation or cumulant
truncation.

Table~\ref{tab:support_comparison} compares the theoretical prediction with the
numerical observations. Agreement is obtained to all reported significant
figures, confirming that the support is exactly given by the interval spanned
by the extreme frozen fixed points.
\begin{table}[h!]
\centering
\renewcommand{\arraystretch}{1.2}
\begin{tabular}{lcc}
\hline
Jump Law & Predicted Support Bounds & Observed Numerical Support \\
\hline
Three-valued (a) & $[-0.5512,\ 0.9860]$ & $[-0.5512,\ 0.9860]$ \\
Four-valued (b)  & $[-0.5159,\ 0.8784]$ & $[-0.5159,\ 0.8784]$ \\
\hline
\end{tabular}
\caption{Comparison between the support predicted by
Eq.~\eqref{eq:compact_support} and the numerical observations.}
\label{tab:support_comparison}
\end{table}

Figure~\ref{fig:stationary_density} compares the simulated stationary
densities with the asymptotic scaling laws
$P_{\rm st}(x)\sim |x-x_i^*|^{\sigma_i}$
predicted by Eq.~\eqref{eq:edge_exponents}. The results confirm that every
atom of the jump distribution generates a corresponding singular point at its
frozen fixed point $x_i^*$.

Singularities associated with interior fixed points are two-sided, whereas
those located at the boundaries are one-sided. The three-state distribution
produces a central singularity at $x=0$, associated with the noise state
$\xi=0$, while the four-state distribution generates two distinct interior
singularities.

These observations illustrate an important feature of the frozen-noise
description: the singular structure of the stationary density reflects the
geometry of the jump distribution itself. Adding noise states produces new
fixed points and therefore new singularities, in agreement with the atom
hierarchy discussed in Section~\ref{sec:atoms_cam}.

The selection criterion for integrable divergences,
Eq.~\eqref{eq:divergence_criterion}, is verified as well. A divergent peak
($\sigma_i<0$) occurs only when
$\lambda_i\tau>1-w_i$.
For both jump laws, the uppermost fixed point does not satisfy this condition,
yielding positive exponents,
$\sigma=0.046$ for law~(a) and
$\sigma=0.081$ for law~(b).
The density therefore vanishes at the upper boundary rather than diverging.

The results demonstrate that the number of singular points is determined by the
number of discrete noise states, whereas the number of divergent peaks is
controlled by the local balance between relaxation and escape embodied in
Eq.~\eqref{eq:divergence_criterion}.

\subsection{Validation of Singular Exponents}

The support boundaries and locations of the singular points are fixed by the
frozen fixed points. The remaining prediction of the theory concerns the local
exponents governing the behavior of the density near those points,
Eq.~\eqref{eq:edge_exponents}.

Exponents were extracted by fitting the cumulative distribution function over
the interval $3\times10^{-3}<d<3\times10^{-2}$ around each divergent fixed
point. The results are summarized in
Table~\ref{tab:exponents_comparison}.

\begin{table}[h!]
\centering
\renewcommand{\arraystretch}{1.2}
\begin{tabular}{llccc}
\hline
Jump Law & Fixed Point $x^*$ &
Predicted $\sigma+1$ &
Measured $\sigma+1$ &
Relative Error \\
\hline
Three-valued & $-0.5512$ (boundary)
& $0.5846$ & $0.5928$ & $1.4\%$ \\
             & $\ \ 0.0000$ (interior)
& $0.5000$ & $0.4973$ & $0.5\%$ \\
Four-valued  & $-0.5159$ (boundary)
& $0.6349$ & $0.6452$ & $1.6\%$ \\
             & $-0.1852$ (interior)
& $0.6481$ & $0.6618$ & $2.1\%$ \\
             & $+0.2174$ (interior)
& $0.7609$ & $0.7332$ & $3.6\%$ \\
\hline
\end{tabular}
\caption{Comparison between the analytical prediction
Eq.~\eqref{eq:edge_exponents} and numerical estimates obtained from cumulative
distribution fits.}
\label{tab:exponents_comparison}
\end{table}

Agreement is obtained at the percent level for every singularity considered,
including both boundary and interior fixed points. The largest discrepancy
remains below $4\%$, despite the finite fitting window and the presence of a
regular background contribution to the density. The results therefore confirm
that Eq.~\eqref{eq:edge_exponents} captures the local structure of the
stationary measure with high accuracy.

The interior fixed point of the three-state distribution provides a
particularly stringent test of the theory. For $\xi_2=0$, one has
$\lambda_2=\gamma$ and $w_2=1/2$, yielding the parameter-free prediction
$\sigma_2=-\frac12$
when $\gamma\tau=1$.
Estimates obtained independently from the left and right sides of the
singularity, and over progressively shrinking fitting windows, cluster around
the predicted value. This confirms both the exponent itself and the expected
left-right symmetry of the interior singularity.

Taken together with the support verification reported in
Table~\ref{tab:support_comparison}, these measurements validate the complete
local description of the stationary density provided by the frozen-noise
theory. The support boundaries, the locations of the singular points, and the
associated singular exponents are all correctly predicted from the frozen
fixed-point structure.

\subsection{Power-Law Tails in the Unstable Regime}

The final prediction of Section~\ref{sec:kesten_subsec} concerns the unstable regime, where the loss of uniform contractivity leads to Kesten-type power-law tails. When $\epsilon\beta\max|\xi|>\gamma$, the stationary distribution is no longer compactly concentrated around stable frozen fixed points and develops an algebraic tail.
Equation~\eqref{eq:kesten} predicts a tail exponent $\kappa$ satisfying

\[
\left\langle
\frac{1}{1+\kappa\lambda(\xi)\tau}
\right\rangle_\xi
=
1.
\]

Table~\ref{tab:kesten_comparison_num} compares the corresponding analytical
predictions with Hill estimates obtained from the upper $2\%$ of sampled
positive values.

\begin{table}[h!]
\centering
\renewcommand{\arraystretch}{1.2}
\begin{tabular}{llccc}
\hline
Jump Law & $(\beta,\epsilon)$ &
$\lambda_{\min}$ &
Analytical $\kappa$ &
Hill Estimate \\
\hline
Three-valued & $(1.2,1.0)$ & $-0.697$ & $0.8933$ & $0.8847$ \\
             & $(1.5,1.0)$ & $-1.121$ & $0.5000$ & $0.5002$ \\
             & $(2.0,1.0)$ & $-1.828$ & $0.2612$ & $0.2581$ \\
             & $(1.5,1.2)$ & $-1.546$ & $0.3288$ & $0.3310$ \\
Four-valued  & $(1.2,1.0)$ & $-0.560$ & $1.2488$ & $1.1955$ \\
             & $(1.5,1.0)$ & $-0.950$ & $0.6710$ & $0.6677$ \\
             & $(1.8,1.0)$ & $-1.340$ & $0.4349$ & $0.4328$ \\
             & $(1.5,1.3)$ & $-1.535$ & $0.3636$ & $0.3631$ \\
\hline
\end{tabular}
\caption{Comparison between the Kesten exponent predicted by
Eq.~\eqref{eq:kesten} and numerical Hill estimates.}
\label{tab:kesten_comparison_num}
\end{table}

Agreement is typically better than $1.2\%$ despite the substantial variation of $\kappa$ across the parameter range considered. The largest discrepancy, about $4.3\%$, occurs for the largest exponent ($\kappa\simeq1.25$), corresponding to the lightest tail. This behaviour is consistent with the well-known finite-sample bias of Hill estimators when the asymptotic tail region occupies only a small fraction of the sampled distribution.
 Representative comparisons are shown in
Figs.~\ref{fig:stationary_exponents}(b) and
\ref{fig:stationary_exponents}(d).

Taken together, the numerical evidence validates the entire stationary-density
theory developed in Section~\ref{sec:stationary}. Support boundaries agree with
Eq.~\eqref{eq:compact_support}, singular points obey the exponent formula
\eqref{eq:edge_exponents}, and power-law tails are governed by the Kesten
condition \eqref{eq:kesten}. The frozen-noise approach therefore captures both
the local and global structure of the invariant measure.

Equally importantly, the results clarify the role of higher Boolean cumulants.
Their effect is structural rather than perturbative: they modify the support,
the number and location of singular points, and ultimately the tail properties
of the stationary density. These changes cannot be reproduced by adding
cumulant orders within the two-state closure, because the latter remains the
exact stationary measure of the Bernoulli process associated with
$b_4=b_6=\cdots=0$.

The numerical evidence thus illustrates in concrete terms the distinction
established in Section~\ref{sec:two_routes}. Boolean-cumulant expansions
provide a systematic hierarchy in kernel space and accurately describe
relaxation rates, low-order moments, and closure errors. The geometry of the
stationary density, by contrast, is controlled by the frozen-noise dynamics and
by the structure of the jump distribution itself. Support boundaries,
singularities, and Kesten tails are therefore naturally analyzed through the
frozen-noise representation rather than through truncations of the kernel
expansion.

\begin{figure}[t]
\centering
\includegraphics[width=\textwidth]{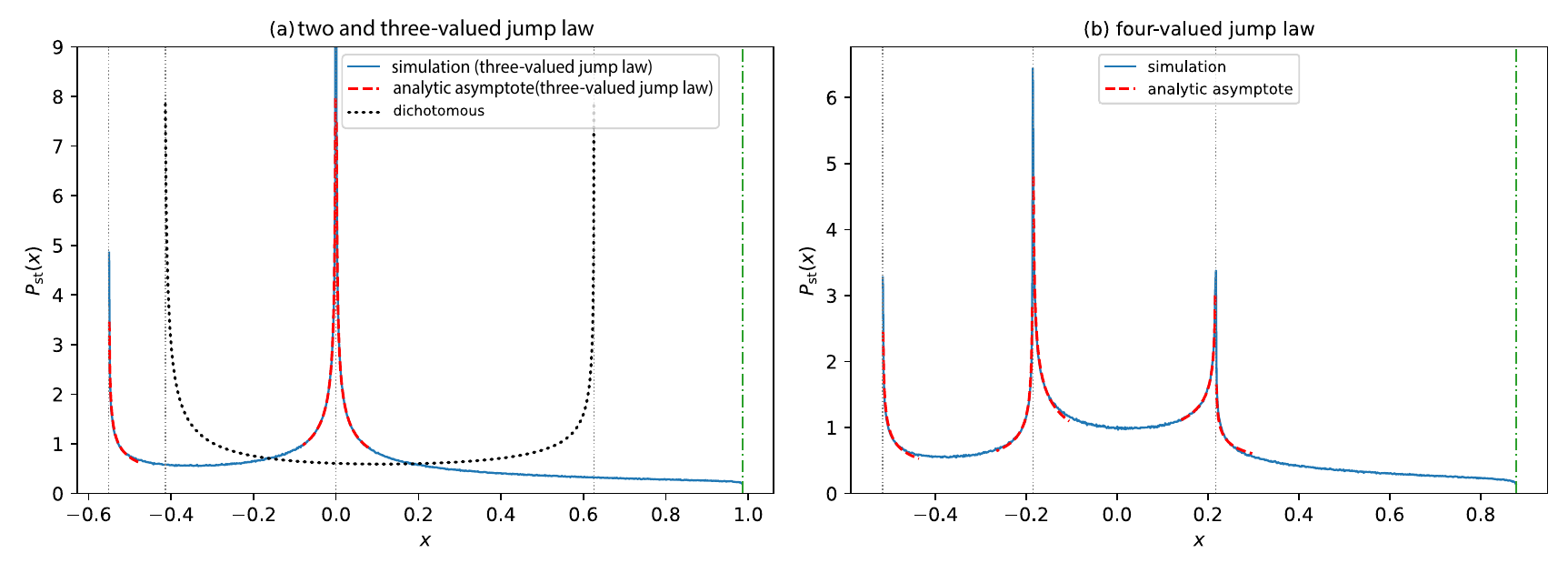}
\caption{Stationary density of the LIMI model \eqref{eq:worked_sde} driven by step renewal noise ($\gamma=\tau=1$, $\epsilon=0.5$, $\beta=0.4$, $8\times10^{6}$ renewal steps). 
(a) Three-valued jump distribution $\xi \in \{-\sqrt2, 0, \sqrt2\}$ with weights $\{1/4, 1/2, 1/4\}$. Frozen fixed points lie at $x^* = -0.5512, 0, 0.9860$ with exponents $\sigma = -0.415, -0.500, +0.046$. 
(b) Four-valued distribution $\xi \in \{-1.3, -0.4, 0.4, 1.3\}$ with weights $\{0.2, 0.3, 0.3, 0.2\}$. Fixed points lie at $x^* = -0.5159, -0.1852, 0.2174, 0.8784$ with exponents $\sigma = -0.365, -0.352, -0.239, +0.081$. 
Dotted vertical lines indicate divergent points; dash-dotted lines mark non-divergent boundary points ($\sigma > 0$). Red dashed curves show the asymptotic scalings $|x-x^*_i|^{\sigma_i}$ from Eq.~\eqref{eq:edge_exponents} with fitted prefactors. The dotted curve in (a) depicts the dichotomous baseline solution, Eq.~\eqref{eq:beta_pdf}.}
\label{fig:stationary_density}
\end{figure}

\begin{figure}[t]
\centering
\includegraphics[width=\textwidth]{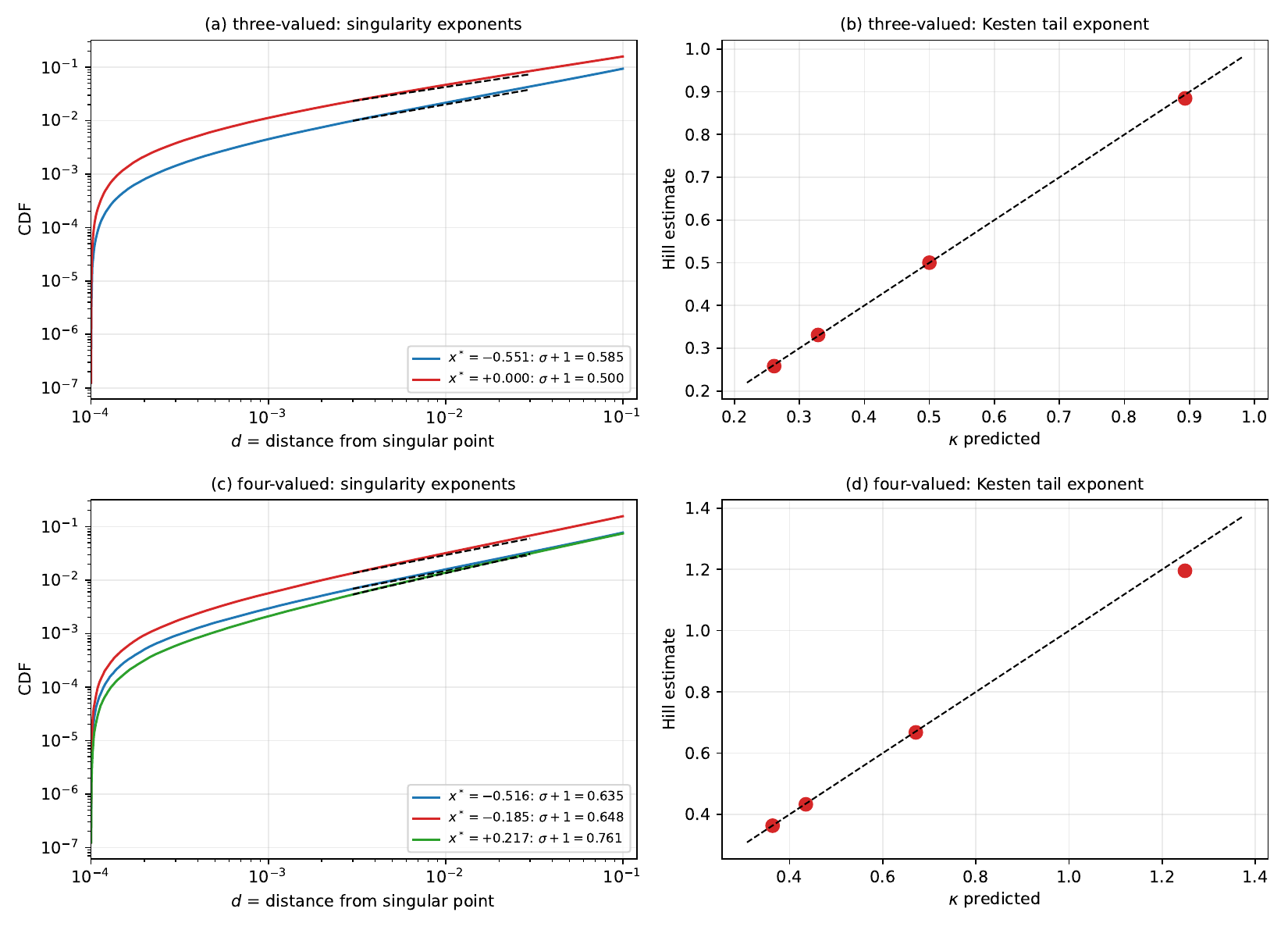}
\caption{(a),(c) Log-log plots of the cumulative distribution near divergent fixed points for the three- and four-valued jump distributions. Black dashed lines indicate the theoretical slopes $\sigma_i+1$ from Eq.~\eqref{eq:edge_exponents}. 
(b),(d) Comparison between numerical Hill estimates and the theoretical Kesten exponent $\kappa$ from Eq.~\eqref{eq:kesten} in the unstable regime ($\epsilon\beta\max|\xi| > \gamma$). Dashed lines represent exact equality. Trajectories use $3\times10^{6}$ steps, with tail indices estimated over the upper $2\%$ of positive state values.}
\label{fig:stationary_exponents}
\end{figure}

\section{Scope and Generalizability Beyond Linear Dynamics
\label{sec:beyond_cam}}

The results derived in the previous sections were obtained for the LIMI/CAM
model, characterized by the linear drift $C(x)=\gamma x$ and the linear
coupling $I(x)=1+\beta x$. It is therefore natural to ask which conclusions are
structural and survive beyond the linear setting, and which instead rely on the
special affine form of the frozen dynamics.

The distinction developed throughout this work between kernel-space and
density-space descriptions is particularly useful in this regard. Some results
follow directly from the renewal structure of the driving process and therefore
remain valid for arbitrary drift and coupling functions. Others exploit the
fact that the LIMI model reduces, under frozen noise, to a random affine
recursion and are correspondingly more restrictive.

\subsection{Model-Independent Properties}

Several results depend exclusively on the statistical structure of the renewal
noise and therefore remain valid for arbitrary drift and coupling functions
$C(x)$ and $I(x)$.

\begin{enumerate}

\item \emph{Exactness of Dichotomous Noise.}

Theorem~\ref{thm:uniqueness} establishes that the memory-kernel expansion
terminates at second order if and only if the jump distribution is a symmetric
Bernoulli law. This property is determined entirely by the Boolean generator
$\eta(u)$ of the driving process and is therefore independent of the drift, the
coupling, and the associated spatial operators.

\item \emph{Operator Error Structure.}

The exact expression for the leading discarded contribution,
\[
\mathcal D=\frac{b_4}{b_2}\mathcal X^2,
\]
holds for arbitrary systems. The separation between the statistical factor
$b_4/b_2$ and the dynamical factor $\mathcal X^2$ is therefore a structural feature of
the renewal kernel itself rather than a consequence of the LIMI model. In
particular, the ratio $b_4/b_2$ remains a model-independent measure of the
relative closure error.

\item \emph{Fokker--Planck Structure.}

For general drift and coupling functions $C(x)$ and $I(x)$, the interaction
operator $\tilde{\mathcal L}_I$ remains a first-order differential operator
through the Lie-derivative representation~\cite{bJMP59}.
Consequently, localization of the second-cumulant kernel again produces a
second-order differential operator, yielding a well-defined Fokker--Planck
equation even for nonlinear systems.

What changes from one model to another is not the existence of the reduction
itself, but the quantitative accuracy with which it reproduces the underlying
dynamics. The distinction between the weak-coupling parameter $\lambda$ and the
density-control parameter $\gamma\tau$ must therefore be regarded as a
structural consequence of the renewal framework rather than as a peculiarity of
the LIMI model.

\end{enumerate}

\subsection{Conditional Generalizations}

Some geometric properties of the stationary measure extend far beyond linear
systems, but only under additional dynamical assumptions.

The compact-support property,
Eq.~\eqref{eq:compact_support}, provides a representative example. Consider a
general frozen dynamics
$\dot x=-C(x)+\epsilon\xi I(x)$.
Whenever every frozen realization possesses a stable attracting fixed point and
all such fixed points are bounded, trajectories remain confined to the region
spanned by their extrema. In this sense, compact support is not a consequence
of linearity but of the existence of a bounded family of attracting frozen
states.

As an illustration, consider the bistable system
\[\dot x
=
1.2x-x^3+\epsilon\xi,\]
driven by additive uniform noise. For $\epsilon=0.6$, numerical integration
using a fourth-order Runge--Kutta scheme yields the support
$[-1.3947,\,1.3947]$, in agreement with the extrema of the frozen fixed points
to four significant figures.

The situation becomes more subtle in multistable systems. For the same model,
but with a smaller noise amplitude ($\epsilon=0.3$), the dynamics remains
confined within a single potential well and the observed support is restricted
to approximately $[-1.269,\,-0.640]$. Although the frozen fixed points still
define a larger admissible region, the associated invariant measure explores
only one connected component of phase space because barrier-crossing events are
effectively absent.

The essential requirement is therefore not merely the existence of frozen
attractors, but also sufficient exploration of the dynamically accessible
phase space. For multistable systems,
Eq.~\eqref{eq:compact_support} characterizes the support of the full invariant
measure only when the renewal forcing generates transitions among all relevant
basins on the time scale of observation. Otherwise, the same theory applies
separately within each dynamically isolated basin.

This distinction is typical of nonlinear systems. The geometric predictions of
the frozen-noise approach remain valid, but their realization depends on the
ergodic properties of the dynamics and on the ability of the noise to connect
the different attracting regions of phase space.

\subsection{Limitations of the Linear Theory}

Several quantitative results derived for the LIMI model rely directly on the
linearity of the frozen dynamics and therefore do not extend unchanged to
nonlinear systems.

\begin{itemize}

\item \emph{Drift-Independent Relaxation Rates.}

The invariance of the rate correction $\mu$
[Eq.~\eqref{eq:mu_indep}] is a consequence of the fact that
$\mathcal L_a$ acts proportionally to the identity on the one-dimensional
subspace associated with $\langle x\rangle$. For nonlinear drifts, the local
relaxation rate $|C'(x)|$ becomes state dependent. In general, neither a
single parameter $\gamma$ nor a globally defined relaxation exponent can then
be introduced, and the cancellation mechanism responsible for the
$\gamma$-independence of the LIMI result need no longer hold.

\item \emph{Affine Recursions and Kesten Tails.}

For the LIMI model, sampling at renewal times reduces the dynamics exactly to
the affine recursion \eqref{eq:affine_recursion}. Nonlinear drift or coupling
functions generally destroy this structure, replacing it with a nonlinear
iterated function system. The frozen-noise description itself remains valid,
but the classical Kesten theory no longer applies directly and explicit
expressions such as Eq.~\eqref{eq:kesten} are not generally available.

\item \emph{Edge Exponents and Closed-Form Densities.}

The modified Beta density \eqref{eq:beta_pdf} and the exponent formula
\eqref{eq:edge_exponents} ultimately derive from the linearity of the frozen
velocity field near each fixed point. In nonlinear systems, local
linearization often provides accurate asymptotic approximations, but exact
closed-form expressions are generally not expected. What survives is the
geometric interpretation: singularities remain associated with frozen
attractors and their local stability properties.

\end{itemize}

Despite these limitations, it is important to note that they concern the
availability of explicit formulas rather than the validity of the underlying
renewal framework. The frozen-noise representation and the resummed kernel
remain exact for arbitrary $C(x)$ and $I(x)$; what is generally lost is the
ability to reduce the resulting dynamics to analytically solvable forms.

\subsection{Bifurcations and Barrier Crossings}

Nonlinear drift fields introduce phenomena that are entirely absent in the
linear LIMI model, most notably bifurcations of the frozen dynamics itself.
These effects provide a natural arena in which the frozen-noise
representation becomes more informative than perturbative approaches.

Consider again the bistable system $\dot x=1.2x-x^3+\epsilon\xi$, whose frozen
fixed points solve
\[
x^3-1.2x=\epsilon\xi .
\]

For a fixed realization of the noise, the number and nature of the stationary
points depend on the value of $\epsilon\xi$. The frozen dynamics possesses
three fixed points whenever
\[
|\epsilon\xi|
<
2(0.4)^{3/2}
\simeq 0.506 ,
\]
whereas only a single fixed point survives beyond this threshold.

For sufficiently large noise amplitudes, such as
$\epsilon\max|\xi|=1.04$, certain realizations remove one of the potential
wells altogether. During such epochs, the deterministic dynamics no longer has
to cross a barrier; instead, the barrier temporarily ceases to exist and the
trajectory relaxes toward the remaining attractor. The resulting transition
mechanism is therefore qualitatively different from conventional
thermally activated escape, which relies on fluctuations overcoming a
persistent barrier.

This observation illustrates a broader advantage of the frozen-noise
perspective. Because the dynamics is analyzed through the family of frozen
flows, changes in the topology of the phase portrait, such as the creation or
annihilation of fixed points, become directly visible at the level of the
renewal epochs themselves.

Moreover, Theorem~\ref{thm:uniqueness} shows that dichotomous renewal noise
admits an exact second-order closure without requiring any weak-noise
assumption. The present framework therefore provides a natural starting point
for the study of strongly non-perturbative barrier-crossing phenomena,
noise-induced bifurcations, and rare-event transitions in nonlinear systems,
where traditional weak-coupling expansions are often difficult to justify.

\section{Conclusions
\label{sec:conclusions}}

This work identifies two exact and complementary descriptions of dynamics
driven by memoryless step-renewal noise. The first is a resummed kernel
representation, naturally adapted to relaxation rates, spectral observables,
moment hierarchies, and closure estimates. The second is a frozen-noise
representation, naturally adapted to stationary densities, support
boundaries, singular structures, and tail behaviour. Together, these two
descriptions provide a unified framework for finite-correlation stochastic
forcing beyond the conventional Gaussian paradigm.

\paragraph{Boolean structure of renewal dynamics.}

The central theoretical result is that memoryless renewal noise is naturally
organized by Boolean rather than classical cumulants. For exponentially
distributed waiting times, the totally time-ordered ($G$-)cumulants coincide
exactly with the Boolean cumulants of the jump distribution, establishing an
explicit connection between renewal processes and Boolean probability theory.
This identification yields an exact resummation of the memory kernel in terms
of the Boolean generator $\eta$, providing a closed spectral description of
the reduced dynamics for arbitrary drift and state-dependent coupling.

The appearance of Boolean cumulants is ultimately rooted in the memoryless
property of the exponential waiting-time distribution. The factorization
\[
\prod_i \phi(d_i)
=
\phi\!\left(\sum_i d_i\right)
\]
is precisely the mechanism that converts interval partitions into Boolean
cumulants and gives rise to the exact kernel structure derived throughout
this work. For non-exponential waiting times this factorization is lost: the
interval-partition support survives and the Boolean cumulants still appear as
block coefficients, but the weights of the partitions no longer reduce to a
single-epoch term (Section~\ref{sec:nonmarkov}).

\paragraph{Closure, Bernoulli noise, and the kernel hierarchy.}

Within the Boolean hierarchy, the symmetric Bernoulli distribution occupies a
distinguished position. Theorem~\ref{thm:uniqueness} shows that the hierarchy
terminates at second order if and only if the jump distribution is symmetric
Bernoulli, yielding the extended telegraph equation
\eqref{MEG_2} without requiring weak coupling, short correlation times, or
other perturbative assumptions.

In this precise algebraic sense, Bernoulli noise plays within the Boolean
kernel hierarchy the same role that Gaussian noise plays within the hierarchy
of classical cumulants: it is the unique distribution for which all higher
cumulants vanish and the associated hierarchy closes exactly.

The analogy should not be pushed further, however. In contrast with Edgeworth
or Gram--Charlier expansions, where cumulants directly organize successive
approximations of a probability density, the Boolean hierarchy acts on the
reduced memory kernel. Truncating the hierarchy therefore approximates the
dynamics rather than the stationary measure itself.

For general jump distributions, the leading closure error is controlled by the
universal combination
\[
\mathcal D=\frac{b_4}{b_2}\mathcal X^2,
\]
which reduces to the scalar estimate $(b_4/b_2)\lambda^2$, verified
quantitatively in the LIMI model. The error therefore factorizes naturally into a
statistical contribution determined by the jump law and a dynamical
contribution determined by the system.

\paragraph{Kernel space versus density space.}

A second major conclusion is that reduced dynamics and stationary densities
belong naturally to different approximation hierarchies.

The Boolean expansion provides a hierarchy in kernel space. It governs
relaxation rates, low-order moments, spectral observables, and closure
errors. The frozen-noise representation instead provides a hierarchy in
density space. There the natural approximation parameter is not the cumulant
order but the number of atoms retained in the jump distribution.

This distinction explains one of the central observations of the present work.
Second-order closures can reproduce rates and moments with remarkable accuracy
while simultaneously failing to reproduce support boundaries, singularities,
or tail behaviour. Such failures are not consequences of poor convergence but
of approximating the wrong object.

The Bernoulli case illustrates this point particularly clearly. The $N=2$
approximation is not the first element of a hierarchy of stationary densities
waiting to be improved by higher cumulants. It is already the exact density of
a different stochastic process, namely the dichotomous renewal process.
Improvement of stationary densities therefore comes from adding atoms rather
than from adding cumulant orders.

The two hierarchies are nevertheless closely related. Replacing the jump law
with its $N$-point Gauss quadrature preserves all correlation functions and
$G$-cumulants of order $n\le2N-1$, for any waiting-time law
(Proposition~\ref{prop:corr_preservation}), and for Poissonian renewal the
Boolean generator of the surrogate is the $[N/N-1]$ Pad\'e approximant of
$\eta$. Kernel and atom hierarchies are therefore two truncations of the same
object at the same order, polynomial and Pad\'e respectively; only the latter
is the kernel of a genuine renewal process, and the Bernoulli case is its
lowest instance.

\paragraph{Stationary measures and extreme events.}

The frozen-noise representation converts the stationary problem into one of
classical probability theory. For the LIMI model, sampling at renewal times
reduces the dynamics exactly to a random affine recursion, yielding explicit
results for support boundaries, singularity exponents, and Kesten tails.

Whenever
\[
\epsilon\beta\max|\xi|<\gamma,
\]
bounded jump distributions generate compactly supported stationary measures.
The support is determined by the extrema of the frozen fixed points, while
each atom of the jump distribution generates a corresponding singular point in
the stationary density. When the stability condition is violated, the system
enters the Kesten regime and develops power-law tails whose exponent is
determined by Eq.~\eqref{eq:kesten}. Numerical simulations confirm all of
these predictions.

These results clarify the role of finite-correlation forcing in the
generation of rare and extreme events. In contrast with white-noise models,
the existence and nature of extreme excursions depend not only on the
variance and correlation time of the forcing but also on the support and
geometry of the jump distribution itself. Bounded jumps may completely exclude
arbitrarily large fluctuations, while unbounded jumps can generate heavy
tails and divergent moments.

\paragraph{Universality, scope, and outlook.}

The exact kernel resummation, the Boolean interpretation of renewal
correlations, the Bernoulli closure theorem, and the operator error structure
are independent of the specific dynamical model considered, as is the
correlation-preservation property of atom surrogates, which moreover holds for
any waiting-time law. These results are properties of the renewal process
itself.

The explicit formulas obtained for the LIMI model rely on the additional
simplification that frozen dynamics reduce to a random affine recursion.
Nevertheless, the frozen-noise representation remains exact for arbitrary
drift and coupling functions and naturally extends to systems with multiple
attractors, barrier crossings, and noise-induced bifurcations.

Several directions remain open. For non-exponential waiting times, preliminary evidence indicates that the
per-block coefficients interpolate between Boolean cumulants and ordinary
moments, the latter being approached only for power-law waiting times; this
regime, together with aging, will be addressed in a forthcoming work. Multi-state jump distributions lead naturally to Fuchsian systems and
Heun-type equations, suggesting the possibility of additional exact stationary
solutions. More broadly, renewal processes appear to provide a natural meeting
point between projection-operator methods, random dynamical systems, and
non-commutative probability.

\paragraph{Final perspective.}

The main message of this work is that memoryless renewal dynamics admit two
different but complementary exact descriptions. One is organized by Boolean
cumulants and naturally expressed through reduced memory kernels. The other is
organized by frozen deterministic flows and naturally expressed through
stationary measures.

The first explains closure, rates, and moments. The second explains support,
singularities, and rare events. Taken together, they show that finite-
correlation reductions are controlled not only by the variance and correlation
time of the forcing, but also by the combinatorial structure of the renewal
process itself. For memoryless renewal noise, that structure is Boolean.

\section*{Acknowledgements}
This research was carried out using ISMAR--CNR institutional funds.

\appendix

\section{Boolean Cumulants: Definitions and Properties\label{app:boolean}}

This appendix summarizes the definition and properties of Boolean cumulants used in the main text. While well established in the context of non-commutative probability \cite{SpeicherWoroudi1997,NicaSpeicher2006}, these concepts are less widely known in statistical mechanics, where equivalent objects frequently appear under different names.

\subsection{Physical Applications and Contexts}

Unlike classical cumulants, which describe arbitrary sets of spatial or temporal correlations, Boolean cumulants arise naturally whenever physical quantities are generated by a sequential \emph{chain of consecutive, non-overlapping events}. 

Historically, partial cumulants were introduced by von Waldenfels \cite{vonWaldenfels1973} in the theory of spectral line broadening, and later further developed via M\"obius calculus on interval partitions to treat multi-time expectations in Poisson jump processes \cite{vonWaldenfels1975}. In these systems, interactions or spectral shifts occur sequentially in time, restricting the relevant configurations to ordered interval structures.

A second physical occurrence appears in the resummation of the Dyson equation for the full propagator, $G = G_0 + G_0\Sigma G$, or equivalently $G = G_0 / (1 - \Sigma G_0)$. This expression matches the ordinary generating function relation of the Boolean moment--cumulant transform, Eq.~\eqref{eq:app_ogf}, with the self-energy $\Sigma$ acting as the Boolean cumulant generator. As shown in Section~\ref{sec:four}, for Gaussian jump processes this correspondence reflects an exact property: Boolean cumulants directly isolate the one-particle-irreducible (1PI) content of the moments.

Furthermore, Boolean cumulants describe combinatorial properties of first-return probabilities for random walks, particularly on free-product groups \cite{FreeIntegralCalculusI}, where moments represent total returns and Boolean cumulants isolate irreducible first returns. Finally, as exploited in the present work, step-renewal noise exhibits correlation functions supported strictly on interval partitions, as observation times within a single renewal epoch form connected temporal intervals (Section~\ref{sec:renewal}).

\subsection{Interval Partitions and Combinatorics}

Let $\mathcal P(n)$ be the set of all partitions of $\{1,\dots,n\}$. A partition is defined as an \emph{interval partition} if every block consists of contiguous integers. The subset of interval partitions is denoted by $\mathcal I(n) \subset \mathcal P(n)$. 

Interval partitions correspond bijectively to compositions of an ordered sequence. Cutting the sequence at a subset $S \subseteq \{1,\dots,n-1\}$ of the $n-1$ available gaps generates a unique interval partition. Consequently, the cardinality of $\mathcal I(n)$ is given by
\begin{equation}
|\mathcal I(n)| = 2^{n-1}.
\end{equation}
By comparison, the total number of partitions $|\mathcal P(n)|$ is given by the Bell numbers, while non-crossing partitions $|NC(n)|$ are counted by the Catalan numbers. For $n=4$, the eight possible interval partitions are
\begin{equation}
\begin{aligned}
&\{1234\},\quad \{1\}\{234\},\quad \{123\}\{4\},\quad \{12\}\{34\},\\
&\{1\}\{2\}\{34\},\quad \{1\}\{23\}\{4\},\quad \{12\}\{3\}\{4\},\quad \{1\}\{2\}\{3\}\{4\}.
\end{aligned}
\end{equation}
Partitions exhibiting crossing structures (such as $\{13\}\{24\}$) or nested configurations (such as $\{14\}\{23\}$) are excluded from $\mathcal I(n)$. Interval partitions enforce a strictly sequential order, making them the natural combinatorial structure for stochastic processes that remain constant between renewal events.

\subsection{Formal Definitions and M\"obius Inversion}

Given a moment sequence $\{m_n\}$, the Boolean cumulants $b_n$ are defined by restricting the moment expansion to interval partitions:
\begin{equation}
m_n = \sum_{\pi \in \mathcal I(n)} \prod_{B \in \pi} b_{|B|},
\label{eq:app_moment}
\end{equation}
which replaces the summation over the full partition lattice $\mathcal P(n)$ used for classical cumulants. Inversion of Eq.~\eqref{eq:app_moment} is obtained via M\"obius inversion on the poset $\mathcal I(n)$ \cite{Rota1964}:
\begin{equation}
b_n = \sum_{\pi \in \mathcal I(n)} \mu_{\mathcal I}(\pi, 1_n) \prod_{B \in \pi} m_{|B|}, \qquad \mu_{\mathcal I}(\pi, 1_n) = (-1)^{|\pi|-1}.
\label{eq:app_mobius}
\end{equation}
Because the interval $[\pi, 1_n]$ is isomorphic to the Boolean lattice of subsets on $|\pi|-1$ elements, the M\"obius function reduces to a pure sign factor $(-1)^{|\pi|-1}$, omitting the factorial or Catalan weights present in classical or free probability. Equation~\eqref{eq:app_mobius} corresponds to the inclusion--exclusion identity \eqref{k_nVSm_n} for $G$-cumulants \cite[\S4.4.3]{bbJSTAT4}.

Since the sum in Eq.~\eqref{eq:app_moment} runs over compositions rather than arbitrary set partitions, the associated generating functions are ordinary rather than exponential. Defining $M_o(z) = \sum_{n \ge 1} m_n z^n$ and $\eta(z) = \sum_{n \ge 1} b_n z^n$, summation over the number of blocks $p$ yields the geometric series
\begin{equation}
M_o(z) = \sum_{p \ge 1} \eta(z)^p = \frac{\eta(z)}{1 - \eta(z)} \quad \iff \quad \eta(z) = \frac{M_o(z)}{1 + M_o(z)}.
\label{eq:app_ogf}
\end{equation}
This algebraic relation leads to two key structural consequences:
\begin{itemize}
    \item \emph{Additivity under Boolean Independence:} When all mixed Boolean cumulants of two variables vanish, the generating function $\eta(z)$ is additive. In terms of moments, this condition implies
    \begin{equation}
    \frac{1}{1 + M_o^{(A \sqcup B)}} = \frac{1}{1 + M_o^{(A)}} + \frac{1}{1 + M_o^{(B)}} - 1,
    \end{equation}
    which defines Boolean convolution \cite{SpeicherWoroudi1997}.
    
    \item \emph{The Boolean Gaussian Distribution:} Truncating the cumulants at second order ($b_n = 0$ for $n \ge 3$) yields $\eta(z) = b_2 z^2$, leading to $M_o(z) = b_2 z^2 / (1 - b_2 z^2)$. This generating function corresponds to moments $m_{2k} = b_2^k$ and $m_{2k+1} = 0$, representing the symmetric Bernoulli distribution $\frac{1}{2}(\delta_{-a} + \delta_a)$ with $a^2 = b_2$. In Muraki's classification \cite{Muraki2003}, the Bernoulli distribution serves as the central limit law for the Boolean independence hierarchy, analogous to the Gaussian distribution in classical probability. This property forms the foundation of Theorem~\ref{thm:uniqueness}.
\end{itemize}

\subsection{Analytic Formulation and Heavy-Tailed Laws\label{sec:analytic_zeta}}

While Equations~\eqref{eq:app_moment}--\eqref{eq:app_ogf} apply to formal power series requiring the existence of all moments, the Boolean generator $\eta(w)$ can also be defined analytically, allowing the results of Section~\ref{sec:heavy} to extend to heavy-tailed jump distributions.

For a random variable $\xi$ with probability measure $p$, the identity $1 + M_o(w) = \sum_{n \ge 0} \langle \xi^n \rangle w^n = \langle (1 - \xi w)^{-1} \rangle$ allows rewriting Eq.~\eqref{eq:app_ogf} as
\begin{equation}
\eta(w) = 1 - \frac{1}{1 + M_o(w)} = 1 - \left\langle \frac{1}{1 - \xi w} \right\rangle^{-1} = 1 - w \, F(1/w),
\end{equation}
where $G(z) = \langle (z - \xi)^{-1} \rangle$ denotes the Cauchy transform of $p$, and $F(z) = 1 / G(z)$ is its reciprocal. For example, in the symmetric Bernoulli case $F(z) = (z^2 - a^2)/z$, recovering $\eta(w) = a^2 w^2$.

Because the expectation $\langle (1 - \xi w)^{-1} \rangle$ converges for any probability measure, $\eta(w)$ remains well-defined even when moment expansions diverge. A notable consequence is that every probability measure is infinitely divisible under Boolean convolution \cite{SpeicherWoroudi1997}.

\section{Derivation of the General Order Formula \eqref{eq:renewal_boolean_form}\label{app:cancellation}}

To express Equation~\eqref{eq:renewal_Cn} in terms of a Boolean moment--cumulant structure, we exploit the fact that its support is restricted to $\mathcal I(n)$. Although the weighting factors $(1 - x_i)$ depend on pairs of adjacent blocks across cuts, the transformation to a strictly multiplicative form over blocks can be established in three steps.

\paragraph{(i) Refinement expansion.}
Inserting the Boolean moment--cumulant expansion for the jump variable, $\overline{\xi^k} = \sum_{\pi \in \mathcal I(k)} \prod_{j} b_{|B_j|}$, into Eq.~\eqref{eq:renewal_Cn} expands each block of $\pi_S$ into sub-blocks. The product $\prod_{B \in \pi_S} \overline{\xi^{|B|}}$ is thereby expressed as a summation over all compositions $\pi$ that refine $\pi_S$, or equivalently over all $\pi$ whose cut set $K = \mathrm{cut}(\pi)$ contains $S$:
\begin{equation}
\prod_{B \in \pi_S} \overline{\xi^{|B|}} = \sum_{\pi:\, K \supseteq S} \prod_{B \in \pi} b_{|B|}.
\end{equation}

\paragraph{(ii) Summation exchange.}
Exchanging the order of summation over the subset $S$ and the refining partition $\pi$ yields
\begin{equation}
C_n(u_1,\dots,u_n) = \sum_{\pi \in \mathcal I(n)} \left[ \prod_{B \in \pi} b_{|B|} \right] \sum_{S \subseteq K} \prod_{i \notin S} x_i \prod_{i \in S} (1 - x_i).
\label{eq:C_temp}
\end{equation}

\paragraph{(iii) Factorization and cancellation.}
The internal sum over $S \subseteq K$ separates the $n-1$ gaps into internal block gaps $\mathrm{int}(\pi)$ and boundary cut gaps $K$. Internal gaps $i \in \mathrm{int}(\pi)$ are excluded from $S$ by construction and contribute a factor $x_i$. For cut gaps $i \in K$, summing over inclusion or exclusion in $S$ simplifies to $x_i + (1 - x_i) = 1$. The sum therefore factorizes as
\begin{equation}
\sum_{S \subseteq K} \prod_{i \notin S} x_i \prod_{i \in S} (1 - x_i) = \left[ \prod_{i \in \mathrm{int}(\pi)} x_i \right] \prod_{i \in K} \left[ x_i + (1 - x_i) \right] = \prod_{i \in \mathrm{int}(\pi)} x_i.
\label{eq:cut_cancellation}
\end{equation}
Thus, boundary factors cancel identically, leaving only intra-block survival probabilities. For exponentially distributed waiting times, the product over the interior gaps of a block $B$ telescopes:
\begin{equation}
\prod_{i \in \mathrm{int}(B)} x_i = e^{-\sum_{i} d_i / \tau} = \phi(\mathrm{span}(B)).
\end{equation}
Substituting Eq.~\eqref{eq:cut_cancellation} back into Eq.~\eqref{eq:C_temp} completes the derivation of Eq.~\eqref{eq:renewal_boolean_form}.

 \newcommand{\noop}[1]{}

\end{document}